\documentclass[twoside,11pt]{article}

\usepackage[preprint]{jmlr2e}

\usepackage{etoolbox}
\usepackage[ansinew, utf8]{inputenc}
\usepackage{mathtools, amssymb, bbm}
\usepackage{multicol,latexsym, array, xcolor}
\usepackage{pifont}
\usepackage{comment}
\usepackage{bm}
\usepackage{dsfont}
\usepackage{ mathrsfs }
\usepackage{enumitem}

\hypersetup{hidelinks}

\newtheorem{cond}[theorem]{Condition}
\newtheorem{set}[theorem]{Setting}

\allowdisplaybreaks[1]

\newcounter{claimcount}
\newenvironment{claim}{\refstepcounter{claimcount}\emph{Claim \arabic{claimcount}:}}{\vspace{1pt}}
\AtBeginEnvironment{proof}{\setcounter{claimcount}{0}}

\newcommand*{\QEDA}{\null\nobreak\hfill\ensuremath{\blacksquare}}%
\newcommand{\R}{\mathbb{R}}
\newcommand{\Rd}{\mathbb{R}^d}
\newcommand{\norm}[1]{\left\lVert#1\right\rVert}
\newcommand{\iid}{\overset{i.i.d.}{\sim}}
\newcommand{\RR}{\mathbb{R}}

\newcommand{\Prob}{\mathbb{P}}

\newcommand{\supp}{\mathrm{supp}}

\newcommand{\Haus}[1]{\mathscr{H}^{#1} }
\newcommand{\Bl}{[}
\newcommand{\Br}{]}
\newcommand{\E}[1]{\ensuremath{\mathbb{E}\left[#1\right]}}
\newcommand{\Esplit}[2]{\ensuremath{\mathbb{E}_{#1}\left[#2\right]}}
\newcommand{\Var}[1]{\ensuremath{\mathrm{Var}\left(#1\right)}}

\newcommand{\diam}[1]{\ensuremath{\mathrm{diam}\left(#1\right)}}
\newcommand{\arccot}[1]{\ensuremath{\mathrm{arccot}\left(#1\right)}}

\newcommand{\indifunc}[1]{\mathds{1}_{\left\lbrace#1 \right\rbrace }}
\newcommand{\X}{\mathcal{X}}
\newcommand{\muX}{\mu_\mathcal{X}}
\newcommand{\empmuX}{\widehat{\mu}_\mathcal{X}}

\newcommand{\mmspaceX}{\left(\X,||\cdot||,\muX \right) }

\newcommand{\dtm}{\mathrm{d}^2_{\X,m}}
\newcommand{\dtmvar}[1]{\mathrm{d}^2_{\X,#1}}
\newcommand{\empdtm}{{\delta}^2_{\X,m}}
\newcommand{\dtmone}{\mathrm{d}^2_{\X,1}}
\newcommand{\empdtmone}{{\delta}^2_{\X,1}}

\newcommand{\dfct}{\mathbbm{d}}

\newcommand{\fdtm}{f_{\dtm}}

\newcommand{\fdtmone}{f_{\dtmone}}

\newcommand{\kde}{\widehat{f}_{\empdtm}}
\newcommand{\kdeone}{\widehat{f}_{\empdtmone}}
\newcommand{\semikde}{\widehat{f}_{\dtm}}
\newcommand{\semikdeone}{\widehat{f}_{\dtmone}}

\newcommand{\fdtmvar}[1]{f_{\mathrm{d}^2_{\X,#1}}}
\newcommand{\Fdtmvar}[1]{F_{\mathrm{d}^2_{\X,#1}}}
\newcommand{\kdeyvar}[2]{\widehat{f}_{{\delta}^2_{{\mathcal{Y}_#1},#2}}}
\newcommand{\kdexvar}[2]{\widehat{f}_{{\delta}^2_{{\mathcal{X}_{#1}},#2}}}

\usepackage{lastpage}
\jmlrheading{25}{2024}{1-\pageref{LastPage}}{11/22}{7/24}{22-1344}{Katharina Proksch, Christoph Alexander Weikamp, Thomas Staudt, Beno\^{\i}t Lelandais, Christophe Zimmer }

\ShortHeadings{Dtm Density based Geometric Analysis of Complex Data}{Proksch, Weitkamp, Staudt, Lelandais and Zimmer}
\firstpageno{1}

\begin{document}
	
	\title{From Small Scales to Large Scales: Distance-to-Measure Density based Geometric Analysis of Complex Data}
	\author{\name Katharina Proksch\thanks{ Katharina Proksch and Christoph Weitkamp contributed equally to this work.}
		\email k.proksch@utwente.nl \\
		\addr Faculty of Electrical Engineering, Mathematics \& Computer Science\\ 
		University of Twente\\
		Hallenweg 19, 7522NH Enschede, Netherlands
		\AND
		\name Christoph A.\ Weitkamp\footnotemark[1]
		\email cweitka@mathematik.uni-goettingen.de \\
		\addr Institute for Mathematical Stochastics\\
		University of G{\"o}ttingen\\
		Goldschmidtstra\ss e 7,
		37077 G{\"o}ttingen, Germany
		\AND 
		\name Thomas Staudt \email thomas.staudt@uni-goettingen.de\\
		\addr Institute for Mathematical Stochastics\\
		University of G{\"o}ttingen\\
		Goldschmidtstra\ss e 7,
		37077 G{\"o}ttingen, Germany
		\AND
		\name Beno\^{\i}t Lelandais \email benoit.lelandais@pasteur.fr\\
		\addr Institut Pasteur\\ Imaging and
		Modeling Unit\\ UMR 3691, CNRS,
		Paris, France
		\AND
		\name Christophe Zimmer \email czimmer@pasteur.fr\\
		\addr Institut Pasteur\\ Imaging and
		Modeling Unit\\ UMR 3691, CNRS,
		Paris, France
	}

	\editor{Florence d'Alche-Buc}
	
	\maketitle
	
	\begin{abstract}
		How can we tell complex point clouds with different small scale characteristics apart, while disregarding global features? Can we find a suitable transformation of such data in a way that allows to discriminate between differences in this sense with statistical guarantees?
		
		In this paper, we consider the analysis and classification of complex point clouds as they are obtained, e.g., via single molecule localization microscopy.
		We focus on the task of identifying differences between noisy point clouds based on small scale characteristics, while disregarding large scale information such as overall size.
		We propose an approach based on a transformation of the data via the so-called Distance-to-Measure (DTM) function, a transformation which is based on the average of nearest neighbor distances. For each data set, we estimate the probability density of average local distances of all data points and use the estimated densities for classification.
		While the applicability is immediate and the practical performance of the proposed methodology is very good, the theoretical study of the density estimators is quite challenging, as they are based on \textit{non-i.i.d.\ }observations that have been obtained via a complicated transformation. In fact, the transformed data are stochastically dependent in a non-local way that is not captured by commonly considered dependence measures. Nonetheless, we show that the asymptotic behaviour of the density estimator is driven by a kernel density estimator of certain i.i.d. random variables by using theoretical properties of $U$-statistics, which allows to handle the dependencies via a Hoeffding decomposition. 
		We show via a numerical study and in an application to simulated single molecule localization microscopy data of chromatin fibers that unsupervised classification tasks based on estimated DTM-densities achieve excellent separation results.
	\end{abstract}
	
	
	\section{Introduction}\label{sec:introduction}
	The analysis and extraction of information from complex point clouds has become a main task in many applications. Prominent examples can be found in geomorphology, where structure in point-clouds obtained from laser scanners is investigated to infer the shape of the Earth \citep{Vosselman,Yuichi}, or in cosmology, where the Cosmic Web is analysed based on a discrete set of points from $N$-body simulations or galaxy studies \citep{Libeskind2}. Related questions also arise in biology, when data from single molecule localization microscopy (SMLM), which is based on the localization of fluorescent molecules that appear at different times, are analyzed \citep{SMLM1,SMLM2}. 
	Data obtained in SMLM are 2D or 3D point clouds, where the points correspond to particular molecular localization events. 
	In this paper, we consider a specific example which is related to the analysis of super-resolution visualization of human chromosomal regions as it has recently been investigated in \citet{Paris}. In this application, the goal is to better understand the 3D organization of the chromatin fiber in cell nuclei, which plays a key role in the regulation of gene expression. \\
	In all aforementioned examples, it is important to identify significant differences between noisy point clouds, where a focus is on general structure and small scale information rather than on global features such as the overall shape of a point cloud. 
	
	\begin{figure}[ht]
		\centering
		\includegraphics[width=0.8\textwidth]{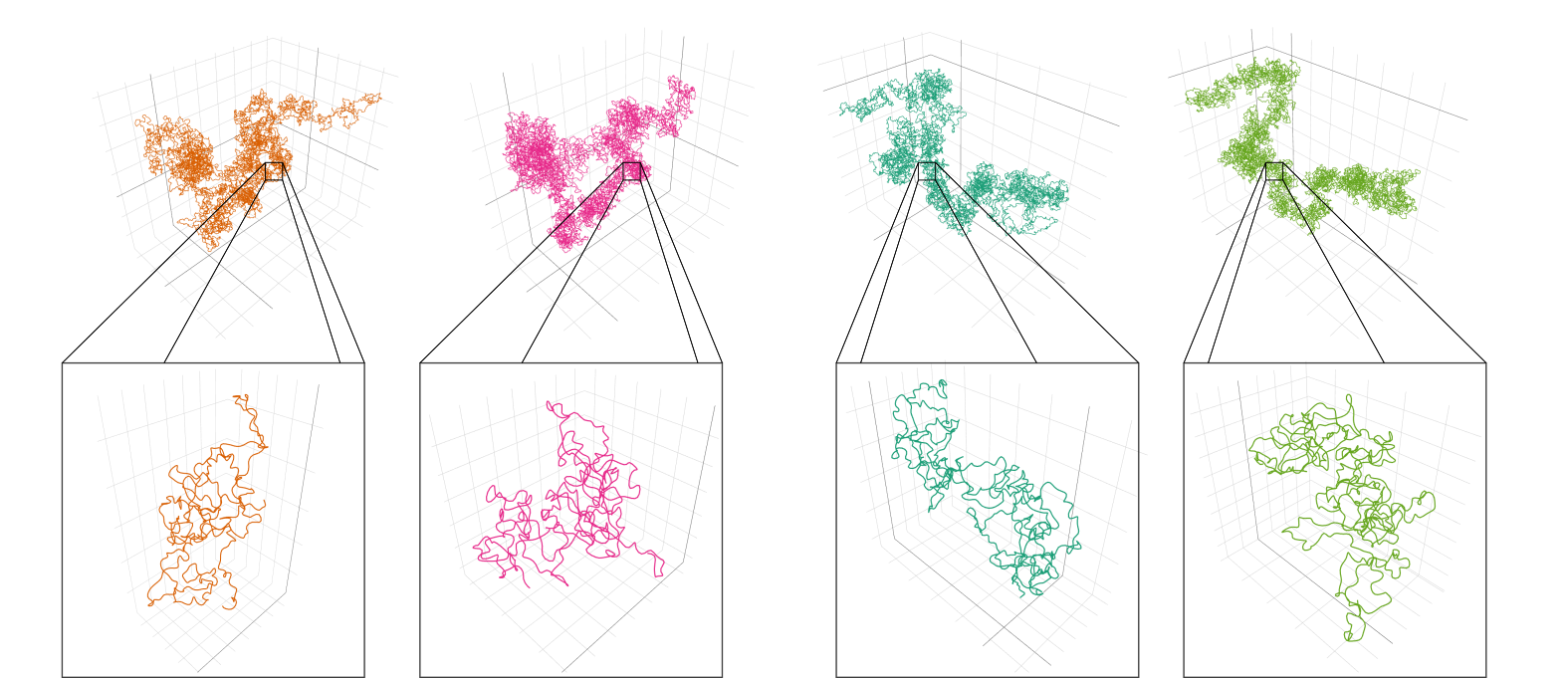}
		\caption{\textbf{Example Data:} Four different simulated chromatin fibers in two different conditions: Condition A (orange (far left) and blue-green (middle right)) and Condition B (pink (middle left) and green (far right)) for the purpose of comparison. }
		\label{fig:data_example_intro}
	\end{figure}
	For illustration, Figure \ref{fig:data_example_intro} shows four simulated chromatin fibers in two different conditions. The displayed structures form loops of different sizes and frequencies, based on the condition under which they were simulated. 
	The differences between the conditions are so subtle that they are not clearly distinguishable by eye. 
	In the application considered in this paper, we analyse noisy samples of such simulated structures. The noise accounts for localization errors as they are present in real SMLM data. 
	The loops are of sizes comparable to the resolution of the images (see Section \ref{sec:CLA} and \citet{Paris} for more details), which makes the problem tractable but difficult. The aim is to classify the point clouds based on their loop distribution (i.e., based on their small scale characteristics), while disregarding their total size or other large scale features. It is natural to transform such complicated data prior to the analysis, in particular when one has a clear objective in mind. In the above reference, the statistical analysis of the simulated and real data was based on a transformation of each data cloud onto a set of two parameters, capturing smoothness and local curvature of the point clouds. While this transformation provided a clear discrimination between different groups, the amount of information preserved in a two-dimensional parameter is not sufficient as a basis for point-by-point classification.
	In this paper, we propose an approach which is similar in spirit, but which provides a transformation into a curve, with different characteristics for the different conditions. In our analysis, the whole curves are then used as features. To this end, we perform the following two steps.
	\begin{itemize}
		\item[(i)] A transformation of the point cloud based on the \textit{Distance-to-Measure (DTM) signature} \citep{chazal2016rates,chazal2017robust} to a one-dimensional data set.
		\item[(ii)] The analysis of the distribution of the DTM-transformed data via their estimated probability density.
	\end{itemize}
	The DTM signature is closely related to certain nearest neighbor distributions, which makes this approach very intuitive. In particular, this framework allows for a comprehensive exploratory analysis of complex data, for which we might seek a simple graphical representation that captures and summarizes the local structural information well. 
	
	\subsection{The DTM-Density as a Representation for Local Features}\label{sec:proposed approach} 
	We now introduce the statistical framework of the paper and carefully define the previously mentioned DTM-signature. Throughout the following, we consider random point clouds as samples from a \textit{Euclidean metric measure space} $\X=\mmspaceX$, i.e., a triple, where $\X\subset\RR^d$ denotes a compact set, $||\cdot||$ stands for the Euclidean distance and $\muX$ denotes a probability measure that is fully supported on the compact set $\X$. If, additionally, $\muX$ has a Lipschitz continuous density with respect to the $d$-dimensional Lebesgue measure, then we call $\X$ a \textit{regular Euclidean metric measure space}. For a metric measure space $\X$, we  define the corresponding \emph{Distance-to-Measure (DTM) function} with \emph{mass parameter} $m\in(0,1]$ for $x\in\Rd$ as 
	\begin{equation}\label{eq:dtm-function}
		\dtm(x)=\frac{1}{m}\int_0^m\!F_x^{-1}(u)\,du,
		\end{equation}
	where $F_x(t)=P(\norm{X-x}^2 \leq t)$, $X\sim \muX$, and $F_x^{-1}$ denotes the corresponding quantile function. The DTM-function, which is essential for the definition of the DTM-signature, is a population quantity that is generally unknown in practice and thus has to be estimated from the data. In order to do so, we replace the quantile function in definition \eqref{eq:dtm-function} by its empirical version as follows. Let $X_1\dots,X_n\iid\muX$ and denote the corresponding empirical measure by $\empmuX$. We define for $t\geq 0$
	\begin{equation}\label{eq:defofFxn}
		\widehat{F}_{x,n}(t)=\frac{1}{n}\sum_{i=1}^n\indifunc{||x-X_i||^2\leq t}
	\end{equation}
	and denote by $\widehat{F}_{x,n}^{-1}$ the corresponding quantile function, giving rise to a plug-in estimator $\empdtm(x)$ for the Distance-to-Measure function $\dtm(x)$:
	\begin{equation}\label{eq:emp dtm}
		\empdtm(x)=\frac{1}{m}\int_0^m\!\widehat{F}_{x,n}^{-1}(u)\,du.
	\end{equation}
	In the special case that $m=\frac{k}{n}$, it is possible to rewrite \eqref{eq:emp dtm} as a \textit{nearest neighbor statistic} as follows
	\begin{equation}\label{eq:emp dtm v2}
		\empdtm(x)=\frac{1}{k}\sum_{X_i\in N_k(x)}||X_i-x||^2, 
	\end{equation}
	where $N_k(x)$ is the set containing the $k$ nearest neighbors of $x$ among the data points $X_1,\dots,X_n$.\\
	
	As discussed previously, we require a good descriptor for the small scale behavior of our data. Hence, in a similar spirit as \citet{Brecheteau}, we reduce the potentially complex Euclidean metric measure space to a one-dimensional probability distribution by considering the \emph{Distance-to-Measure (DTM) signature} $\dtm(X)$, where $X\sim\muX$. That is, the deterministic point $x\in\X$ is replaced by the random variable $X$ representing our observations. The distribution of $\dtm(X)$ captures the relative frequency of the mean of the distances of a random point in $\X$ to its ``$m\cdot100\%$ nearest neighbors''. We will empirically illustrate that the distribution of $\dtm(X)$ is a good descriptor for the small scale behavior of the considered data for small values of $m$ and verify that it is well-suited for chromatin loop analysis. Furthermore, it is easy to see that for $m=1$ the random quantity $\dtmone(X)$ is closely related to the lower bound $\text{FLB}_p$ of the Gromov-Wasserstein distance defined in \citet{memoli2011gromov}:
	\begin{align*}
		\text{FLB}_p=\frac{1}{2}\inf_{\mu\in\mathcal{M}(\muX,\mu_{\mathcal{Y}})}
		\left(\int_{\mathcal{X}\times\mathcal{Y}}|s_{\mathcal{X},p}(x)-s_{\mathcal{Y},p}(y)|^p\mu(dx,dy)\right)^{\frac{1}{p}},
	\end{align*}
	where $s_{\mathcal{X},p}=\|d_{\mathcal{X}}(x,\cdot)\|_{L^p(\mu_{\mathcal{X}})}$ and $\mathcal{M}(\muX,\mu_{\mathcal{Y}})$ denotes the set of all couplings of $\muX$ and $\mu_{\mathcal{Y}}$.
	This motivates why $\dtmone(X)$ is well suited for object discrimination with a focus on large scale characteristics. Although this case is not of interest in our specific data example, we include it in our analysis, since variants of $\dtmone(X)$ have been proven very useful for pose invariant object discrimination when large scale differences are crucial \citep{hamza2003geodesic,gellert2019substrate}, see also Section \ref{subsec:disc props} for an example.\\
	
	Since we propose to reduce (possibly complex) multi-dimensional metric measure spaces to a one-dimensional probability distribution, the next step is to visualize and investigate these distributions. It is well known that probability densities (if they exist) can provide a useful visual insight into the probability distributions considered. In this regard, they are usually better suited than cumulative distribution functions (see, e.g., \citet{chen2018modern}). Therefore, we focus on the estimation of the density of $\dtm(X)$ in this paper. A natural estimator for the density of $\dtm(X)$, in the following denoted as \textit{DTM-density}, in case of a known DTM-function, is given by
	\begin{equation}\label{eq:def of semikde}
		\semikde(y)=\frac{1}{nh}\sum_{i=1}^n K\left(\frac{\dtm(X_i)-y}{h}\right).
	\end{equation}
	However, since $\muX$ is  unknown, we cannot compute $\dtm$  in practice and consequently it is generally not feasible to estimate $\fdtm$ via $\semikde$. Instead, we propose to replace $\dtm$ by its empirical version $\empdtm $ and estimate $\fdtm$ based on the plug-in estimator 
	\begin{equation}\label{eq:def of kde}
		\kde(y)=\frac{1}{nh}\sum_{i=1}^n K\left(\frac{\empdtm(X_i)-y}{h}\right).
	\end{equation}
	
	It is important to note that, in contrast to $\semikde$, the plug-in estimator $\kde$ is based on the \textit{non-i.i.d.\ }observations $\empdtm(X_1),\ldots,\empdtm(X_n)$. In fact, for each $i\neq j$, $\empdtm(X_i)$ and $\empdtm(X_j)$ are stochastically dependent.
	The asymptotic behaviour of kernel density estimators under dependence has been studied extensively in the literature for various mixing and linear processes connected to weakly dependent time series \citep{Castellana,Robinson,liebscher1996strong,Lu,Mielniczuk}. In all these settings, results on asymptotic normality similar to the i.i.d.\ case can be derived.
	Related results for spatial processes can be found, e.g., in \citet{Hallin2004}.
	For \textit{long-range dependent} data, the asymptotic behaviour of kernel density estimators changes drastically. Here, the empirical density process (based on kernel estimators of the marginal densities) converges weakly to a tight limit (see \citet{csorgo1995density}). For the sequence $\empdtm(X_1),\ldots,\empdtm(X_n)$, however, a structure as in the above examples (in space or time) is not given. For each $i\neq j$, $\empdtm(X_i)$ and $\empdtm(X_j)$ are stochastically dependent
	in a way that is not captured by the dependency models considered in the literature discussed above. Curiously perhaps, $\kde$ allows a decomposition into a non-degenerate $U$-statistic $U_n(X_1,\ldots,X_n)$ of order 2 and negligible remainder terms (Step 2 in the proof of Theorem \ref{thm:kde limit}, see Lemma \ref{lemma:ustat reform} in Section \ref{sec:technical lemmas}). $U$-statistics are averages of symmetric functions \textit{applied to all tuples} of fixed size (here 2) in the sample and  are a concept from classical mathematical statistics due to \cite{Hoeffding}. By proper projection we can find an \textit{i.i.d.} approximation to $U_n(X_1,\ldots,X_n)$ which drives the asymptotics and ensures a Gaussian limit (Step 3 in the proof of Theorem \ref{thm:kde limit}, see Lemma \ref{lemma:Hoeffding decomposition} in Section \ref{sec:technical lemmas}, see also Chapter 12.1 in \cite{van2000asymptotic} for mathematical properties of $U$-statistics in general).
	
	\subsection{Main Results}
	The main theoretical contribution of the paper is the distributional limit of the kernel density estimator defined in \eqref{eq:def of kde}. More precisely, we prove (cf. Theorem \ref{thm:kde limit}), given certain regularity conditions on $\fdtm$, $\dtm(y)$ and $\X$, (see Condition \ref{condition} in Section \ref{sec:setting and assumptions}) that for 
	$n\to\infty$, $h=o\left(n^{-1/5}\right)$ and $nh\to \infty$ 
	\begin{equation}\sqrt{nh}\left(\kde(y)-\fdtm(y)\right)\Rightarrow N\left(0,\fdtm(y)\int\!K^2(u)\,du\right)\label{eq:pointwise limit}.\end{equation}
	This means that, although the kernel density estimator $\kde$ is based on transformed, dependent random variables, asymptotically, it behaves 
	precisely as the inaccessible kernel density estimator $\semikde$ based on independent random variables. This entails that many methods which are feasible for kernel density estimators based on i.i.d.\ data, can be applied in this much more complex setting as well, with the same asymptotic justification.
	
	\subsection{Application}\label{subsec:application}
	Chromosomes, which consist of chromatin fibres, are essential parts of cell nuclei in human beings and carry the genetic information important for heredity transmission. It is known by now that there are small scale self-interacting genomic regions, so called
	topologically
	associating domains (TADs) which
	are often associated with loops in the chromatin fibers \citep{NueblerE6697}. As an application, we consider
	chromatin loop analysis, one aspect of which is to study the presence or absence of loops in the chromatin (see Section \ref{sec:CLA}). \\
	The local loop structure is very well characterized by local nearest neighbor means as illustrated on the right of Figure \ref{fig:pipeline_and_results} and hence we propose to use DTM-signatures for tackling this issue. Figure \ref{fig:pipeline_and_results} shows the pipeline for the data transformation (left) and the resulting kernel density estimators ($m=1/250$, biweight kernel, bandwidth selection as in Section \ref{sec:CLA}) for the four data sets shown in Figure \ref{fig:data_example_intro} (right of Figure \ref{fig:pipeline_and_results}, same coloring). It shows that the kernel density estimators mainly differ between the different conditions and not between the corresponding chromatin fibers and that the differences between the conditions are clearly pronounced. This demonstrates  that the transformation is well suited for a qualitative analysis of the data.
	\begin{figure}[ht]
		\centering
		\includegraphics[width=\textwidth]{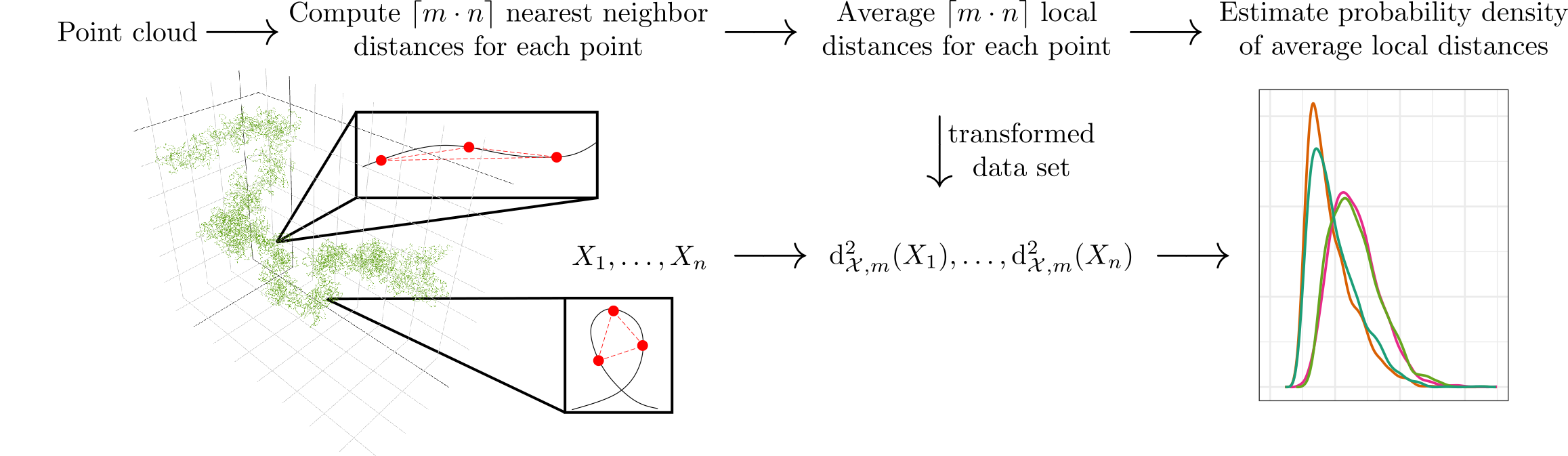}
		\caption{\textbf{Data analysis pipeline:} Illustration of the different steps in the proposed data analysis. The red dots in the details of the image represent data points, the red lines show the point-to-point distances, whereas the underlying chromatin structure is depicted by a black line.   Right: The resulting DTM-density estimates of the point clouds illustrated in Figure \ref{fig:data_example_intro} (same coloring).}
		\label{fig:pipeline_and_results}
	\end{figure}
	\subsection{Related Work}
	The use of the DTM-signature for the purpose of pose invariant object discrimination was proposed by \citet{Brecheteau}, who in particular established a relation between the DTM-signature and the Gromov-Wasserstein distance (see \citet{memoli2011gromov} for a definition). In the aforementioned work, the author considers the asymptotic behavior of the Wasserstein distance between sub-sampled estimates of the DTM-signatures for two different spaces. While the approaches are related, one big advantage of our method of
	estimating the DTM-densities  is that it does not require sub-sampling and all data points can be used for the analysis, guaranteeing a powerful procedure without loss of information.\\
	As illustrated in Section \ref{sec:proposed approach}, the DTM-signature is based on the DTM-function (see Eq.\ \ref{eq:dtm-function}). This function has been thoroughly studied and applied in the context of support estimation and topological data analysis \citep{chazal2011geometric,chazal2013persistence,buchet2014topological} and for its sample counterpart (see Eq.\ \ref{eq:emp dtm}) many consistency properties have been established in \citet{chazal2016rates,chazal2017robust}. \\
	
	Distance based signatures for object discrimination have been applied and studied in a variety of settings \citep{osada2002shape,gelfand2005robust,belongie2006matching,shi2007direct,brinkman2012invariant,berrendero2016shape}. Recently, lower bounds of the Gromov-Wasserstein distance (see \citet{memoli2011gromov}) have received some attention in applications \citep{gellert2019substrate} and in the investigation of their discriminating properties and their statistical behavior \citep{memoli2018distance, weitkamp2020gromovwasserstein}.

	Furthermore, it is noteworthy that nearest neighbor distributions are of great interest in various fields in biology \citep{zou1995nearest,meng2020k} as well as in physics \citep{torquato1990nearest,bhattacharjee2003nth,hsiao2020mean}. In these fields it is quite common to consider the (mean of the) distribution of all nearest neighbors for data analysis. This case corresponds to $m=1/n$ in our method (see Remark \ref{rem:sampling} for a discussion of different types of sampling mechanisms). We would like to emphasize that taking the mean over a certain percentage of nearest neighbors makes our method a lot more robust against noise, leading to a more reliable performance in the analysis of noisy point clouds.
	\\ In the analysis of SMLM images, methods from spatial statistics are often employed. Related to the global distribution of all distances is Ripley's K, which is used to infer the amount and the degree of clustering in a given data set as compared to a point cloud generated by a homogeneous Poisson point process (see, e.g., \citet{SMLM1} for the application of Ripley's K in this context). Despite the connection via certain distributions of distances, the objectives and underlying models are quite different to the setting of this paper, such that a direct comparison is not sensible.
	\\
	
	Kernel density estimation from dependent data is a broad and well investigated topic. 
	In addition to the references provided in Section \ref{sec:proposed approach}, kernel density estimators of symmetric functions of the data and dyadic undirected data have been considered \citep{frees1994estimating,graham2019kernel}. In these settings, the summands of the corresponding kernel density estimators admit a ``$U$-statistic like'' dependency structure that has to be accounted for. While this is more closely related to the dependency structure which we are encountering in our analysis, the structure of the statistics that appear in the decomposition of the kernel density estimator \eqref{eq:def of kde} is quite different, such that those results cannot directly be transferred to our setting.
	\subsection{Organization of the Paper}
	In Section \ref{sec:Distributional} we state the main results and are concerned with the derivation of distributional limit \eqref{eq:pointwise limit} 
	and the assumptions required. Afterwards, in Section \ref{sec:simulations} we illustrate our findings via simulations.
	In Section \ref{sec:CLA}, we apply our methodology to the classification within the framework of chromatin loop analysis.
	
	\subsection{Notation:} Throughout the following, we denote the $d$-dimensional Lebesgue measure by $\lambda^d$ and the $(d-1)$-dimensional surface measure in $\Rd$ by $\sigma^{d-1}$. We write $B(x,r)$ for the open ball in $\Rd$ (equipped with $||\cdot||$) with center $x$ and radius $r$. Given a function $f$ or a measure $\mu$, we write $\supp(f)$ and $\supp(\mu)$ to denote their respective support. Let $F$ be a distribution function with compact support $[a,b]$ and let $F^{-1}$ denote the corresponding quantile function. As frequently done, we set $F^{-1}(0)=a$ and $F^{-1}(1)=b$. Let $U\subseteq\R^{d_1}$ be an open set. We denote by $C^k(U,\R^{d_2})$ the set of all $k$-times continuously differentiable functions from $U$ to $\R^{d_2}$. Further, we denote by $C^{k,1}(U,\R^{d_2})$ the set of all $k$-times continuously differentiable functions from $U$ to $\R^{d_2}$, whose $k$th derivative is Lipschitz continuous. For $d_2=1$, we abbreviate this to $C^k(U)$ and $C^{k,1}(U)$. If the domain and range of a function $g$ are clear from the context, we will usually write $g\in C^{k}$ or $g\in C^{k,1}$.
	\section{Distributional Limits}\label{sec:Distributional}
	In this section, we recall the setting, establish the conditions required for our asymptotic theory and ensure that these are met in some simple examples before we state our main theoretical results, upon which our statistical methodology is based. To this end, we show that $\kde$ is a reasonable estimator for the density of the DTM-signature by proving the distributional limit \eqref{eq:pointwise limit}.

	\subsection{Setting and Assumptions}\label{sec:setting and assumptions}
	First of all, we summarize the setting introduced in Section \ref{sec:proposed approach}. 
	\begin{set}\label{setting}
		Let $\mmspaceX$ denote a regular Euclidean metric measure space. For $x\in\X$ let $\dtm(x)$ denote the corresponding Distance-to-Measure function with mass parameter $m\in(0,1]$. Let $X\sim\muX$ and assume that the Distance-to-Measure signature $\dtm(X)$ admits a density $\fdtm$. Let $X_1,\dots,X_n\iid\muX$ and denote by $\semikde$ and $\kde$ the kernel density estimators defined in \eqref{eq:def of semikde} and \eqref{eq:def of kde}, respectively. 
	\end{set}
	It is noteworthy that the assumption that $\dtm(X)$ admits a Lebesgue density is slightly restrictive. The probability measure $\mu_{\dtm}$ of the DTM-signature can have a pure point component $\mu_{\dtm,\text{pp}}$ in addition to the continuous component $\mu_{\dtm,\text{cont}}$, if the spaces considered have very little local structure (for examples, see Section \ref{sec:examples}). That is,
	\begin{align*}
		\mu_{\dtm}= \mu_{\dtm,\text{pp}}+\mu_{\dtm,\text{cont}}.
	\end{align*}
	If we define $\fdtm$ to be the Radon-Nikodym derivative of the absolutely continuous component $\mu_{\dtm,\text{cont}}$, i.e., $\fdtm\mathrm{d}\lambda=\mathrm{d}\mu_{\dtm,\text{cont}}$, the pointwise asymptotic analysis of $\kde$ performed in Section \ref{subsec:asym theo} (see Theorem \ref{thm:kde limit}) remains valid for all $y$ with $\mu_{\dtm}(\{y\})=0$ that meet the corresponding assumptions. This guarantees that our analysis remains meaningful even if parts of our space do not provide local structure that is discriminative. \\
	In order to derive statement \eqref{eq:pointwise limit}, we require certain regularity assumptions on the density $\fdtm$, the DTM-function $\dtm$ and the kernel $K$. For the sake of completeness, we first recall some facts about the relation of the level sets of a given function. Let $g:\Rd\to\R$ and let $y\in\R$ be such that $g^{-1}(\{y\})\neq \emptyset$.  Suppose that the function $g$ is continuously differentiable in an open neighborhood of $g^{-1}(\{y\})$.  Assume further that $\nabla g\neq 0$ on the level set $g^{-1}(\{y\})$. Then, it follows by Cauchy-Lipschitz's theory (see e.g., \cite{hirsch1974differential,amann2011ordinary}) that there exists a constant $h_0>0$, an open set $W\supset g^{-1}([y-h_0,y+h_0])$ and a canonical one parameter family of $C^1$-diffeomorphisms $\Phi:[-h_0,h_0]\times W\to \Rd$ with the following property:
	\[\Phi(v, g^{-1}(\{y\}))=g^{-1}(\{y+v\})\]
	for all $v\in [-h_0,h_0]$ (for the precise construction of $\Phi$ see the proof of Lemma \ref{lem:local Lipschitz} in Section \ref{sec:lemmasB1toB4}). Throughout the following, the family $\{\Phi(v,\cdot)\}_{v\in[-h_0,h_0]}$ (also abbreviated to $\Phi$) is referred to as \textit{canonical level set flow of $g^{-1}(\{y\})$}.
	
	\begin{cond}\label{condition}
		Let $\fdtm$ be supported on $[D_1,D_2]$ and let $y\in[D_1,D_2]$. Assume that there exists $\epsilon>0$ such that $\fdtm$ is twice continuously differentiable on $(y-\epsilon,y+\epsilon)$. Further, suppose that the function $\dtm:\Rd\to\R$ is $C^{2,1}$ on an open neighborhood of the level set \[\Gamma_y\coloneqq {\dtm}^{-1}(\{y\})=\{x\in\Rd:\dtm(x)=y\},\] that $\nabla \dtm \neq 0$ on $\Gamma_y$ and that there exists $h_0>0$ such that for all $-h_0<v<h_0$
		\begin{equation}\label{eq:set level set cond}
			\mathcal{I}_{\mathcal{X}}(y;v):= \int_{\Gamma_y}\!\left|\indifunc{x\in \X}-\indifunc{\Phi(v,x)\in \X}\right|\,d\sigma^{d-1}(x)\leq C_y |v|,
		\end{equation}
		where $\{\Phi(v,\cdot)\}_{v\in[-h_0,h_0]}$ denotes the canonical level set flow of $\Gamma_y$ and $C_y$ denotes a finite constant that depends on $y$ and $\dtm$. Suppose that the kernel $K:\R\to \R_+$, is an even, twice continuously differentiable function with $\supp(K)=[-1,1]$. If $m<1$, we assume additionally that there are constants $\kappa>0$ and $1\leq b<5$ such that for $u\in(0,1)$ it holds \begin{equation} \omega_\X(u)\coloneqq\sup_{x\in\X}\sup_{t,t'\in(0,1)^2,|t-t'|<u}\left|F_x^{-1}(t)-F_x^{-1}(t')\right| \leq \kappa u^{1/b}.\label{eq:cond mod of cont}\end{equation}
	\end{cond}
	The satisfiability of Condition \ref{condition} is an important issue that is difficult to address in general. Hence, in Section \ref{sec:examples} we will verify that the requirements of Condition \ref{condition} are met in several simple examples. Nevertheless, in order to put Condition \ref{condition} into a broader perspective, we first gather some known regularity properties of $\dtm$ as well as $\{F_x^{-1}\}_{x\in\X}$ and discuss the technical requirement \eqref{eq:set level set cond} afterwards.\\
	
	\subsubsection{Regularity of \texorpdfstring{$\dtm$ and $\{F_x^{-1}\}_{x\in\X}$}{the Distance-to-Measure-Function}}
	We distinguish between the cases $m<1$ and $m=1$ for the presentation of known regularity results. For $m<1$, the smoothness of $\dtm$ has been investigated in \citet{chazal2011geometric}, where the authors derived the following results.
	\begin{lemma}~
		\vspace{-0.25cm}
		\begin{enumerate}
			\item Let $\mmspaceX$ denote an Euclidean metric measure space. Then, the function $\dtm:\R^d\to \R$ is almost everywhere twice differentiable.
			\item If $\X=\mmspaceX$ denotes a regular Euclidean metric measure space, then the function $\dtm:\R^d\to \R$ is differentiable with derivative
			$$\nabla\dtm (x)=\frac{2}{m}\int\![x-y]\,d\bar{\mu}_x(y),$$
			where $\bar{\mu}_x=\muX|_{B\left(x,F_x^{-1}(m)\right)}$. 
		\end{enumerate}  
	\end{lemma}
	Another important point for the case $m<1$ is the verification of inequality \eqref{eq:cond mod of cont}. This corresponds to bounding a
	uniform modulus of continuity for the family $\{F_x^{-1}\}_{x\in\X}$.
	An application of Lemma 3 in \citet{chazal2016rates} immediately yields the subsequent result. 
	\begin{lemma}
		Let $\mmspaceX$ be a regular Euclidean metric measure space. Suppose that there are constants $a,b>0$ such that for all $r>0$ and all $x\in\X$
		\begin{equation}\label{eq:ab standard}
			\muX(B(x,r))\geq 1\wedge ar^b.
		\end{equation}
		Then, it holds that
		\begin{equation*}
			\omega_\X(u)\leq 2 \left(\frac{h}{a}\right)^{1/b}\diam{\X}.
		\end{equation*}
	\end{lemma}
	\begin{remark}
		Condition \eqref{eq:ab standard} is frequently assumed in the context of shape analysis. Measures that fulfill \eqref{eq:ab standard} are often called \textit{(a,b)-standard} (see \citet{cuevas2009set,fasy2014confidence,chazal2016rates} for a detailed discussion of (a,b)-standard measures). In particular, we observe that our assumption \eqref{eq:cond mod of cont} is met, whenever $b< 5$.
	\end{remark}
	In the case $m=1$, it is important to observe that the DTM-function admits the following specific form:
	\begin{equation}\label{eq:dtm one identity}
		\dtmone(x)=\int_0^1\!F_x^{-1}(u)\,du=\E{||X-x||^2},
	\end{equation}
	where $X\sim\muX$. This identity gives rise to the following lemma.
	\begin{lemma}\label{lemma:dtm m=1}
		Let $\X=\mmspaceX$ denote a regular Euclidean metric measure space and let $X\sim \muX$. Then, it holds that: 
		\begin{enumerate}
			\item The function $\dtmone:\R^d\to \R$ is given as
			\begin{equation}\label{eq:dtmone rep}x=(x_1,\dots,x_d)\mapsto (x_1-c_1)^2+(x_2-c_2)^2+\dots+(x_d-c_d)^2+\zeta,\end{equation}
			where $c=(c_1,\dots,c_{d})^T=\E{X}$ and $\zeta$ denotes a finite constant that can be made explicit.
			\item The function $\dtmone:\R^d\to \R$ is three times continuously differentiable.
			\item We have $\nabla \dtmone(x)=0$ if and only if $x=\E{X}$.
			\item Consider the representation of $\dtmone$ in \eqref{eq:dtmone rep}. Set $\Gamma_y={\dtmone}^{-1}(\{y\})$ and suppose that  $\left(y-\E{||X-\E{X}||^2}\right)>2h_0>0$. In this case, the canonical level set flow $\{\Phi(v,\cdot)\}_{v\in[-h_0,h_0]}$ of $\Gamma_y$ considered as function from $[-h_0,h_0]\times{\dtmone}^{-1}([y-h_0,y+h_0])$ to $\Rd$ is for $x=(x_1,\dots,x_d)$ given as
			\begin{equation}\label{eq:level set flow}(v,x)\!\mapsto\!\! \left(\!\!(x_1-c_1)\sqrt{1+\frac{v}{||x-c||^2}}+c_1,\dots,(x_d-c_d)\sqrt{1+\frac{v}{||x-c||^2}} +c_d\right)\!.\end{equation}
		\end{enumerate}
	\end{lemma}
	In order to increase the readability of this section, the proof of Lemma \ref{lemma:dtm m=1} is postponed to Section \ref{sec:proof of lemma:dtm m=1} in the Appendix.\\
	
	\subsubsection{Discussion of assumption (\ref{eq:set level set cond}) in Condition \ref{condition}}
	To conclude this section, we consider the technical assumption \eqref{eq:set level set cond}. First of all, it is obvious (if $\dtm$ is nowhere constant) that the assumption only comes into play for $d\geq 2$. Furthermore, we observe that it is trivially fulfilled if there exists some $\epsilon>0$ such that $\Gamma_{y-\epsilon}\subset\X$, $\Gamma_{y}\subset\X$ and $\Gamma_{y+\epsilon}\subset\X$. Only if this is not the case, there might be points $y\in\mathcal{X}$ for which \eqref{eq:set level set cond} is not satisfied. However, the assumption will typically be satisfied for all points of regularity of the density $\fdtm$. To provide some intuition on this matter, we will consider the following example.
	\begin{example}[For $(\Bl0,1\Br^2,\mathcal{U}(\Bl0,1\Br^2),\|\cdot\|)$ with $m=1$ and $y=\frac{5}{12}$, \eqref{eq:set level set cond} does not hold.] 
		\label{ex:set level set ex}
		Let $\X=[0,1]^2$ and let $\muX$ stand for the uniform distribution on $\X$. In this case, using relation \eqref{eq:dtm one identity}, we obtain for $x=(x_1,x_2)\in\X$
		\[\dtmone(x)=\E{||X-x||^2} =\left(x_1-\frac{1}{2}\right)^2 + \left(x_2-\frac{1}{2}\right)^2+\frac{1}{6}.\]
		The corresponding DTM-density is supported on $[1/6,2/3]$ and it is  smooth everywhere except for $y=5/12$, where $\fdtm$ has a kink  (more details are provided in Section \ref{subsec:details ex 3} in the Appendix). The level sets $\Gamma_y$ ($y\geq 1/6$), are concentric circles centered at $(1/2,1/2)$ with radii $\sqrt{y-1/6}$. For all $y<5/12$ the level sets are fully contained in the open cube $(0,1)^2$. For all $y>5/12$, we have $\mathbb{R}^2\backslash[0,1]^2\cap\Gamma_y\neq \emptyset$, i.e., the level sets are at least partly outside of the cube $[0,1]^2$. This means that $y=5/12$ is, in a sense, a transition point. In order to check \eqref{eq:set level set cond} for $y= 5/12$, we observe that $x\in\Gamma_{\frac{5}{12}}$  implies $\indifunc{x\in \X}=1$, while for $v>0$ we find $\indifunc{\Phi(v,x)\in \X}=0$ if and only if $x$ is in one of the neighborhoods around the points $(1/2,1), (1,1/2), (1/2,0)$ and $(0,1/2)$ that correspond to the red solid lines in Figure \ref{fig:set level set cond}. Therefore, in order to evaluate integral \eqref{eq:set level set cond} we need to compute the arc lengths of the red curves in Figure \ref{fig:set level set cond} (corresponding to eight times the length of the segment $\gamma$). For this, we note that we have
		\begin{align*}
			\gamma=\beta\cdot\pi=\arccos\Big(\frac{1/2}{1/2+v}\Big)\cdot\pi.
		\end{align*}
		By a Taylor expansion, we obtain
		\begin{align*}
			0=\arccos(1)\leq\arccos\Big(\frac{1/2}{1/2+v}\Big)+\Bigg(1-\frac{1/2}{1/2+v}\Bigg)\cdot\Bigg(-\frac{1}{\sqrt{1-\frac{1/4}{(1/2+v)^2}}}\Bigg),
		\end{align*}
		where we used that the remainder term in the expansion is negative.
		Therefore, for $v\leq 1$,
		\begin{align*}
			\arccos\Big(\frac{1/2}{1/2+v}\Big)\geq\frac{v}{\sqrt{v+v^2}}\geq\frac{v}{\sqrt{2v}}.
		\end{align*}
		This yields
		\begin{align*}
			\mathcal{I}_{\mathcal{X}}(5/12;v)\geq8\cdot\pi\cdot\frac{\sqrt{v}}{\sqrt{2}}>\sqrt{v},
		\end{align*}
		which proves that for $y=5/12$ the requirement \eqref{eq:set level set cond} is not fulfilled. 
	\end{example}

	\begin{figure}
		\centering
		
		\includegraphics[width=0.35\textwidth]{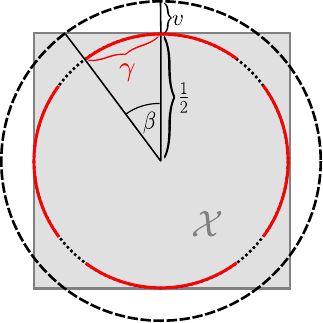}
		
		\caption{\textbf{Tangential Level Set:} Illustration of the behavior of the level sets in a neighborhood of tangential intersection point with the boundary of $\X$ in the setting of Example \ref{ex:set level set ex}.
		}\label{fig:set level set cond}
	\end{figure}

	We conclude this subsection by noting that the dimension of $\X$ heavily influences the regularity of \eqref{eq:set level set cond}. While it seems to be problematic, if $\Gamma_y$ intersects tangentially with the boundary $\partial \X$ of $\X$ for $d=2$, this is not necessarily the case for $d\geq 3$. In particular, if we consider $\X=[0,1]^3$ equipped with the uniform distribution, we find that for $y=3/4$ the level set $\Gamma_y$ tangentially touches $\partial\X$ at 6 points. However, here, it does not cause any problems. Following our considerations from Example \ref{ex:set level set ex}, one can show that condition \eqref{eq:set level set cond} holds for all points $y$ in the support of $\fdtm$ for $d\geq3$.

	\subsection{Examples of DTM-Densities}\label{sec:examples}
	In the following, we will derive $\dtm$ as well as $\fdtm$ in several simple examples explicitly and verify that in these settings Condition \ref{condition} is met almost everywhere. Since calculating $\dtm$ and $\fdtm$ explicitly is quite cumbersome (especially for $m<1$), we concentrate on one- or two-dimensional examples. In order to increase the readability of this section, we postpone the  explicit, but lengthy representations of the derived DTM-functions and densities (as well as their derivation) to Section \ref{sec:detail for examples}.\\
	
	We begin our considerations with the simplest case possible, the interval $[0,1]$ equipped with the uniform distribution. 
	\begin{example}\label{ex:first ex}
		Let $\X=[0,1]$ and let $\muX$ denote the uniform distribution on $\X$. 
		In Section \ref{subsec:details ex 1}, we derive  $\dtmvar{m}$ for general $m$. There  we also compute $\fdtmone$ (see Figure \ref{fig:dtmden} for an illustration). For $m=1$, the requirement \eqref{eq:set level set cond} does not come into play as $\X$ is one-dimensional and $\dtmone$ is nowhere constant. However, we point out that the density $\fdtmone$ is unbounded (but twice continuously differentiable in the interior of its support). In the case $m<1$ things are quite different. The function $\dtmvar{m}$ is constant on $[m/2,1-m/2]$ as all balls of radius $m/2$ centered at $x\in[m/2,1-m/2]$ have the same structure and hence the distribution of the random variable $\dtmvar{m}(X)$ has a pure point component showing that it does not admit a Lebesgue density. 
	\end{example}
	It is immediately clear that the DTM-signature can only admit a density with respect to the Lebesgue measure, if the DTM-function defined in \eqref{eq:dtm-function} is almost nowhere constant.  In the next example, we equip $\X=[0,1]$ with another probability distribution, whose density is not constant on $\X$, thereby adding structure to the metric measure space. In this case, we will find that also for $m<1$ the corresponding DTM-signature admits a Lebesgue density. 
	\begin{example}\label{ex:example 2}
		Let $\X=[0,1]$ and let $\muX$ denote the probability distribution on $[0,1]$ with density $f(x)=2x$. Let $m=0.1$. In Section \ref{subsec:details ex 2}, we derive $\dtmvar{0.1}$ explicitly and demonstrate that the random variable $\dtmvar{0.1}(X)$, $X\sim\muX$, admits a Lebesgue density in this setting (see Figure \ref{fig:dtmden} for an illustration). We observe that $\dtmvar{0.1}$ is continuously differentiable everywhere and three times continuously differentiable almost everywhere. Further, the density $\fdtmvar{0.1}$ admits one discontinuity for $y=\frac{-683}{60}+18\sqrt{\frac{2}{5}}$ and is $C^2$ almost everywhere.
	\end{example}
	We observe that the DTM-densities derived in Example \ref{ex:first ex} and Example \ref{ex:example 2} are both unbounded. This has a simple explanation. Let $\mmspaceX$ be a regular Euclidean metric measure space and denote the $d$-dimensional Lebesgue density of $\muX$ by $g_{\muX}$. Suppose that $\fdtm$ exists. Then, one can show (see e.g.\ Appendix C of \citet{weitkamp2020gromovwasserstein}) that
	\begin{equation}\label{eq:general derivative}
		\fdtm(y)=\int_{\left\{x\in\X:\dtm(x)=y\right\}}\frac{g_{\muX}(u)}{||\nabla\dtm(u)||} \,d\sigma^{d-1}(u).
	\end{equation}
	Since $d\sigma^{0}$ corresponds to integration with respect to the counting measure, the DTM-density of a one-dimensional Euclidean metric measure space is unbounded if there are $u\in\X$ with $|\nabla\dtm(u)|=0$ (this is the case in Example \ref{ex:first ex} and Example \ref{ex:example 2}). However, it is important to note that this behavior mainly occurs for one-dimensional Euclidean metric measure spaces. For higher dimensional spaces, the area (w.r.t. $d\sigma^{d-1}$) of the set $A=\{x\in\X:||\nabla\dtm(u)||=0\}$, is usually a null set. Hence, it is possible that the density $\fdtm$ defined in \eqref{eq:general derivative} remains bounded even if $A$ is non-empty (see Example \ref{ex:constant part} and Example \ref{ex:2d nonunif}).\\
	To conclude this section and in order to illustrate that the showcased regularity of the DTM-function $\dtm$ and the DTM-density $\fdtm$ does not only hold for one-dimensional settings, we consider two simple examples in $\R^2$ next. As the derivation of the family $(F_x^{-1})_{x\in\X}$ is in general very cumbersome, we restrict ourselves to the case $m=1$. 
	\addtocounter{theorem}{-3}
	\begin{example}[Continued]\label{ex:constant part}
		Recall that $\X=[0,1]^2$, $\muX$ stands for the uniform distribution on $\X$ and that $m=1$. Based on our previous considerations it is possible to derive $\fdtmone$ explicitly (see Section \ref{subsec:details ex 3} for the derivation). As illustrated in Figure \ref{fig:dtmden}, the density $\fdtmone$ is continuous. Moreover, it is twice continuously differentiable inside its support for $y\neq \frac{5}{12}$, which is also the only point where the requirements of \eqref{eq:set level set cond} are not met, as discussed previously.
	\end{example}
	\addtocounter{theorem}{2}
	
	We note that the density $\fdtmone$ derived in Example \ref{ex:constant part} is constant on $[1/6,5/12]$. This kind of behavior is also expressed when considering a disc in $\R^2$ equipped with the uniform distribution (it is easy to verify that $\fdtmone$ is a constant function in this case). It is well known that it is difficult for kernel density estimators to approximate constant pieces or a constant function. However, it is not reasonable to assume that the data stems from a uniform distribution over a compact set in many applications (such as chromatin loop analysis). More often, it is possible to assume that the data generating distribution is more concentrated in the center of the considered set. The final example of this section showcases that in such a case the corresponding DTM-signature admits a density without any constant parts even on the disk.
	\begin{example}\label{ex:2d nonunif}
		Let $\X$ denote a disk in $\R^2$ centered at $(0,0)$ with radius 1 and let $\muX$ denote probability measure with density
		\begin{equation}\label{eq:2d nonunif density}
			f(x_1,x_2)=\begin{cases}-\frac{2}{\pi}\left(x_1^2+x_2^2-1\right) & x_1^2+x_2^2 \leq 1 ,\\
				0 & \, \text{else.}\end{cases}
		\end{equation}
		In this framework, we derive $\dtmone$ and $\fdtmone$ in Section \ref{subsec:details ex 4}. We observe that the level sets $\Gamma_y$ (of $\dtmone$) are contained in $\X$ for any $y\in [1/3,4/3]$, i.e., condition \eqref{eq:set level set cond} is met for all $y\in(1/3,4/3)$ in this setting. Further, we realize that $\fdtmone$ (see Figure \ref{fig:dtmden} for an illustration) is smooth and nowhere constant on the interior of its support.
	\end{example}
	\begin{figure}
		\centering
		\begin{minipage}[c]{0.24\textwidth}
			\includegraphics[width = \textwidth,height=\textwidth]{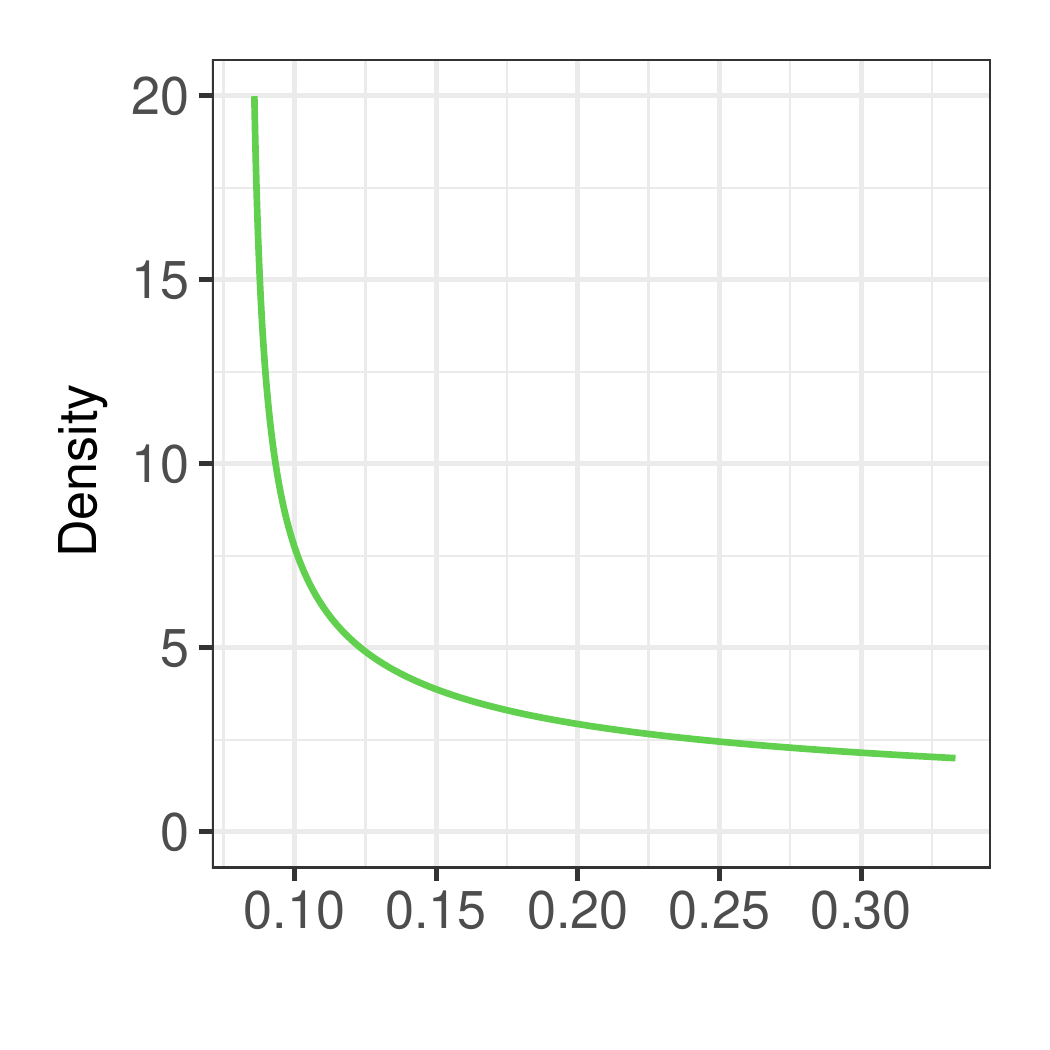}
		\end{minipage}
		\begin{minipage}[c]{0.24\textwidth}
			\includegraphics[width =\textwidth,height=\textwidth]{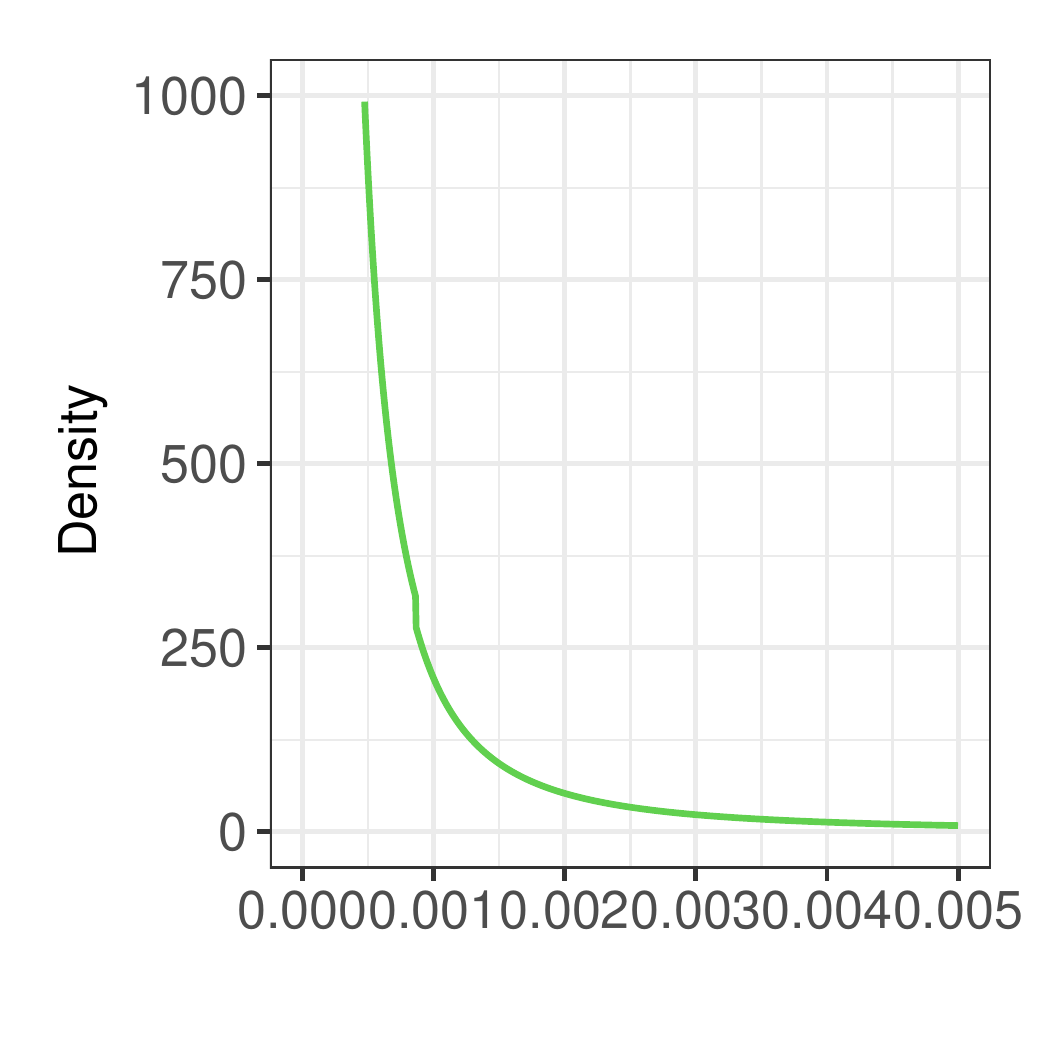}
		\end{minipage}
		\begin{minipage}[c]{0.24\textwidth}
			\includegraphics[width = \textwidth,height=\textwidth]{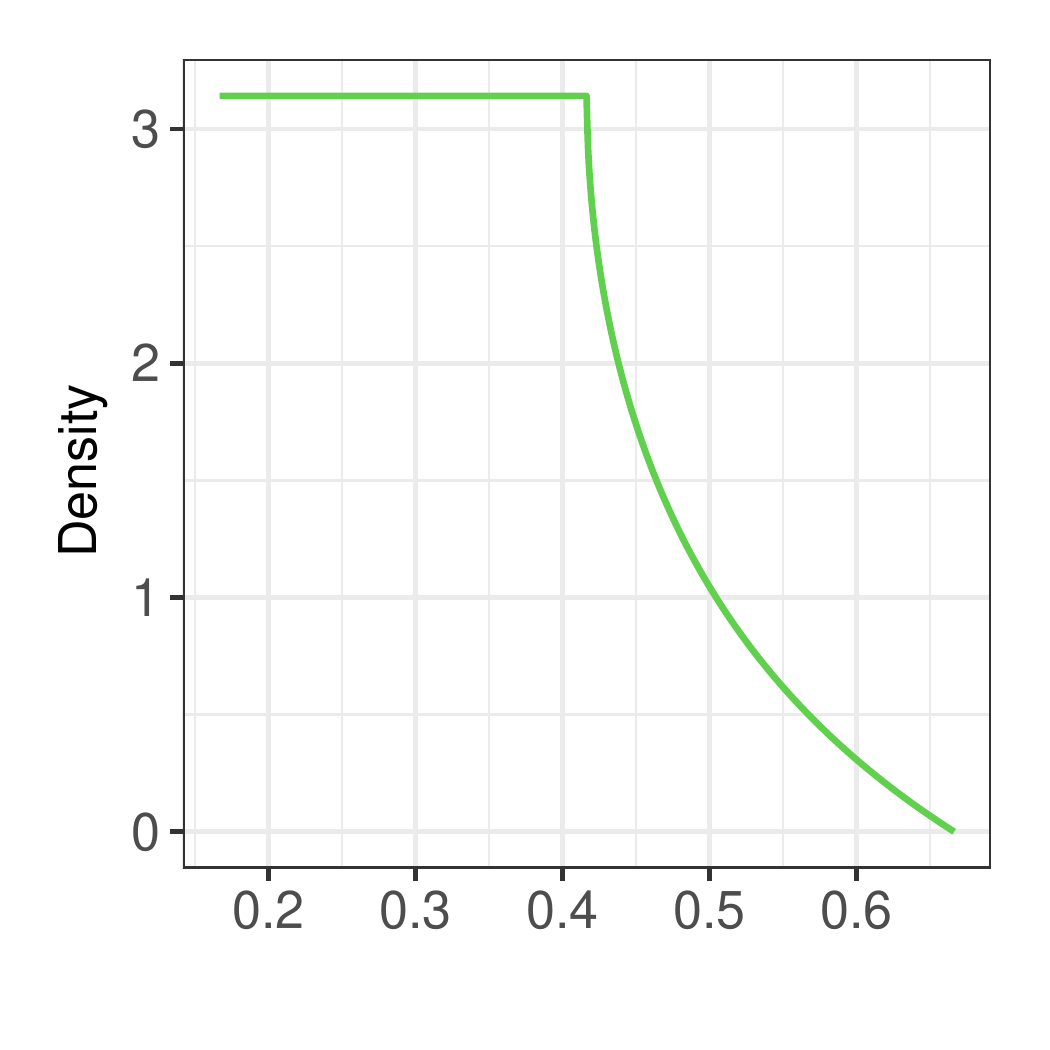}
		\end{minipage}
		\begin{minipage}[c]{0.24\textwidth}
			\includegraphics[width = \textwidth,height=\textwidth]{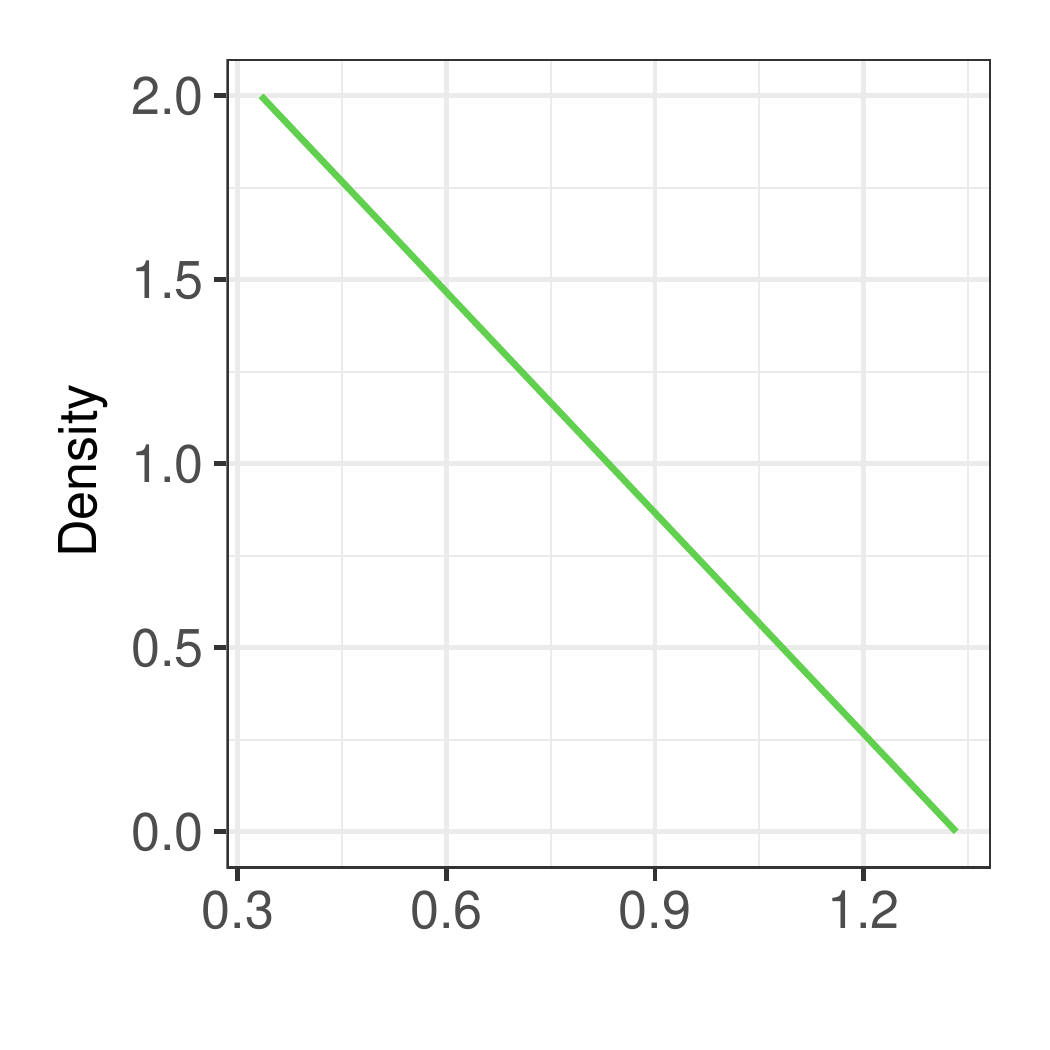}
		\end{minipage}
		\caption{\textbf{Distance-to-Measure signature:} Illustration of the densities calculated in Example \ref{ex:first ex}-\ref{ex:example 2}, Example \ref{ex:constant part} and Example \ref{ex:2d nonunif} (from left to right). 		
		}
		\label{fig:dtmden}
	\end{figure} 
	\subsection{Theoretical Results}\label{subsec:asym theo}
	In this section we study the asymptotic behavior of the kernel estimator of the DTM-density \eqref{eq:def of kde}. Clearly, standard methodology implies the following pointwise central limit theorem for the kernel estimator $\semikde$ defined in \eqref{eq:def of semikde}. 
	\begin{theorem}\label{thm:pointwise limit}
		Assume Setting \ref{setting} and suppose that $\dtm(X_1)$ admits a density that is twice continuously differentiable in a neighborhood of $y$.  Suppose further that the kernel $K:\R\to \R_+$, is an even, twice continuously differentiable function with $\supp(K)=[-1,1]$.
		Then, it holds for $n\to\infty$, $nh\to\infty$ and $h=o\left(n^{-1/5}\right)$ that
		\[\sqrt{nh}\left(\semikde(y)-\fdtm(y)\right)\Rightarrow N\left(0,\fdtm(y)\int\!K^2(u)\,du\right).\]
	\end{theorem}
	Surprisingly perhaps, despite the complicated stochastic dependence of the random variables $\empdtm(X_i)$, asymptotically, $\semikde$ and $\kde$ behave equivalently in the following sense.
	\begin{theorem}\label{thm:kde limit}
		Assume Setting \ref{setting} and let Condition \ref{condition} hold. Then, it holds for $n\to\infty$, $nh\to\infty$ and $h=o\left(n^{-1/5}\right)$ that
		\[\sqrt{nh}\left(\kde(y)-\fdtm(y)\right)\Rightarrow N\left(0,\fdtm(y)\int\!K^2(u)\,du\right).\]
	\end{theorem}
	As the  proof of Theorem \ref{thm:kde limit} is lengthy and quite technical, it has been deferred to Appendix \ref{sec:proof of thm:kde limit}. There, we will write the density estimator $\kde$ as a $U$-statistic plus remainder terms as previously discussed at the end of Section \ref{sec:proposed approach}. Then, using a Hoeffding decomposition, the dependencies can be handled. However, showing that the remainder terms vanish is not trivial and requires the application of some tools from geometric measure theory.
	
	\section{Simulations}\label{sec:simulations}
	In the following, we investigate the finite sample behavior of $\kde$ in Monte Carlo simulations. To this end, we illustrate the pointwise limit derived in Theorem \ref{thm:kde limit} in the setting of Example \ref{ex:2d nonunif} and exemplarily highlight the  potential of $\kde$ to discriminate between different Euclidean metric measure spaces. All simulations were performed in $\mathsf{R}$ (\citet{Rbasicversion}). An R package, that implements the calculation of $\kde$ as well as some basic analytical tools, is available at \url{https://github.com/cweitkamp3/DTMdemo}.
	\subsection{Pointwise Limit} We start with the illustration of Theorem \ref{thm:pointwise limit}. To this end, we consider the Euclidean metric measure space $\mmspaceX$ from Example \ref{ex:2d nonunif}. Recall that in this setting, $\X$ denotes a disk in $\R^2$ centered at $(0,0)$ with radius $1$ and that $\muX$ denotes the probability measure with density
	\[f(x_1,x_2)=\begin{cases}-\frac{2}{\pi}\left(x_1^2+x_2^2-1\right) & x_1^2+x_2^2 \leq 1 ,\\
		0 & \, \text{else.}\end{cases}\]
	Now, we choose $m=1$ and consider
	\[\kdeone(y)=\frac{1}{nh_1}\sum_{i=1}^n K_{Bi}\left(\frac{\empdtm(X_i)-y}{h_1}\right),\]
	where $K_{Bi}$ denotes the Biweight kernel, i.e.,
	\begin{equation}
		K_{Bi}(u)= \begin{cases}\frac{15}{16}\left(1-u^2\right)^2 & |u| \leq 1 ,\\
			0 & \, \text{else.}\end{cases} 
	\end{equation}
	Since we have calculated $\dtmone$ explicitly (see Eq.\ \ref{eq:ex4dtm}), it is of interest to compare the behavior of $\kdeone(y)$ to the one of
	\[\semikdeone(y)=\frac{1}{nh_2}\sum_{i=1}^n K_{Bi}\left(\frac{\dtmone(X_i)-y}{h_2}\right).\]
	As discussed previously, $\semikdeone$  a kernel density estimator based on independent data, whose limit behavior is well understood (see Theorem \ref{thm:pointwise limit}). Nevertheless, for $y\in(1/3,4/3)$, $\semikdeone(y)$ and $\kdeone(y)$ admit the same asymptotic behavior according to Theorem \ref{thm:kde limit}, whose requirements can be easily checked in this setting (see Example \ref{ex:2d nonunif}). In order to illustrate this, we generate two independent samples $\{X_i\}_{i=1}^n$ and $\{X'_i\}_{i=1}^n$ of $\muX$ and calculate $\Delta_n=\{\empdtmone(X_i)\}_{i=1}^n$ as well as $D_n=\{\dtmone(X_i)\}_{i=1}^n$ for $n=50,500, 2500,5000$. We set 
	$$h_1=(1.06\min \{s(\Delta_n),\text{IQR}(\Delta_n)/1.34)\})^{5/4}n^{-1/4}$$
	and 
	$$h_2=(1.06\min \{s(D_n),\text{IQR}(D_n)/1.34)\})^{5/4}n^{-1/4},$$
	where $s$ is the usual sample standard deviation and IQR denotes the inter quartile range. Based on $\Delta_n$ and $D_n$, we choose a central value of $y$ and calculate
	\begin{equation}\label{eq:emp and semi emp pointwise limits}\sqrt{nh_1}(\kdeone(y)-\fdtmone(y))\text{ and }\sqrt{nh_2}(\semikdeone(y)-\fdtmone(y)).\end{equation}
	For each $n$, we repeat this process 5,000 times. The finite sample distributions of the quantities defined in \eqref{eq:emp and semi emp pointwise limits} are compared to their theoretical normal counter part in Figure \ref{fig:pointwiseclt} (exemplarily for the specific choice of $y=0.7$). The kernel density estimators displayed  highlight that the asymptotic behavior of $\kdeone(y)$ (red) matches that of $\semikdeone(y)$ (green). Further, we observe that even for small samples sizes both finite sample distributions strongly resemble their theoretical normal limit distribution (blue).
	\begin{figure}
		\centering
		\begin{minipage}[c]{0.24\textwidth}
			\centering
			\small{\quad$n=50$}
			\includegraphics[width = \textwidth,height=\textwidth]{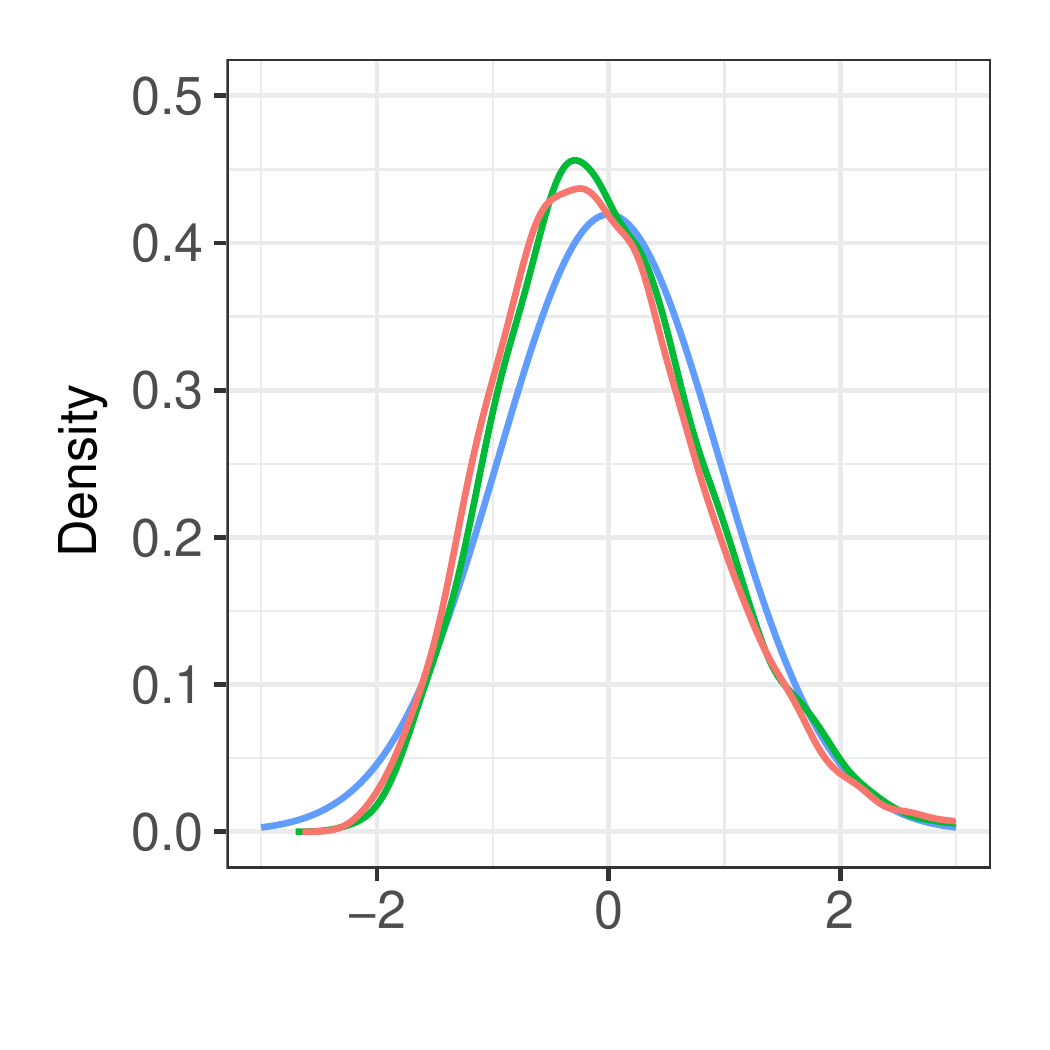}\\
		\end{minipage}
		\begin{minipage}[c]{0.24\textwidth}
			\centering
			\small{\quad$n=500$}
			\includegraphics[width = \textwidth,height=\textwidth]{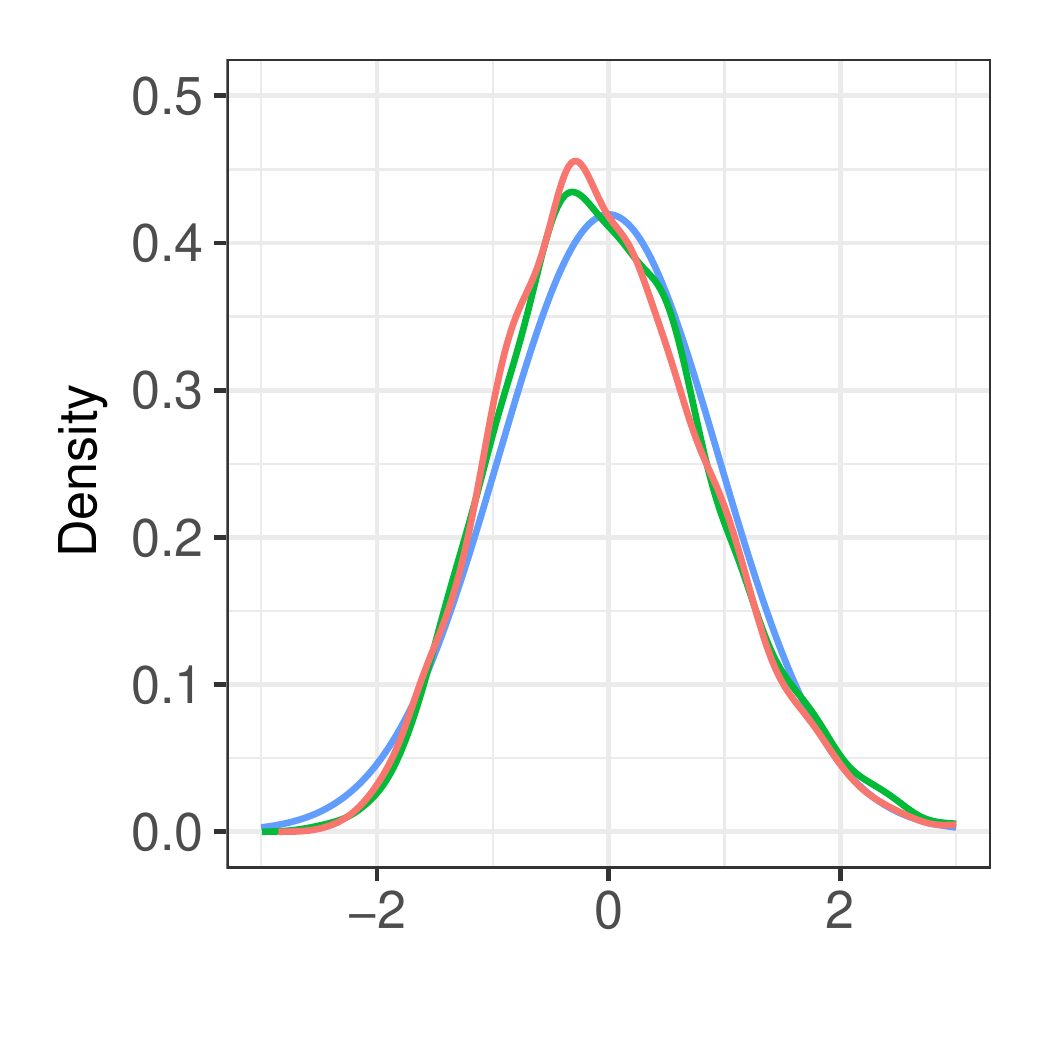}\\
		\end{minipage}
		\begin{minipage}[8]{0.24\textwidth}
			\centering
			\small{\quad$n=2500$}
			\includegraphics[width = \textwidth,height=\textwidth]{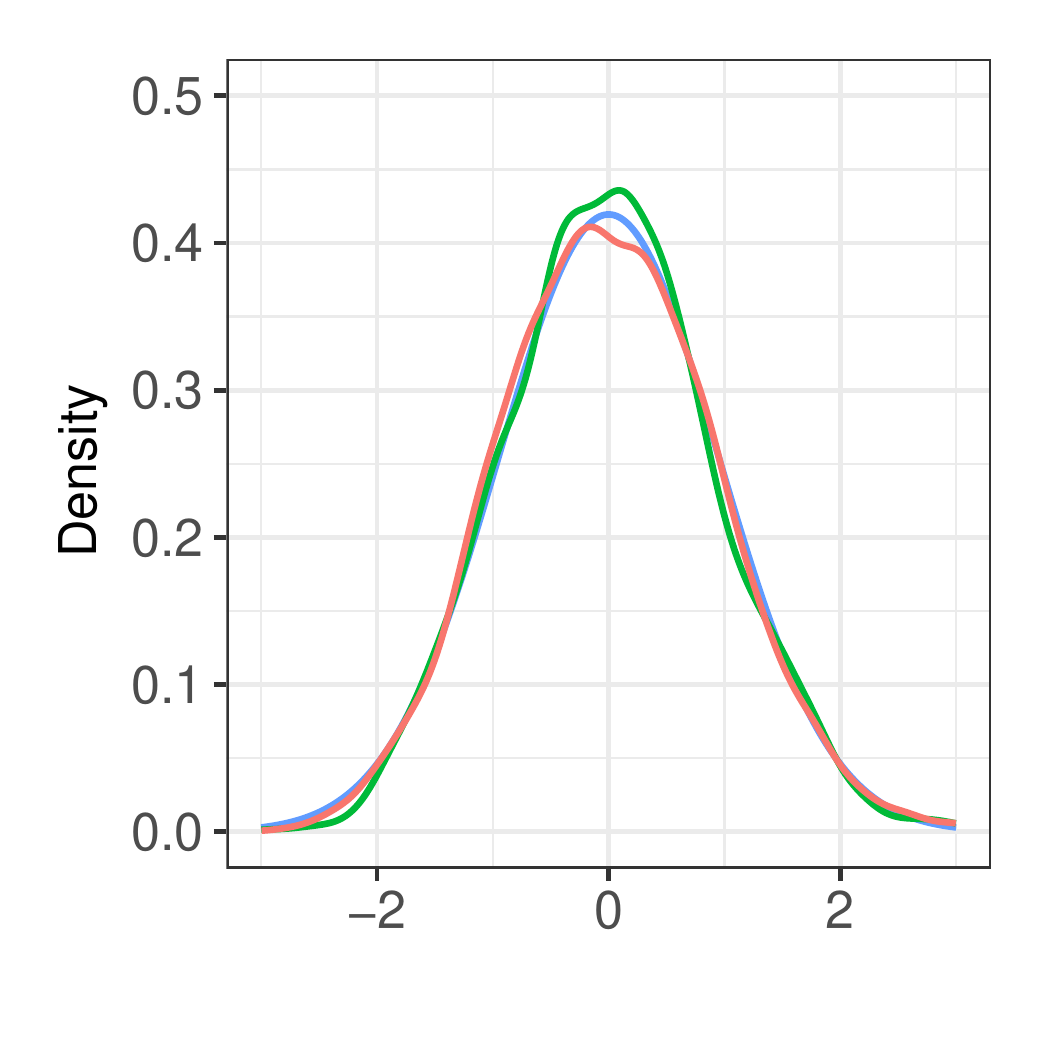}\\
		\end{minipage}
		\begin{minipage}[8]{0.24\textwidth}
			\centering
			\small{\quad$n=5000$}
			\includegraphics[width = \textwidth,height=\textwidth]{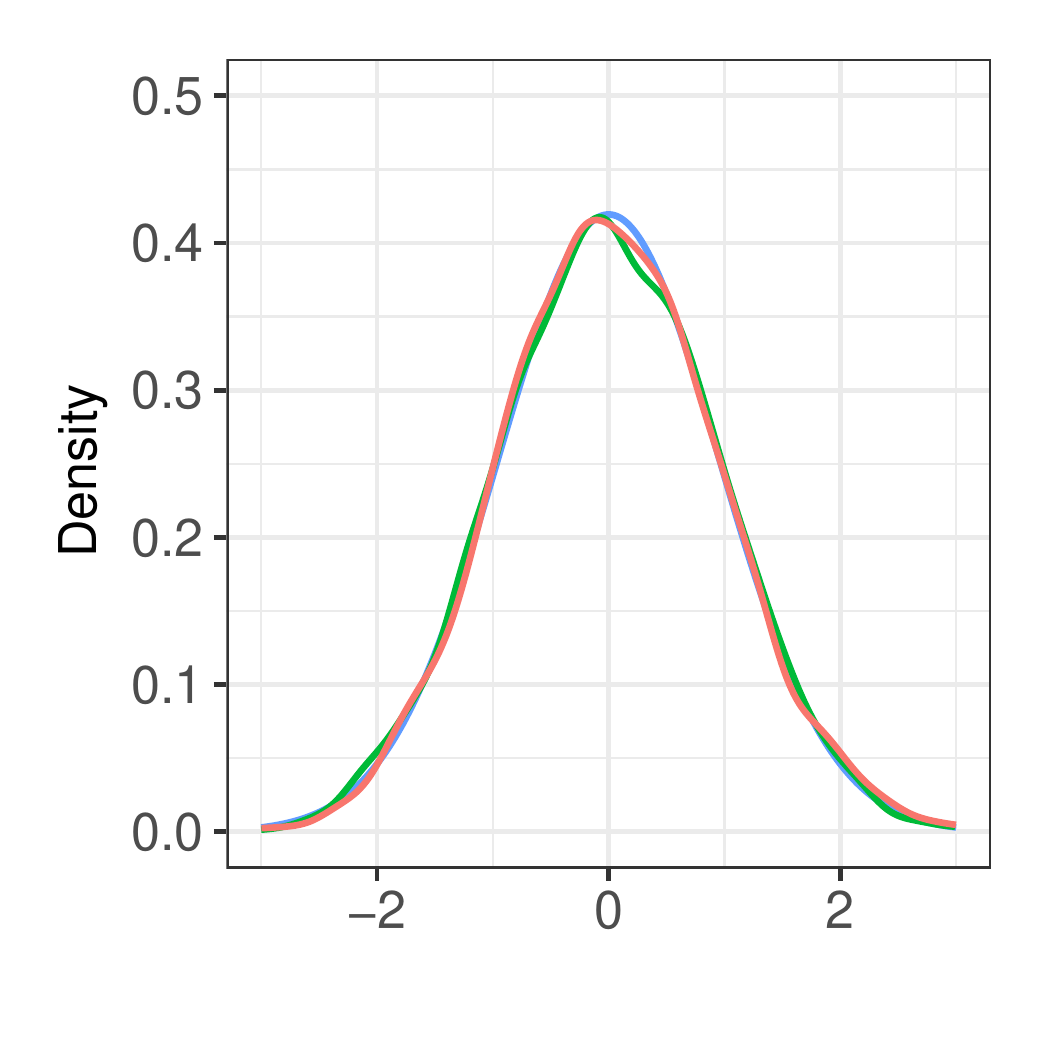}\\
		\end{minipage}
		\caption{\textbf{Pointwise limit distribution:} Kernel density estimators of $\sqrt{nh_1}(\kdeone(0.7)-\fdtmone(0.7))$ (in red) and $\sqrt{nh_2}(\semikdeone(0.7)-\fdtmone(0.7))$ (in green) for $n=50,500,2500,5000$ (from left to right, sample size 5,000) and the normal limiting density (blue).
		}\label{fig:pointwiseclt}
	\end{figure}
	
	\subsection{Discriminating Properties}\label{subsec:disc props}
	In the remainder of this section, we will showcase empirically the potential of the DTM-signature for discriminating between different Euclidean metric measure spaces. To this end, let $\mu_{\mathcal{Y}_1}$ stand for the uniform distribution on a 3D-pentagon (inner pentagon side length: 1, Euclidean distance between inner and outer
	pentagon: 0.4, height: 0.4) and let $\mu_{\mathcal{Y}_7}$ denote the uniform distribution on a torus (center radius: 1.169, tube radius: 0.2) with the same center and orientation (see the plots for $t_1=0$ and $t_7=1$ in Figure \ref{fig:geodesic}). In order to interpolate between these measures, let $\Pi_{\mu_{\mathcal{Y}_1}}^{\mu_{\mathcal{Y}_7}}(t)$, $t\in[0,1]$, denote the 2-Wasserstein geodesic between $\mu_{\mathcal{Y}_1}$ and $\mu_{\mathcal{Y}_7}$ (see e.g.\ \citet[Sec. 5.4]{santambrogio2015optimal} for a formal definition). Figure \ref{fig:geodesic} displays the Euclidean metric measure spaces $\mathcal{Y}_i$, $1\leq i\leq 7$, corresponding to $\mu_{\mathcal{Y}_i}=\Pi_{\mu_{\mathcal{Y}_1}}^{\mu_{\mathcal{Y}_7}}(t_i)$ for $t_i\in\{0,0.1,0.2,0.4,0.6,0.8,1\}$ (the geodesic has been approximated discretely based on 40,000 points with the $\mathsf{WSGeometry}$-package \citep{WSgeom}).
	
	\begin{figure}
		\centering
		\begin{minipage}[c]{0.13\textwidth}
			\centering
			\includegraphics[width = \textwidth,height=\textwidth]{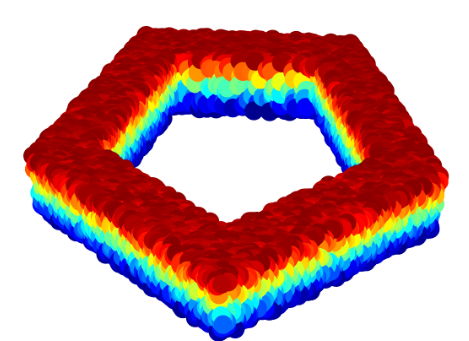}\\
			\vspace*{1mm}
			{~~$t_1=0$}
		\end{minipage}
		\begin{minipage}[c]{0.13\textwidth}
			\centering
			\includegraphics[width = \textwidth,height=\textwidth]{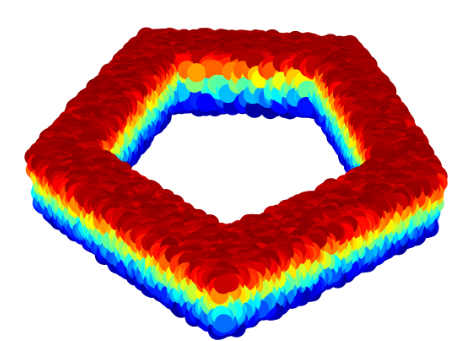}\\
			\vspace*{1mm}
			{~~$t_2=0.1$}
		\end{minipage}
		\begin{minipage}[8]{0.13\textwidth}
			\centering
			\includegraphics[width = \textwidth,height=\textwidth]{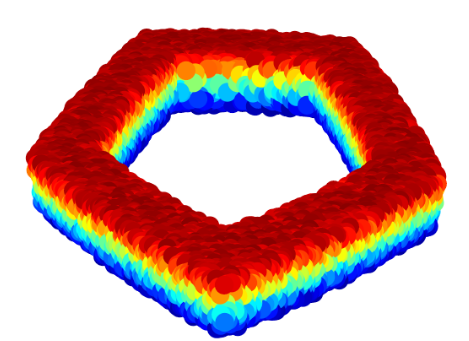}\\
			\vspace*{1mm}{~~$t_3=0.2$}
		\end{minipage}
		\begin{minipage}[c]{0.13\textwidth}
			\centering
			\includegraphics[width = \textwidth,height=\textwidth]{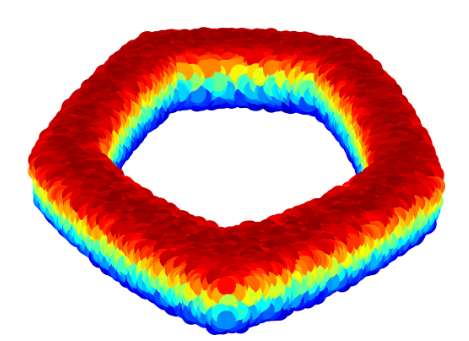}\\
			\vspace*{1mm}	{~~$t_4=0.4$}
		\end{minipage}
		\begin{minipage}[c]{0.13\textwidth}
			\centering
			\includegraphics[width = \textwidth,height=\textwidth]{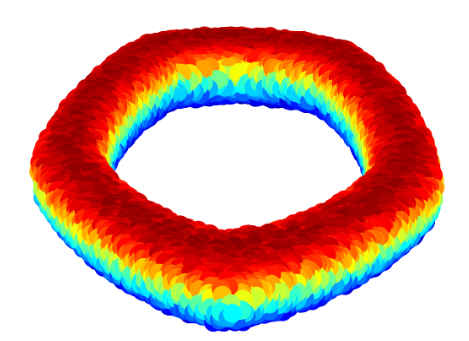}\\
			\vspace*{1mm}{~~$t_5=0.6$}
		\end{minipage}
		\begin{minipage}[c]{0.13\textwidth}
			\centering
			\includegraphics[width = \textwidth,height=\textwidth]{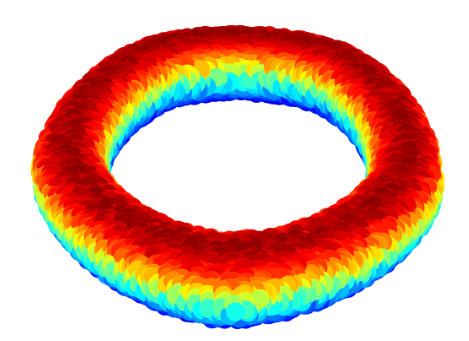}\\
			\vspace*{1mm}{~~$t_6=0.8$}
		\end{minipage}
		\begin{minipage}[c]{0.13\textwidth}
			\centering
			\includegraphics[width = \textwidth,height=\textwidth]{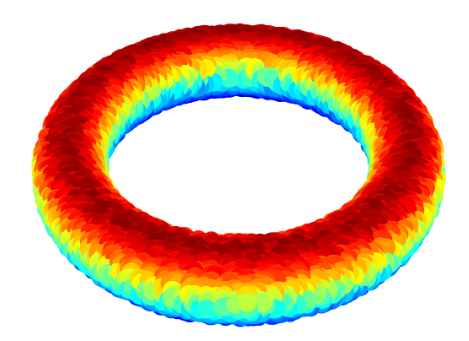}\\
			\vspace*{1mm}{~~$t_7=1$}
		\end{minipage}
		\caption{\textbf{Metric measure spaces:} Graphical illustration of the metric measure spaces $\left\lbrace\mathcal{Y}_i \right\rbrace_{i=1}^7$.	 
		}\label{fig:geodesic}
	\end{figure}
	
	In this example, we are not interested in only finding local changes, but we want to distinguish between Euclidean metric measure spaces that differ globally. Hence, $m=1$ seems to be the most reasonable choice. At the end of this section, we will illustrate the influence of the parameter $m$ in the present setting. We draw independent samples of size $n$ from $\mu_{\mathcal{Y}_i}$, denoted as $\{Y_{j,n,i}\}_{j=1}^n$, and calculate $\Delta_{n,i}=\{\empdtmone(Y_{j,n,i})\}_{j=1}^n$ and $\kdeyvar{i}{1}$ based on each of these samples for $1\leq i\leq 7$ and $n=500, 2500,5000,10000$ (Biweight kernel, $h_i=1.06\min \{s(\Delta_n),\text{IQR}(\Delta_n)/1.34)\}n^{-1/5})$. We repeat this procedure for each $i$ and $n$ 10 times and display the resulting kernel density estimators in the upper row of Figure \ref{fig:dtm dsicrimination}. While it is not possible to reliably distinguish between the realizations of $\kdeyvar{1}{1}$ (blue-green), $\kdeyvar{2}{1}$ (orange) and $\kdeyvar{3}{1}$ (blue) by eye for $n=500$, this is very simple for $n\geq 2500$. Now that we have estimated the densities, we can choose a suitable notion of distance between densities (e.g.\ the $L^1$-distance) and perform a linkage clustering in order to showcase that the illustrations in the upper row are not deceptive and that it is indeed possible to discriminate between the Euclidean metric measure spaces considered based on the kernel density estimators of the respective DTM-densities. To this end, we calculate the $L_1$-distance between the kernel density estimators considered and perform an average linkage clustering on the resulting distance matrix for each $n$. The results are showcased in the lower row of Figure \ref{fig:dtm dsicrimination}. The average linkage clustering confirms our previous observations.

	\begin{figure}
		\centering
		\begin{minipage}[c]{0.22\textwidth}
			\centering
			\scriptsize{\qquad$n=500, m=1$}
			\includegraphics[width = \textwidth,height=\textwidth]{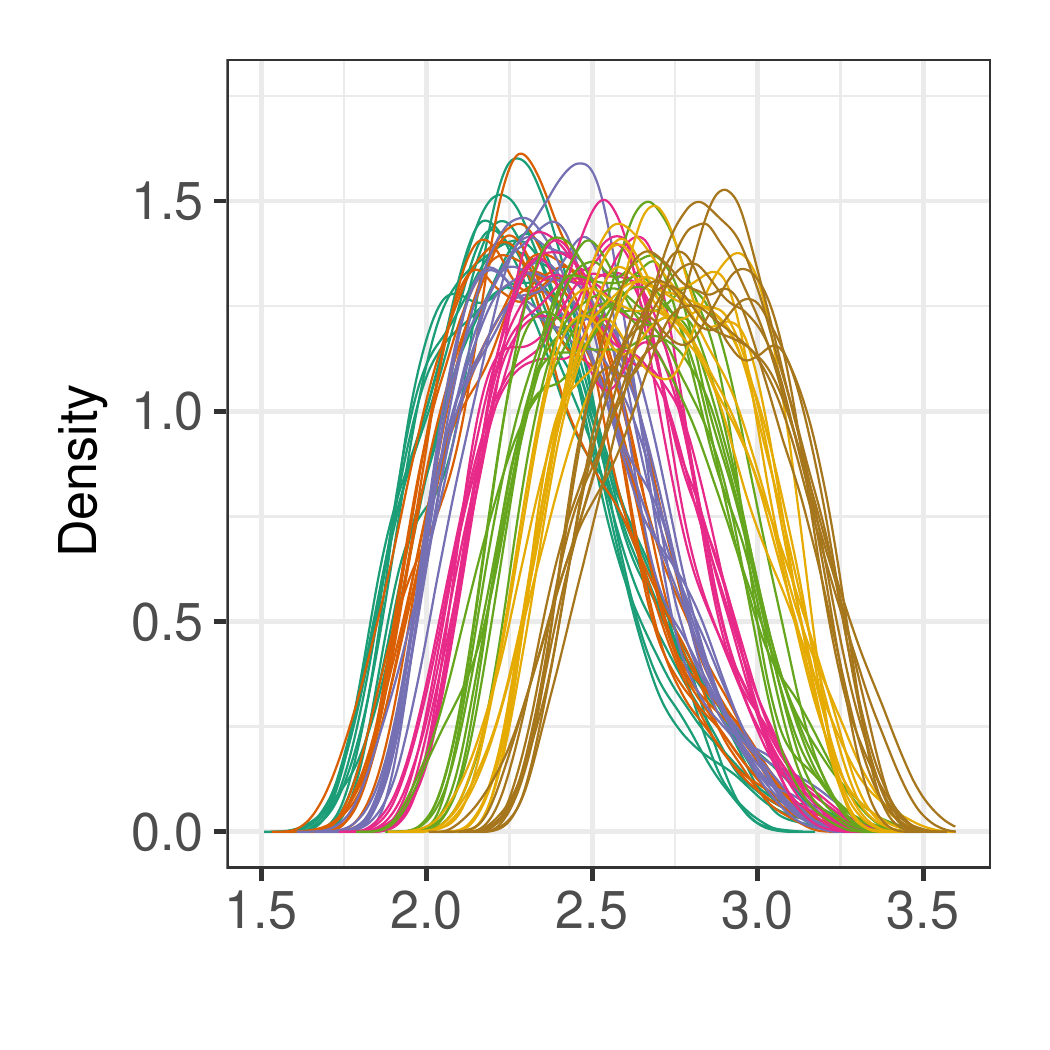}\\
		\end{minipage}
		\begin{minipage}[c]{0.22\textwidth}
			\centering
			\scriptsize{\qquad$n=2500, m=1$}
			\includegraphics[width = \textwidth,height=\textwidth]{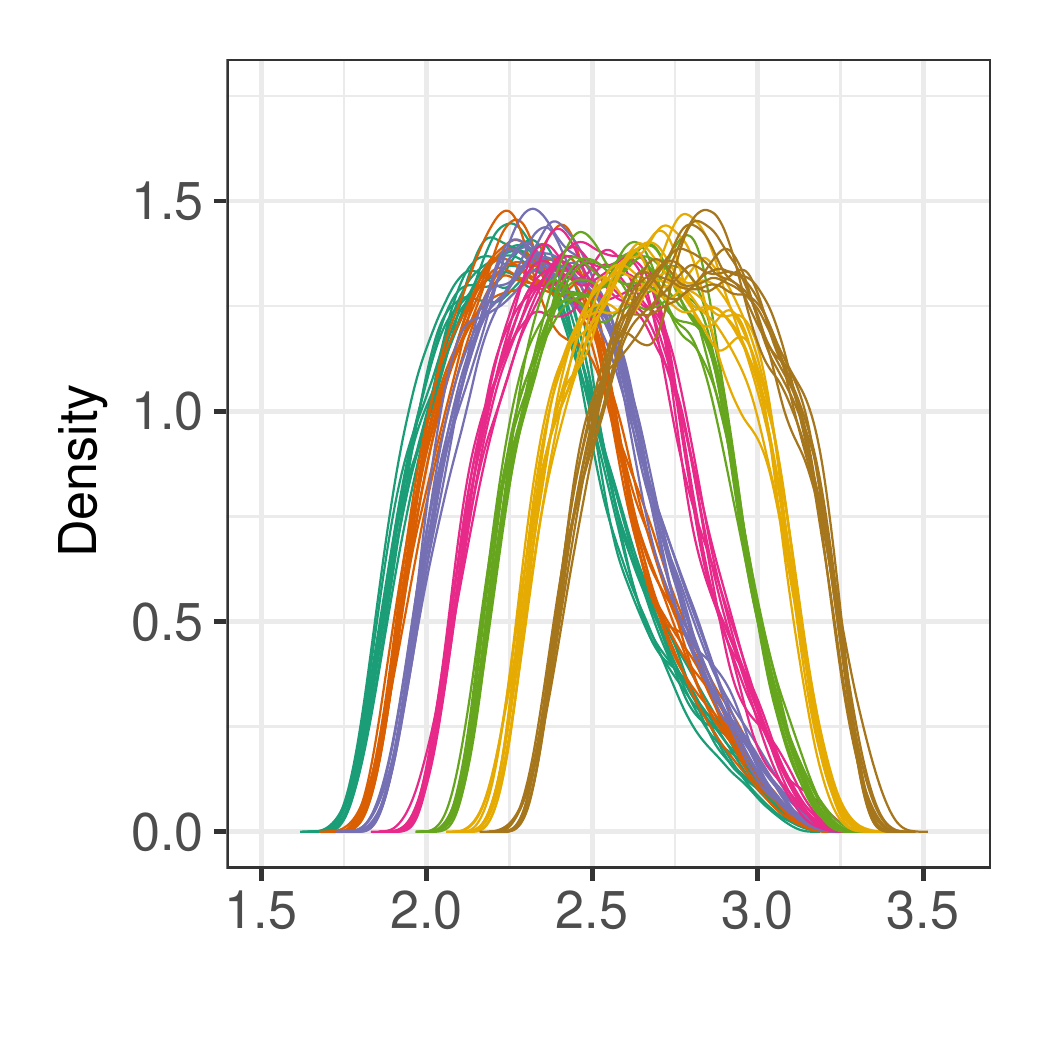}\\
		\end{minipage}
		\begin{minipage}[8]{0.22\textwidth}
			\centering
			\scriptsize{\qquad$n=5000, m=1$}
			\includegraphics[width = \textwidth,height=\textwidth]{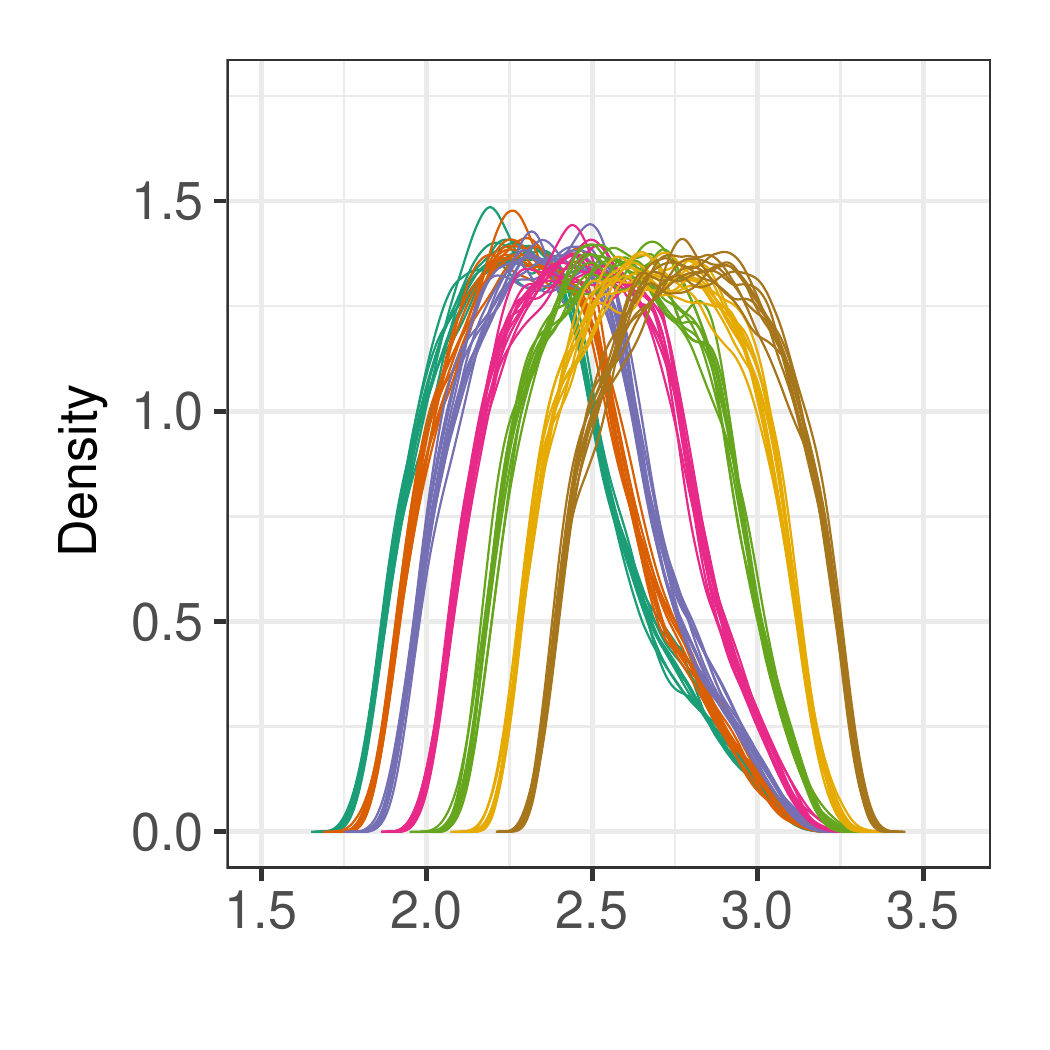}\\
		\end{minipage}
		\begin{minipage}[8]{0.22\textwidth}
			\centering
			\scriptsize{\qquad$n=10000, m=1$}
			\includegraphics[width = \textwidth,height=\textwidth]{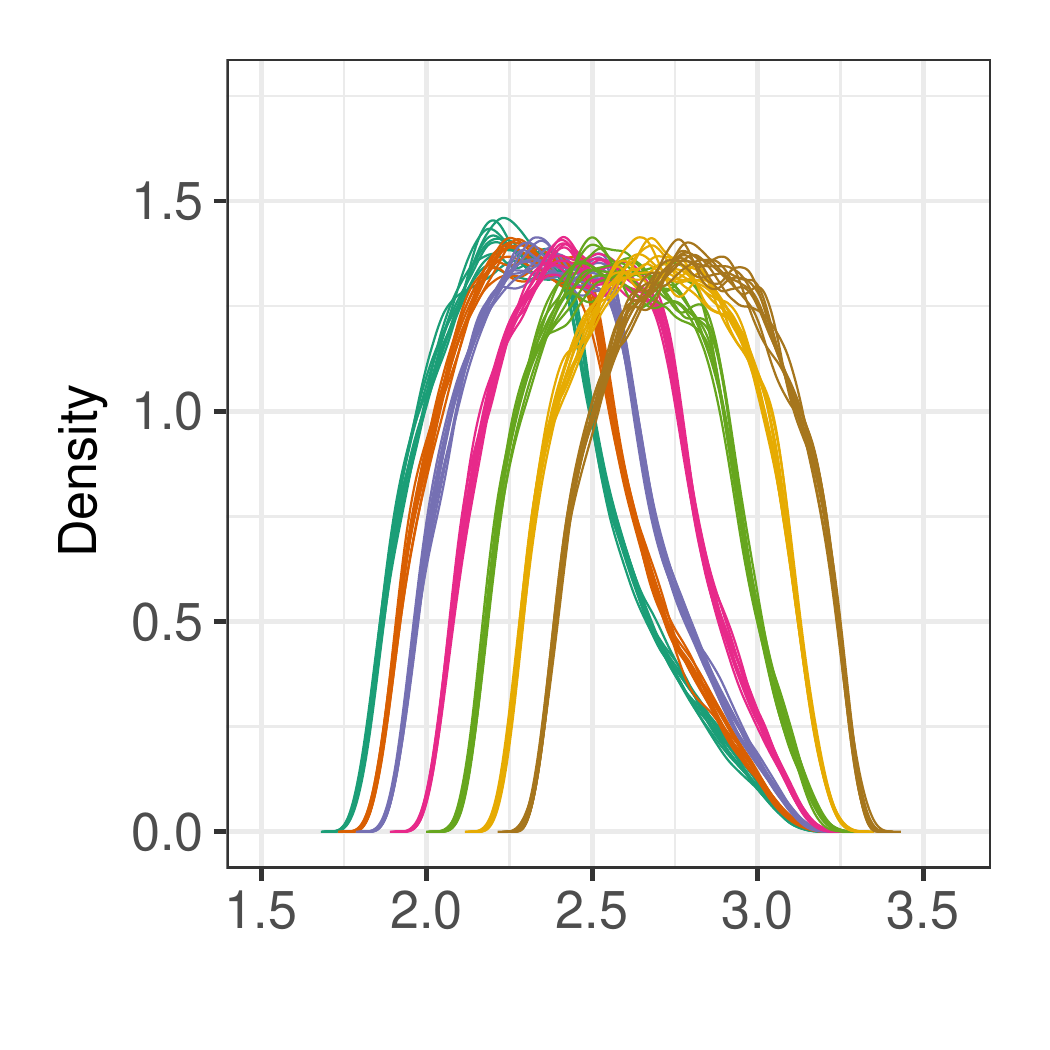}
		\end{minipage}
		\begin{minipage}[c]{0.04\textwidth}
			\hspace{0.1cm}\includegraphics[height=0.135\textheight]{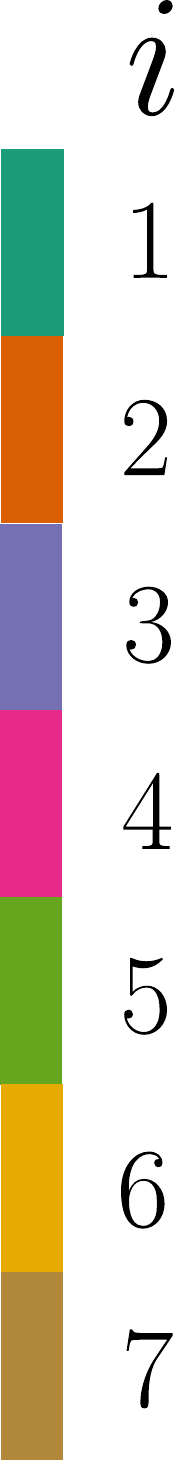}
			\vspace{0.3cm}
		\end{minipage}
		\\~\\
		\begin{minipage}[c]{0.22\textwidth}
			\centering
			\includegraphics[width = \textwidth,height=\textwidth]{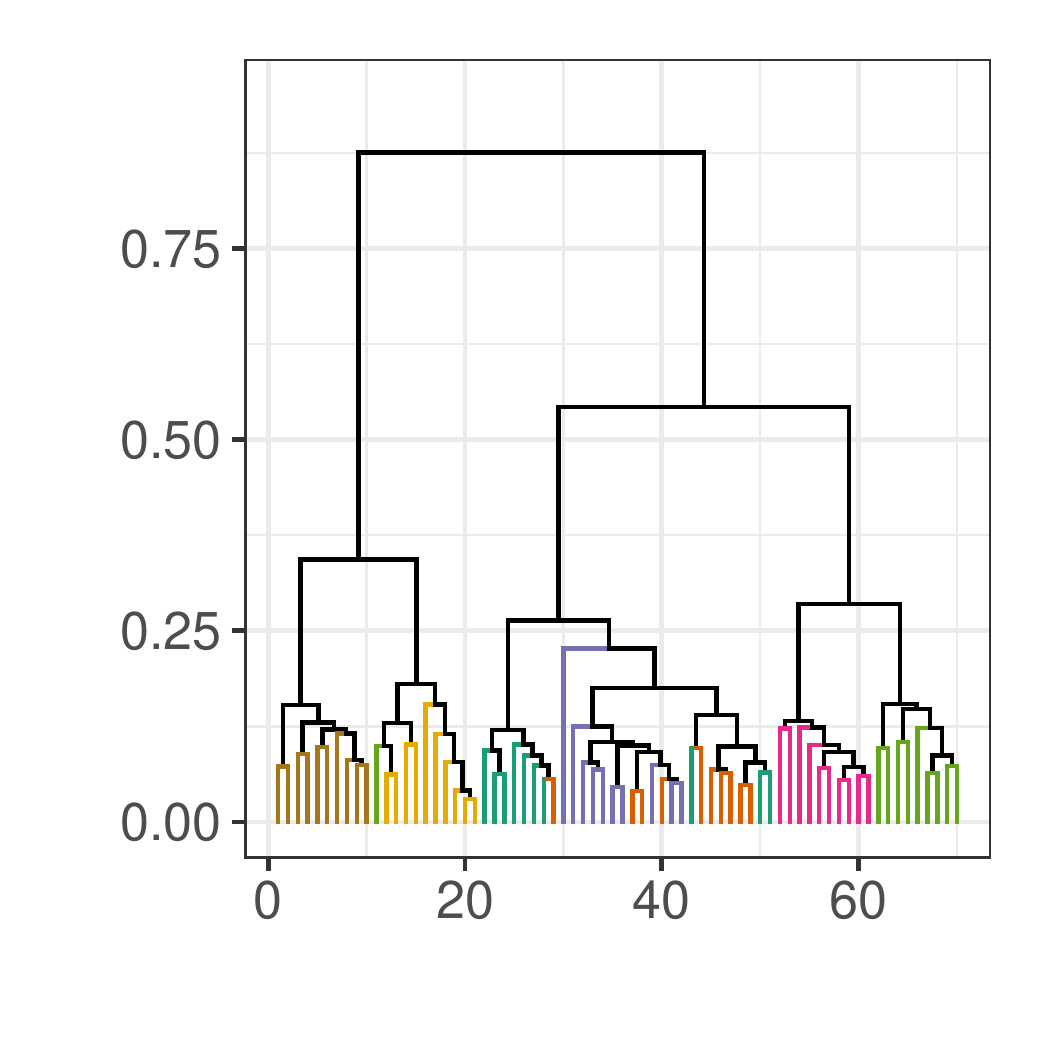}
		\end{minipage}
		\begin{minipage}[c]{0.22\textwidth}
			\centering
			\includegraphics[width = \textwidth,height=\textwidth]{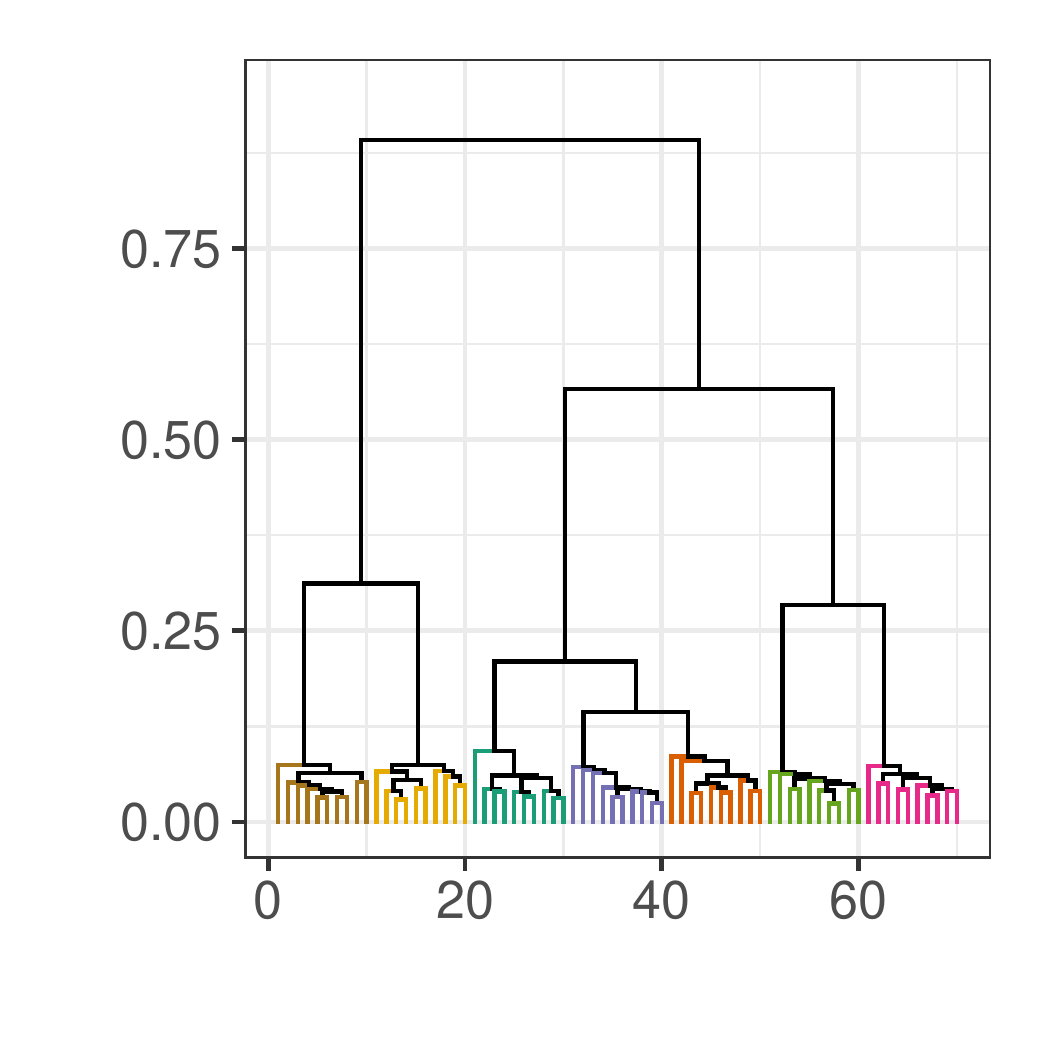}
		\end{minipage}
		\begin{minipage}[8]{0.22\textwidth}
			\centering
			\includegraphics[width = \textwidth,height=\textwidth]{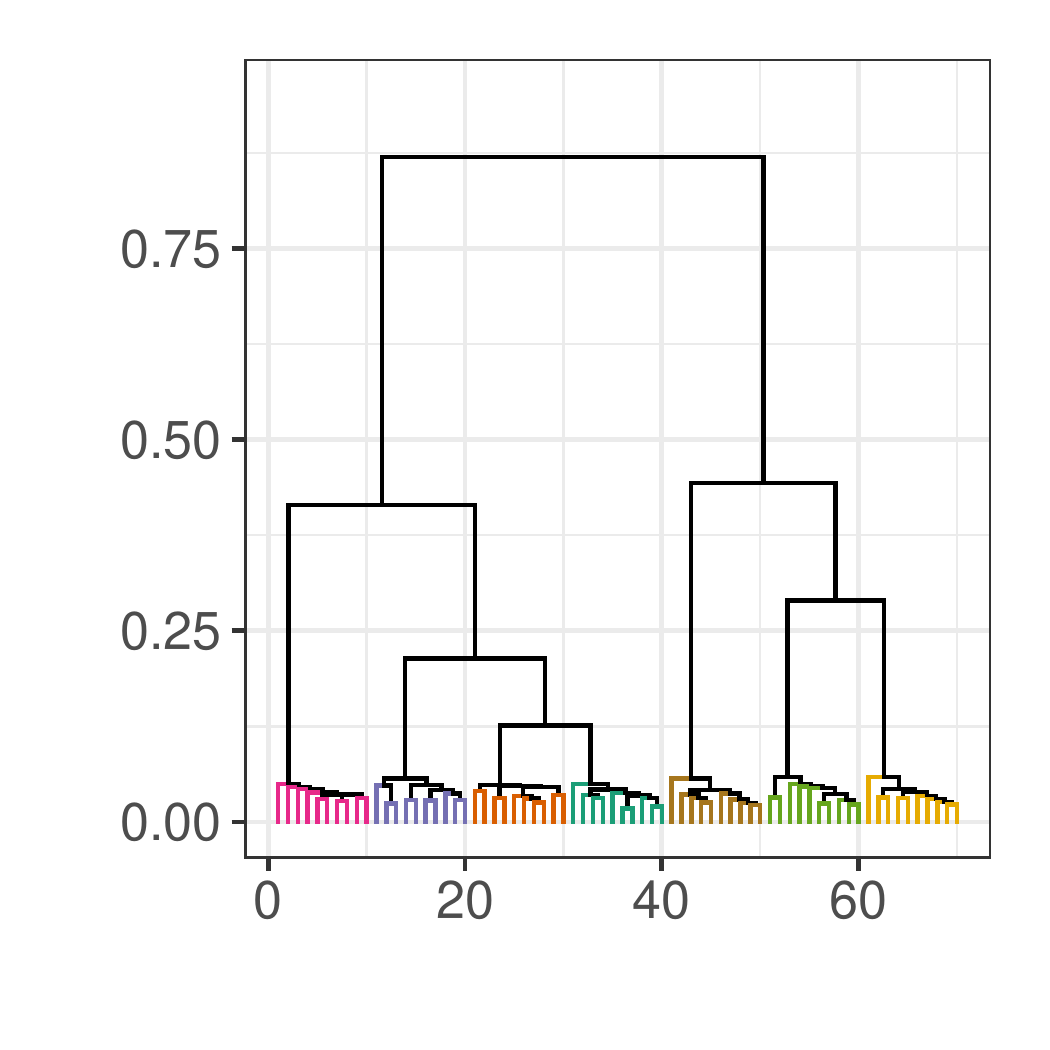}
		\end{minipage}
		\begin{minipage}[8]{0.22\textwidth}
			\centering
			\includegraphics[width = \textwidth,height=\textwidth]{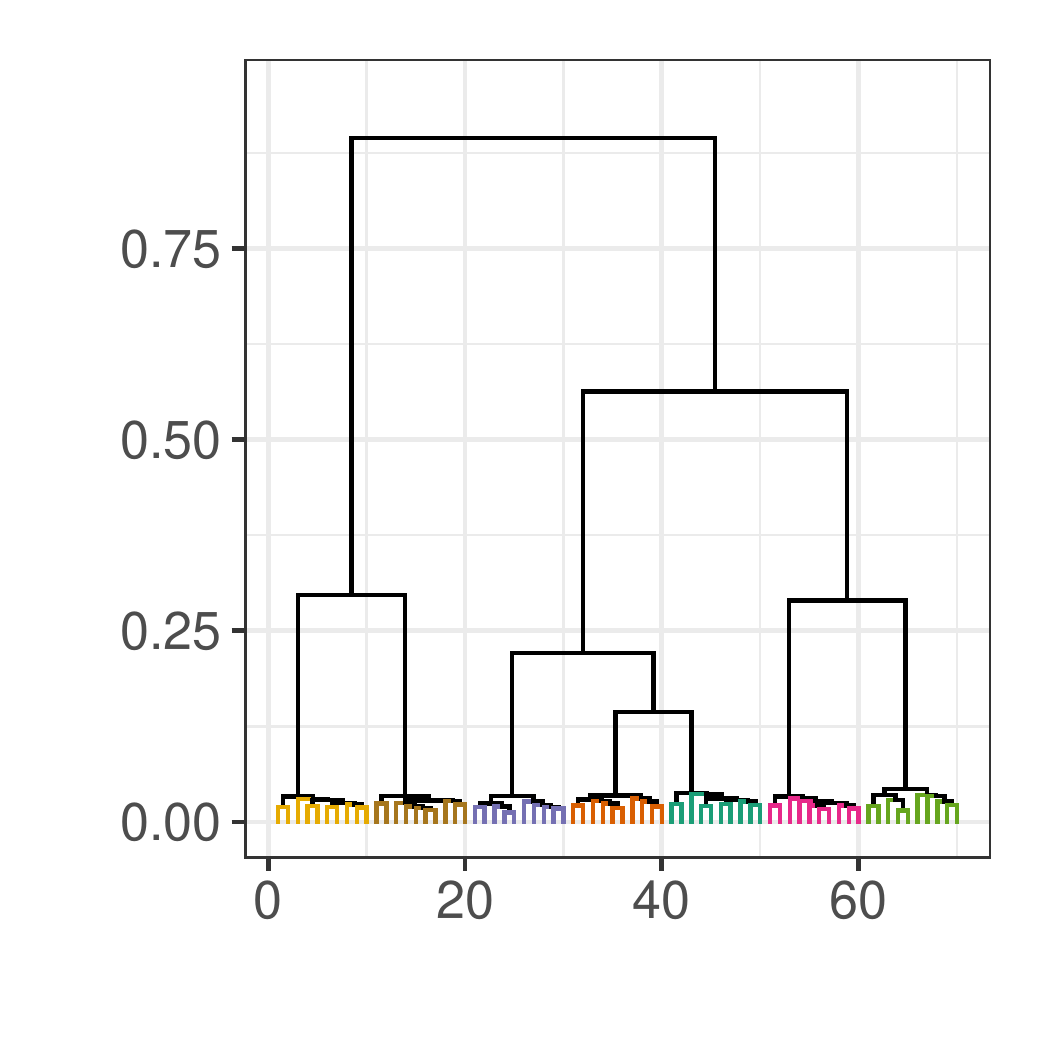}
		\end{minipage}
		\begin{minipage}[c]{0.04\textwidth}
			\hfill
		\end{minipage}
		\caption{\textbf{Discriminating between Euclidean metric measure spaces:} Upper row: Ten realizations of the kernel density estimators $\kdeyvar{1}{1}$ (blue-green), $\kdeyvar{2}{1}$ (orange), $\kdeyvar{3}{1}$ (blue), $\kdeyvar{4}{1}$ (pink), $\kdeyvar{5}{1}$ (green), $\kdeyvar{6}{1}$ (yellow) and $\kdeyvar{7}{1}$ (brown) for $n=500,2500,5000,10000$ (from left to right). Lower row: The results of an average linkage clustering of the considered kernel density estimators based on the $L^1$-distance (same coloring).
		}\label{fig:dtm dsicrimination}
	\end{figure}

	To conclude this section, we illustrate the influence of the choice of $m$. For this purpose, we repeat the above procedure with $n=10,000$ and $m=0.2,0.4,0.6,0.8$ (this means that we can use the alternative representation of $\dtm$ in \eqref{eq:emp dtm v2} with $k=2000,4000,6000,8000$). The resulting kernel density estimators are displayed in the upper row and the corresponding average clustering in the bottom row of Figure \ref{fig:influence of m} (same coloring as previously). As we consider the transformation of $\mu_{\mathcal{Y}_1}$ into $\mu_{\mathcal{Y}_7}$ along a 2-Wasserstein geodesic, it is intuitive that choosing $m$ too small is not informative in this setting (the goal is to distinguish between the whole spaces). Indeed, this is exactly, what we observe. For $m= 0.2$ the kernel density estimators strongly resemble each other and in particular the Euclidean metric measure spaces $\mathcal{Y}_1$ and $\mathcal{Y}_2$ are hardly distinguishable (see the corresponding dendrogram in the lower row of Figure \ref{fig:influence of m}). For $m\geq 0.4$ the kernel density estimators are better separated and the corresponding dendrograms highlight that it is possible to discriminate between the spaces $\mathcal{Y}_i$ based on the kernel density estimators $\kdeyvar{i}{m}$, $i=1,\dots,7$ and $m=0.4,0.6,0.8$. It is noteworthy that although the form of the kernel density estimators drastically changes between $m=0.4$ and $m=1$, the quality of the corresponding clustering only increases slightly with increasing $m$.
	\begin{figure}
		\centering
		\begin{minipage}[c]{0.22\textwidth}
			\centering
			\scriptsize{\qquad$n=10000, m=0.2$}
			\includegraphics[width = \textwidth,height=\textwidth]{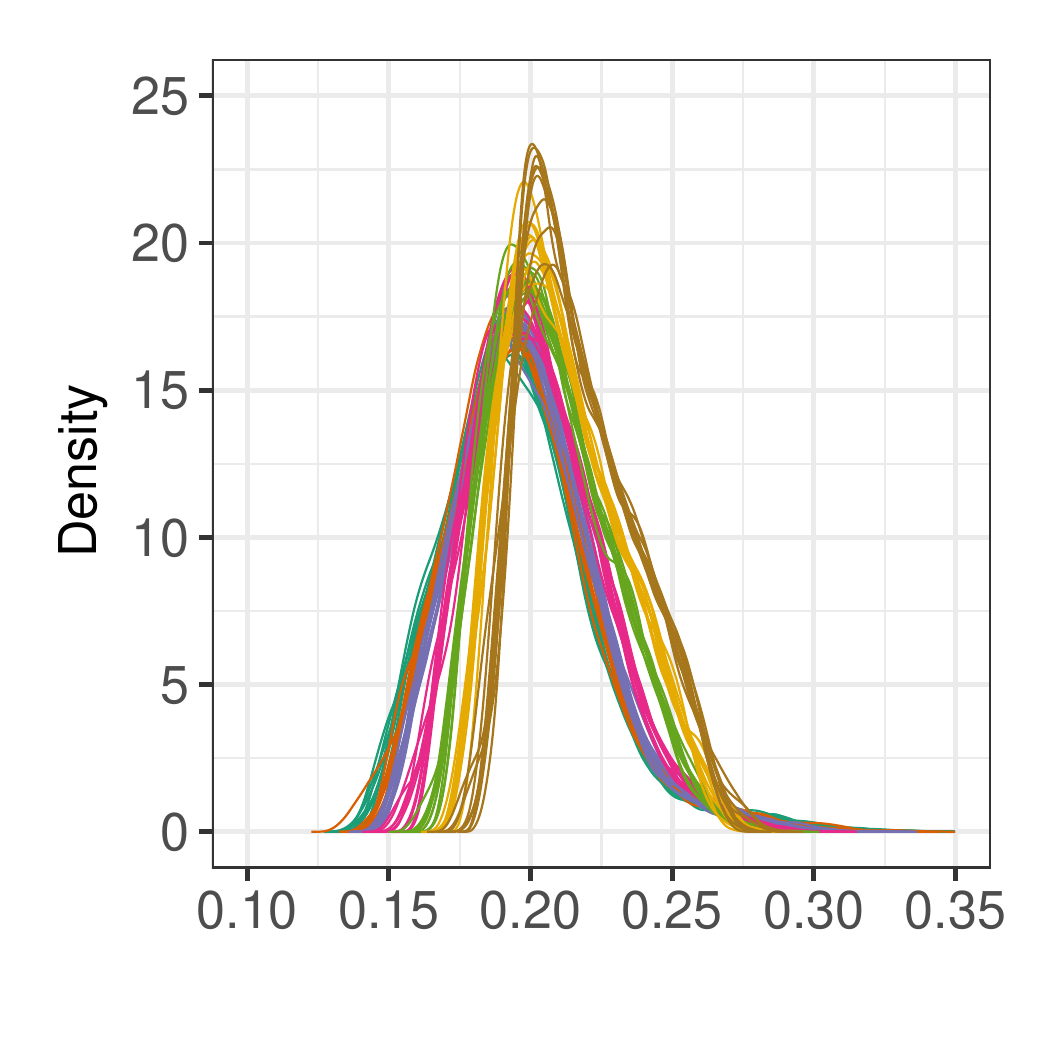}\\
		\end{minipage}
		\begin{minipage}[c]{0.22\textwidth}
			\centering
			\scriptsize{\qquad$n=10000, m=0.4$}
			\includegraphics[width = \textwidth,height=\textwidth]{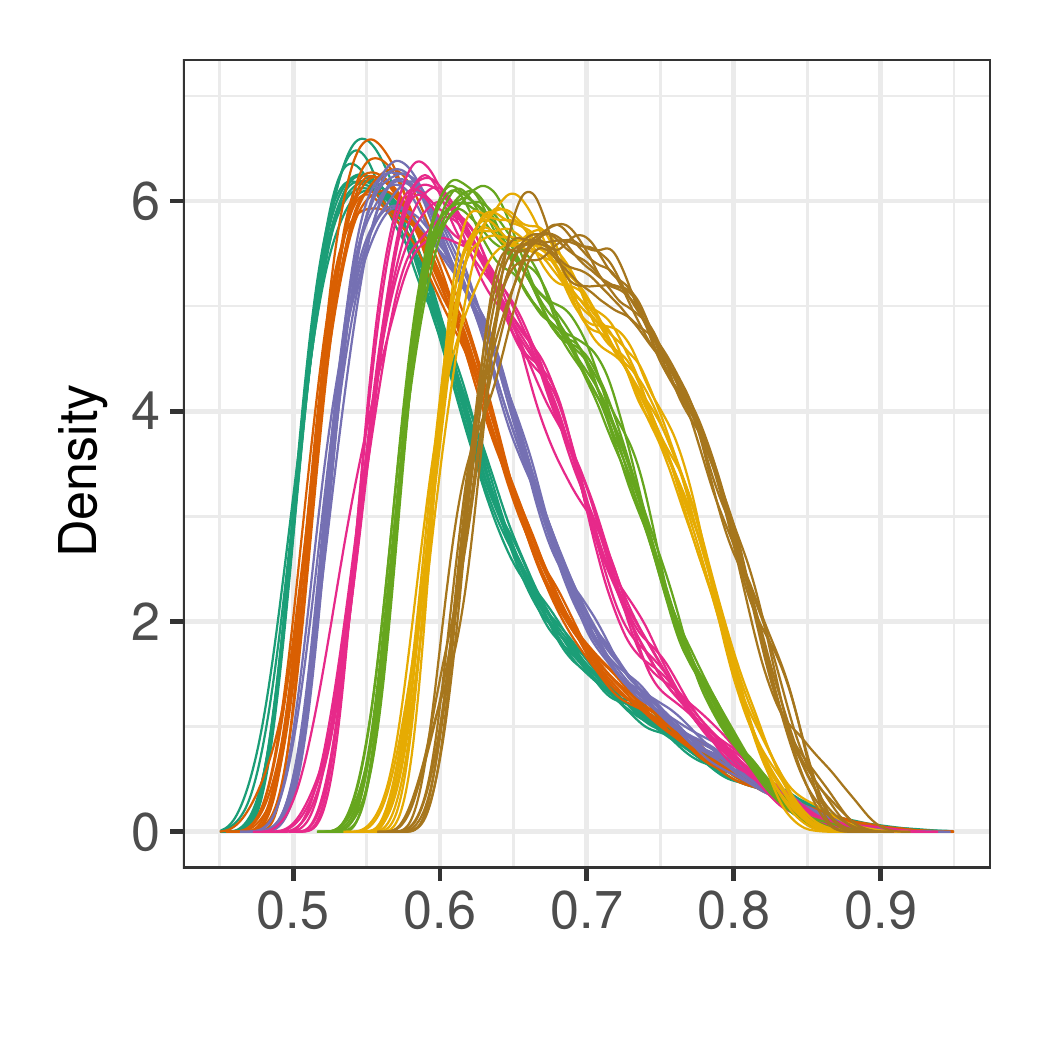}\\
		\end{minipage}
		\begin{minipage}[8]{0.22\textwidth}
			\centering
			\scriptsize{\qquad$n=10000, m=0.6$}
			\includegraphics[width = \textwidth,height=\textwidth]{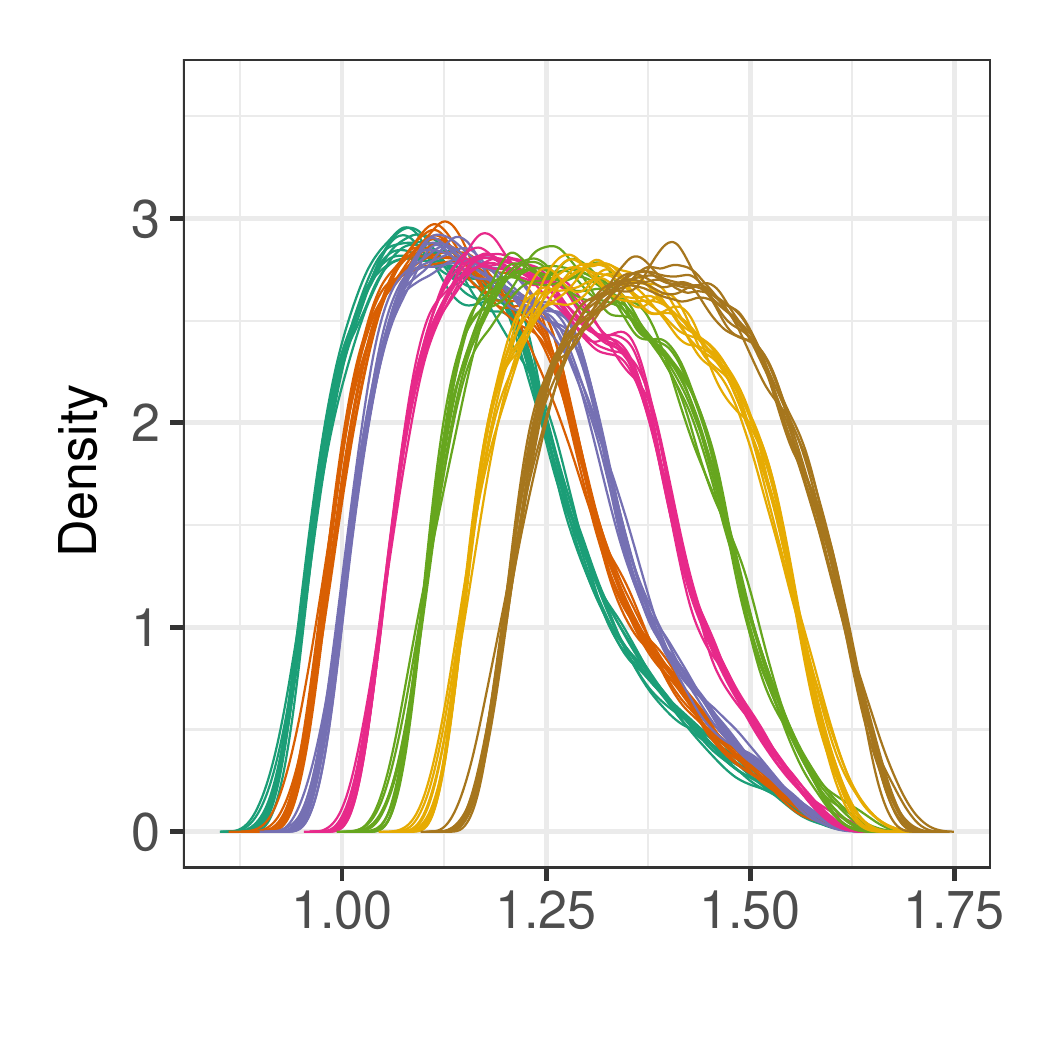}
		\end{minipage}
		\begin{minipage}[8]{0.22\textwidth}
			\centering
			\scriptsize{\qquad$n=10000, m=0.8$}
			\includegraphics[width = \textwidth,height=\textwidth]{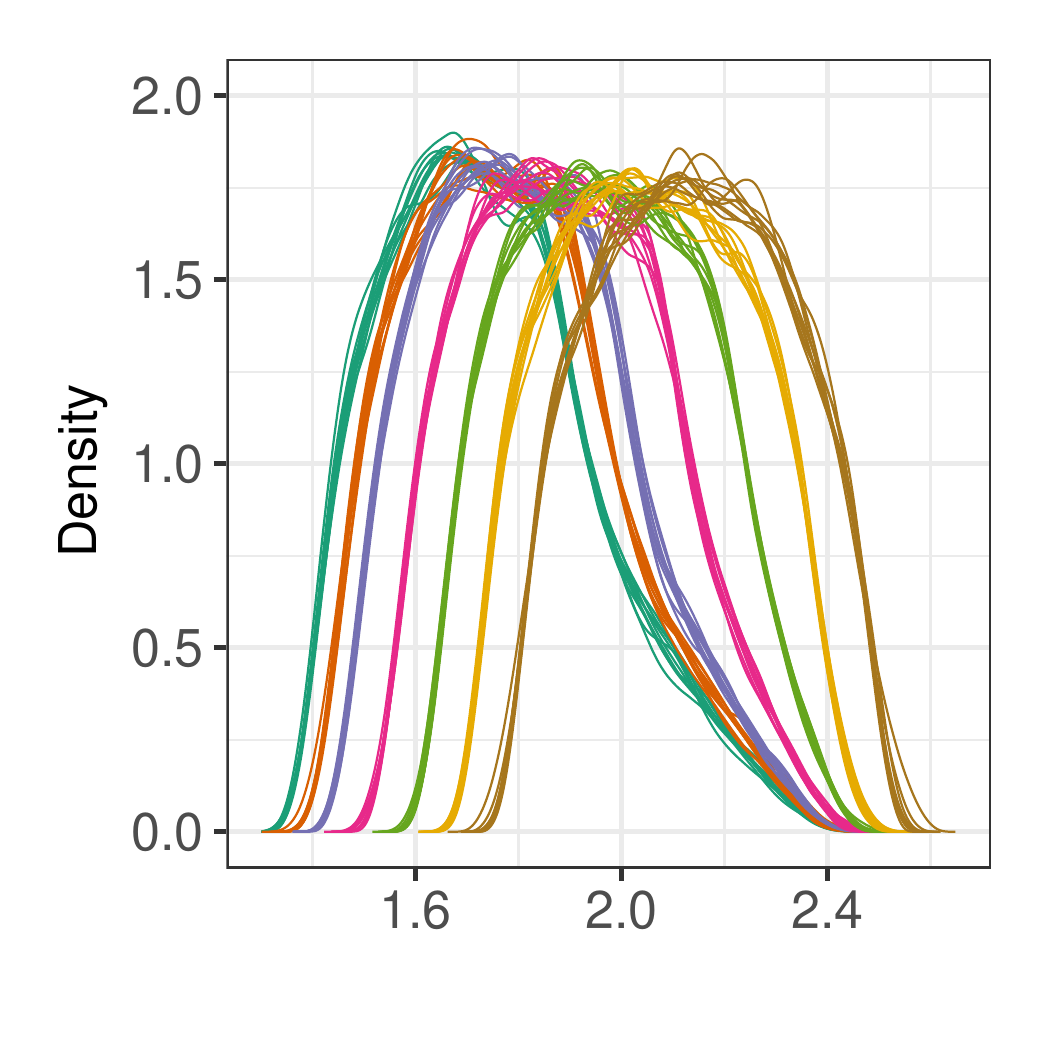}\\
		\end{minipage}
		\begin{minipage}[c]{0.04\textwidth}
			\hfill\hspace{0.1cm}\includegraphics[height=0.13\textheight]{Graphics/colorbar6.pdf}
			\vspace{0.5cm}
		\end{minipage}
		\begin{minipage}[c]{0.22\textwidth}
			\centering
			\includegraphics[width = \textwidth,height=\textwidth]{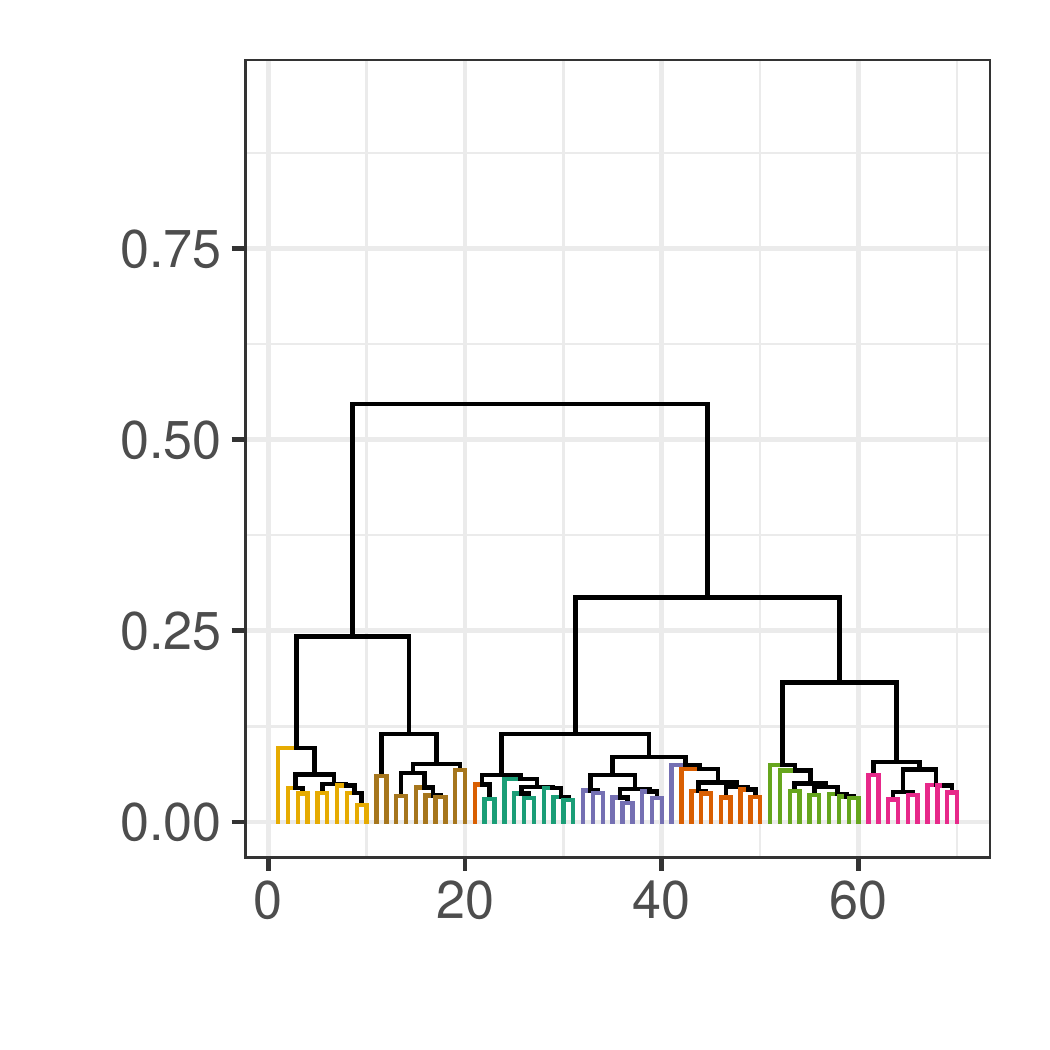}
		\end{minipage}
		\begin{minipage}[c]{0.22\textwidth}
			\centering
			\includegraphics[width = \textwidth,height=\textwidth]{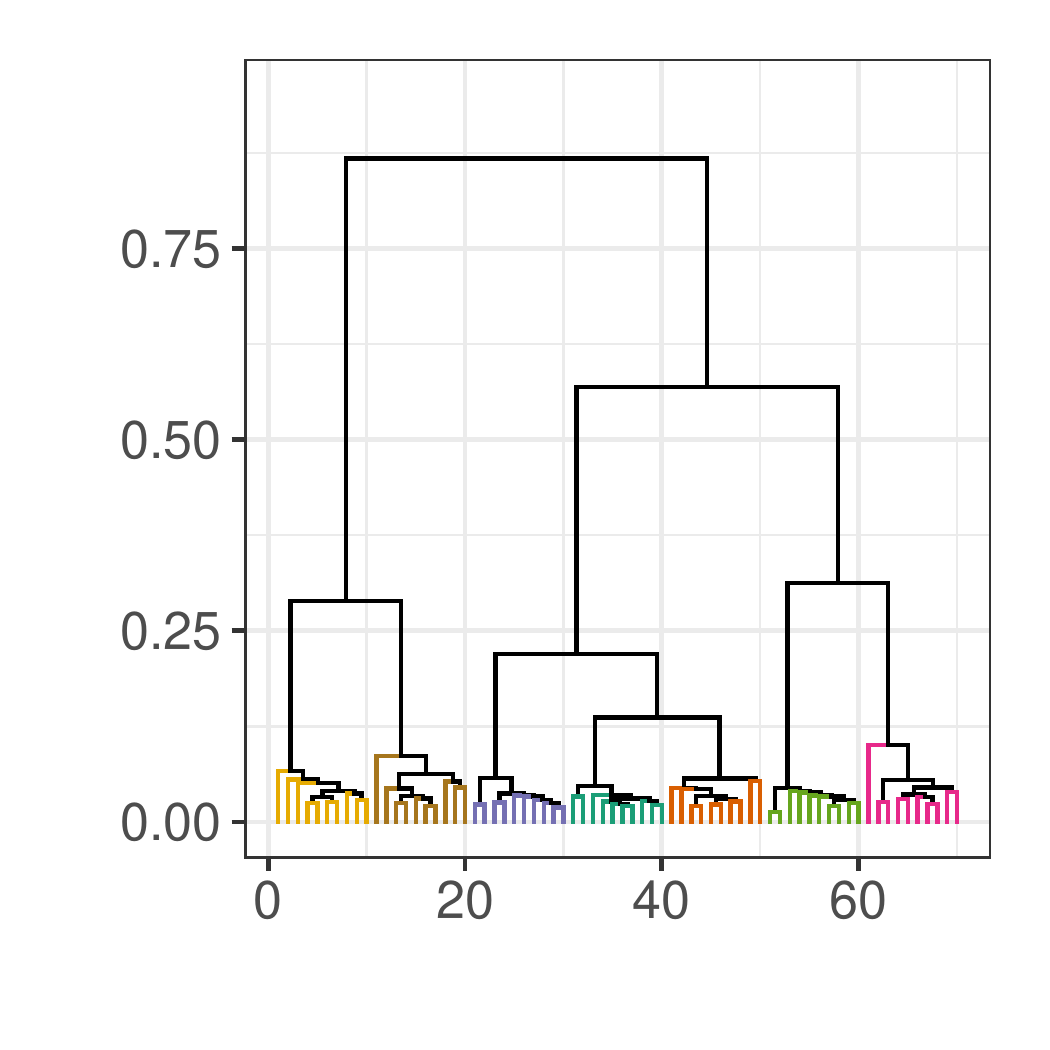}
		\end{minipage}
		\begin{minipage}[8]{0.22\textwidth}
			\centering
			\includegraphics[width = \textwidth,height=\textwidth]{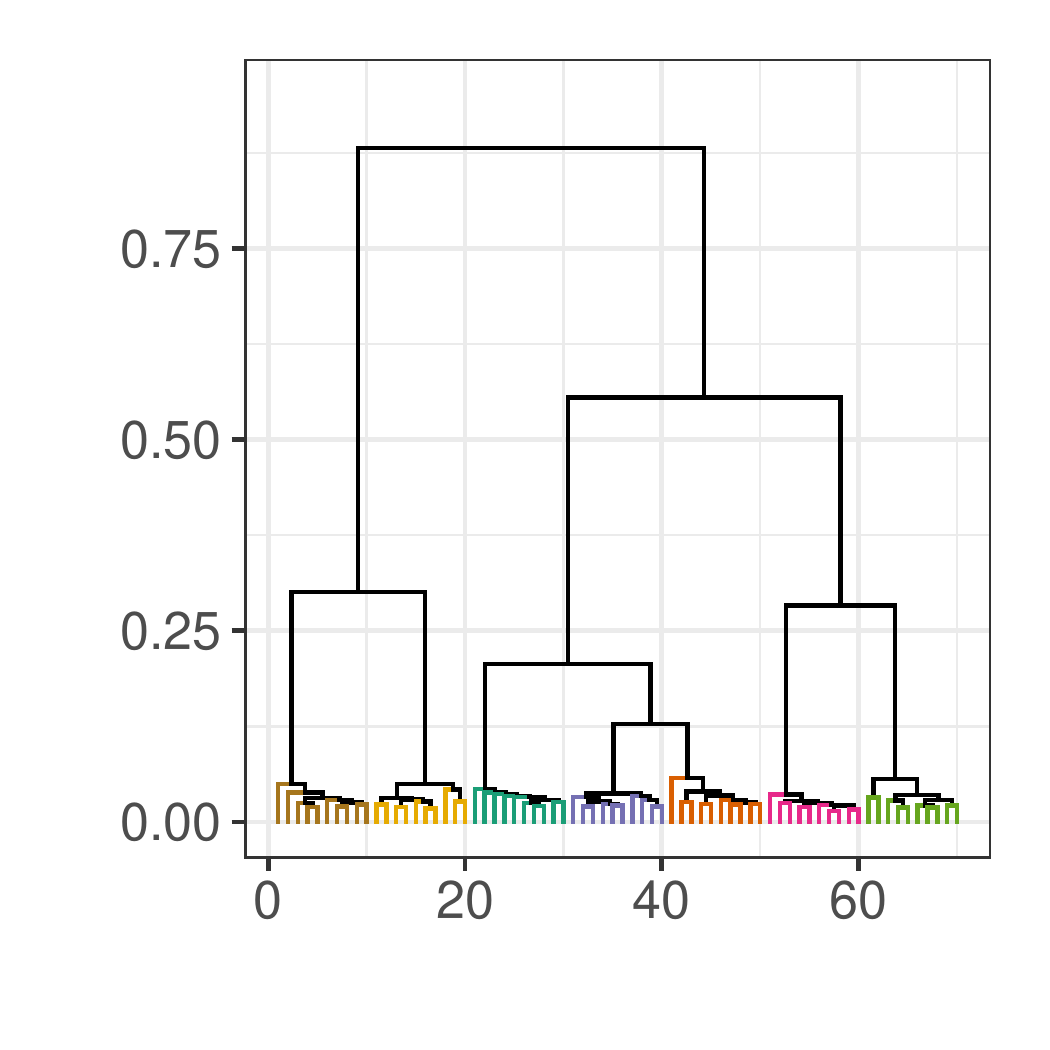}
		\end{minipage}
		\begin{minipage}[8]{0.22\textwidth}
			\centering
			\includegraphics[width = \textwidth,height=\textwidth]{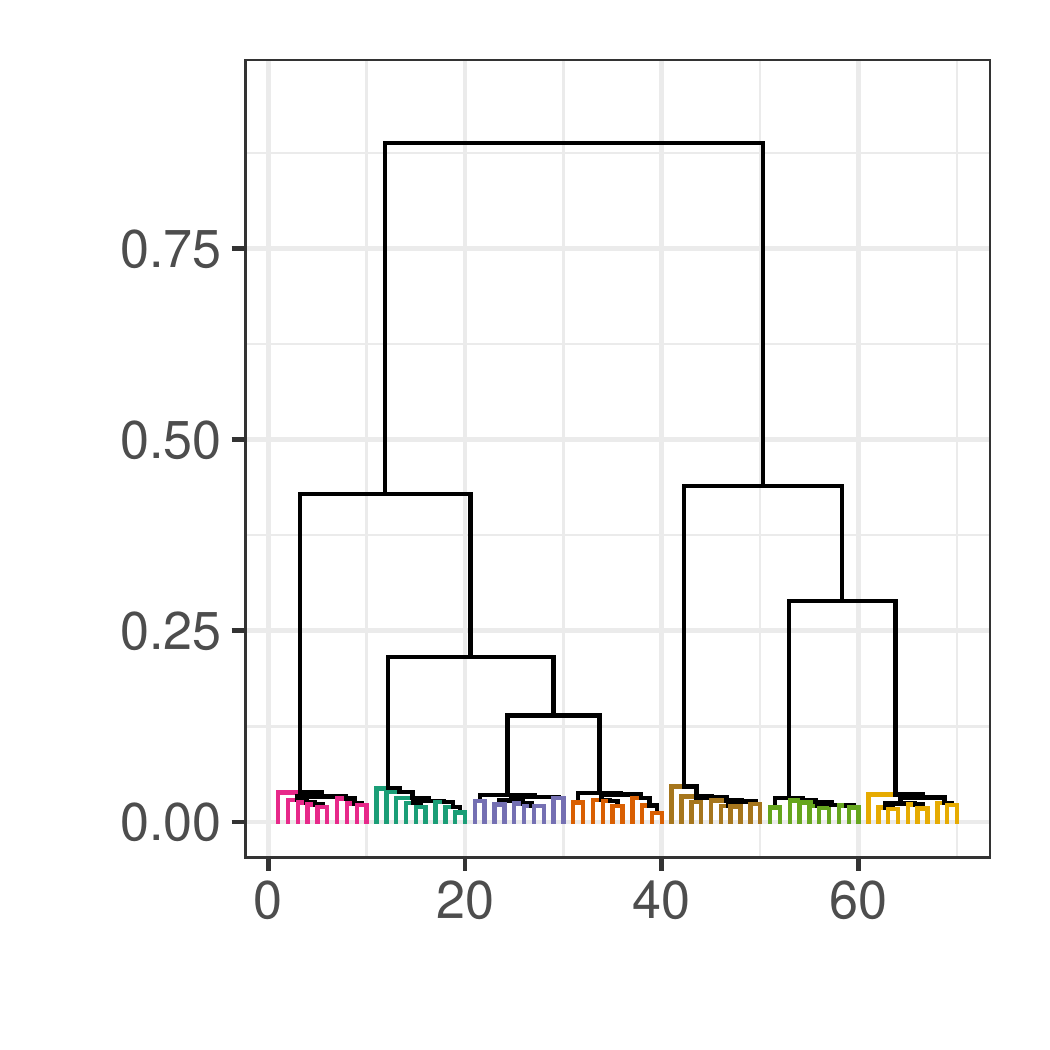}
		\end{minipage}
		\begin{minipage}[c]{0.04\textwidth}
			\hfill
		\end{minipage}
		\caption{\textbf{The influence of $m$:} Upper row: Ten realizations of the kernel density estimators $\kdeyvar{1}{m}$ (blue-green), $\kdeyvar{2}{m}$ (orange), $\kdeyvar{3}{m}$ (blue), $\kdeyvar{4}{m}$ (pink), $\kdeyvar{5}{m}$ (green), $\kdeyvar{6}{m}$ (yellow) and $\kdeyvar{7}{m}$ (brown) for $n=10000$ and $m=0.2,0.4,0.6,0.8$ (from left to right). Lower row: The results of an average linkage clustering of the considered kernel density estimators based on the $L^1$-distance (same coloring). 
		}\label{fig:influence of m}
	\end{figure}
	
	\section{Chromatin Loop Analysis}\label{sec:CLA}
	
	In this section, we will highlight how to use the DTM-density-transformation for chromantin loop analysis. First, we briefly recall some important facts about chromatin fibers, state the goal of this analysis and precisely describe the data used here. \\
	
	For human beings, chromosomes are essential parts of cell nuclei. They carry the genetic information important for heredity transmission and consist of chromatin fibers. Learning the topological 3D structure of the chromatin fiber in cell nuclei is important for a better understanding of the human genome. As discussed in Section \ref{subsec:application}, TADs are self-interacting genomic regions, which
	are often associated with loops in the chromatin fibers. 
	These domains have been estimated to the range of 100-300 nm \citep{NueblerE6697}. Hi-C data \citep{HiC} allow to construct spatial proximity maps of the human genome and are often used to analyze genome-wide chromatin organization and to identify TADs. However, spatial size and form, and how frequently chromatin loops and domains
	exist in single cells, cannot directly be answered based on Hi-C data, whereas in 3D visualization of chromosomal regions via SMLM with a sufficiently high resolution, this information might be more easily accessible \citep{Paris}. Therefore, in the above reference, such an approach is considered, in which two groups of images of chromatin fibers were produced: Chromatin with supposedly fully intact loop structures and chromatin, which had been treated with auxin prior to imaging. Auxin is known to cause a degrading of the loops. Therefore, in the second set of images, the loops are expected to be mostly dissolved. The obtained resolution in these images was of the order of 150 nm, i.e., below the diffraction limit and comparable to the typical sizes of TADs. This means that the analysis of chromatin loops based on these images is tractable but difficult as we will not see detailed loops when zooming in.
	
	In this paper we analyse simulated SMLM data of chromatin fibers that mimic the chromatin structure with loops as local features and compare them to simulated data that mimic the progressive degradation of loop structures in five steps. 
	The simulated structures mimic the first chromosome (of 23 in total) of the human genome, which is the longest with approximately 249 megabases (Mb, corresponding to 249,000,000 nucleotides).
	Each step corresponds to a loop density with a different parameter, which we denote by $c$.
	The value of $c$ is the number of loops per megabase. 
	A value of $c=25$ corresponds to a high loop density with 2490 loops in total and corresponds to the setting without the application of auxin. Values of $c=10,6,4,2$ correspond to decreasing states of resolved loops (1494, 996 and 498 loops) and $c=0$ encodes the fully resolved state.
	These simulated images provide a controlled setting in which we can investigate the applicability of our methods and in which we can explore how small a difference in loop density our method can still pick up and when it starts to break down. Here, we only consider classification into the different conditions based on the estimated DTM density. While it is clear from the results described below that information on loop size and frequency is encoded in these densities, a quantification of these parameters requires a deeper study of the proposed methods and is beyond the scope of this paper.\\
	In our study, we consider 102 synthetic, noisy samples of size 49800 of 6 different loop densities each and denote the corresponding samples as $\X_{i,c}$, $c=25,10,6,4,2,0$, $1\leq i\leq 102$. These samples are created by first discretizing the chromatin structure such that the distance between two points along the chromatin structure corresponds to $45$ nm. Then, we add independent, centered Gaussian errors with covariance matrix
	$$\Sigma=\begin{pmatrix}
		45 & 0 & 0\\
		0 & 45 & 0\\
		0 & 0 & 90
	\end{pmatrix}$$
	to each point (see Figure \ref{fig:chromatinloopsamp} for an illustration of data obtained in this fashion). This high level of noise is chosen to match the experimental data obtained in \citet{Paris}. Throughout the following, we consider the data on a scale of 1:45. We stress once again that the goal of our analysis is to distinguish between the respective loop conditions and not between chromatin fibers from which the points are sampled (the overall form of the chromatin fibers within one condition can be quite different).
	\begin{figure}
		\centering
		\includegraphics[width=0.7\textwidth]{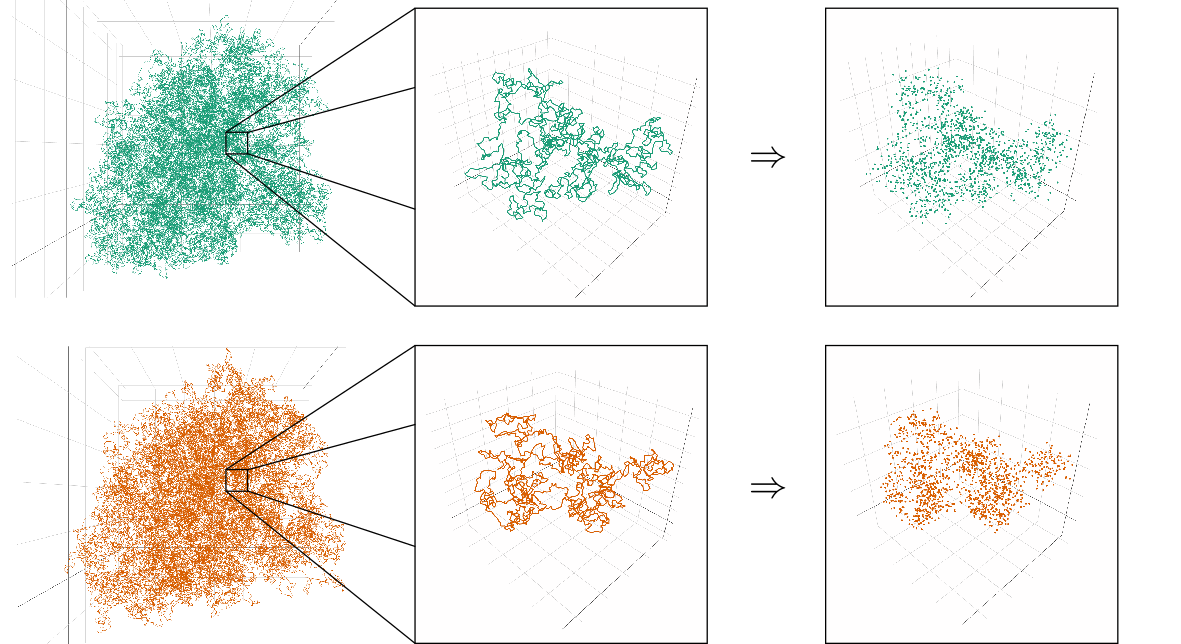}
		\caption{\textbf{Chromatin Loops:} Upper row: Illustration of a chromatin structure and the corresponding sample with loop density $c=25$. Lower row: Illustration of the same chromatin and the corresponding sample with loop density $c=10$.
		}\label{fig:chromatinloopsamp}
	\end{figure}
	We demonstrate in the following that the corresponding DTM-signatures, or more precisely the corresponding kernel density estimators $\kdexvar{i,c}{m}$, $1\leq i\leq 102$, ($m$ chosen suitably small) represent a useful transformation of the data that allows discrimination between the different loop  densities, while disregarding the overall shape of the chromatin fiber. To this end, we follow the strategy proposed in Section \ref{subsec:application} and calculate $\Delta_{i,c}=\{\empdtm(X_{j,i,c}):X_{j,i,c}\in \mathcal{X}_{i,c}\}$ for $1\leq i\leq 102$, $c\in\{25,10,6,4,2,0\}$ and $m\in\{{1}/{9960},{1}/{4980}, {1}/{1245}\}$. These particular choices of $m$ entail that in order to calculate $\empdtm(X_{j,i,c})$ we need to take the mean over the distances to the $k=5,10,40$ nearest neighbors of $X_{j,i,c}\in \mathcal{X}_{i,c}$ (recall the representation of $\empdtm$ in \eqref{eq:emp dtm v2}). 
	%
	\begin{remark}[Choice of the parameters $\bm{m}$ and $\bm{k}$]\label{rem:sampling}
		It is important to note that the practical choice of the parameter $m$ (resp. the parameter $k$) strongly depends on the mathematical model and the sampling mechanism used. In many mathematical models, increasing the sample size corresponds to an increase in sampling density, while the object under consideration remains fixed. In this case, DTM-densities with the same mass parameter $m$, which determines the proportion of points considered, are comparable, although the number of nearest neighbors considered varies with $n$. However, this is not necessarily the case in many applications. In our biological example, varying $n$ means changing the size of the structure under consideration since the distance between two observable points along the polymer remains fixed (up to random errors). In such a case, the corresponding DTM-densities are only comparable if the parameter $k$ is fixed (resulting in variations of $m$). Note that our synthetic data sets are all of the same size and therefore both viewpoints are equivalent in this particular case.
	\end{remark}
	We determine $\kdexvar{i,c}{m}$ (Biweight kernel, $h=1.06\min \{s(\Delta_n),\text{IQR}(\Delta_n)/1.34)\}n^{-1/5})$ based on each of the samples $\Delta_{i,c}$ . The resulting kernel density estimators are displayed in Figure \ref{fig:chromatinanalysisI}. Generally, the kernel density estimators based on the different samples with the same loop density strongly resemble each other and it is possible to roughly distinguish between the different values of  $c$. For all values of $m$ considered, the DTM-density estimators based on $\X_{i,25}$, $1\leq i\leq 102$, (here the respective chromatin fibers form many loops) are well separated from the other kernel density estimators and the estimators based on the samples $\Delta_{i,2}$ and $\Delta_{i,0}$ (which correspond to the lowest loop densities considered) are the most similar when comparing the different loop densities. 
	\begin{figure}
		\centering
		\begin{minipage}[c]{0.24\textwidth}
			\centering
			\small{\quad$m={1}/{9960}$}
			\includegraphics[width = \textwidth,height=\textwidth]{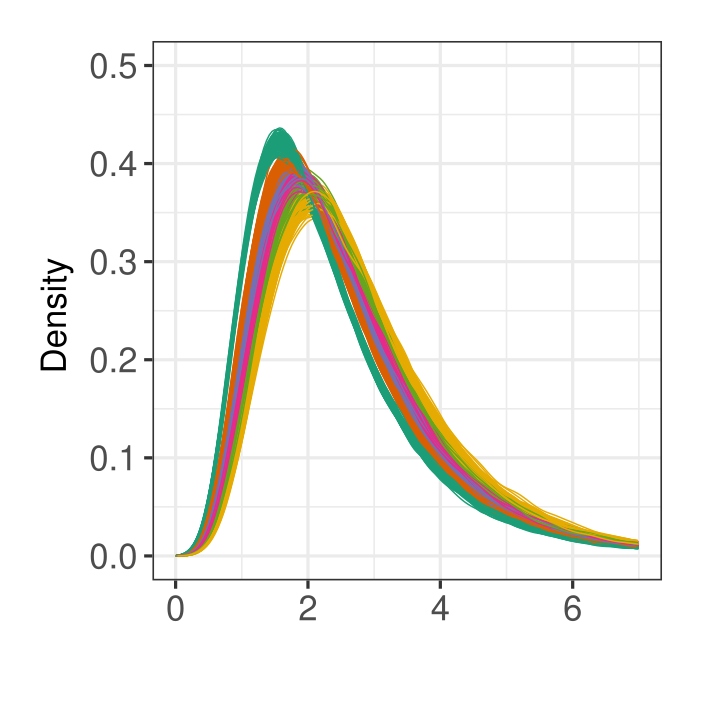}\\
		\end{minipage}
		\begin{minipage}[c]{0.24\textwidth}
			\centering
			\small{\quad$m={1}/{4980}$}
			\includegraphics[width = \textwidth,height=\textwidth]{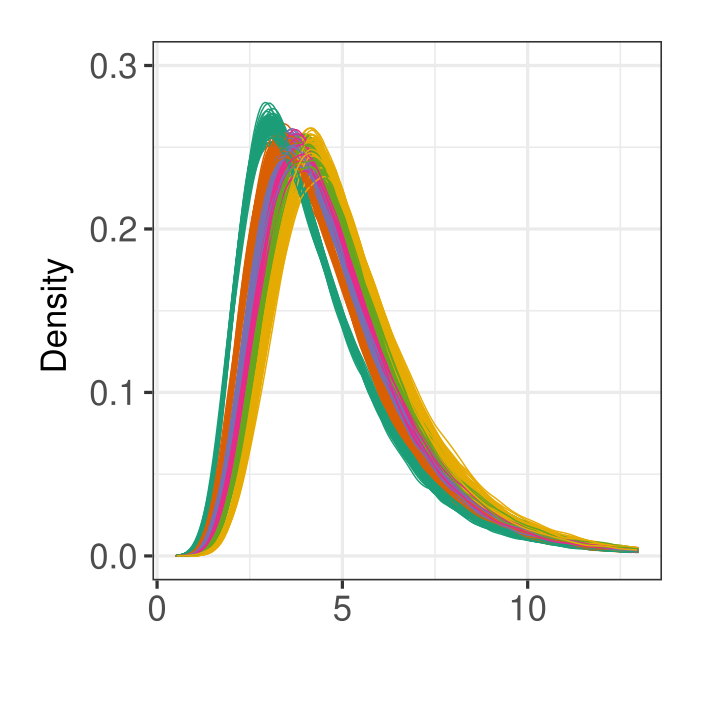}\\
		\end{minipage}
		\begin{minipage}[c]{0.24\textwidth}
			\centering
			\small{\quad$m={1}/{1245}$}
			\includegraphics[width = \textwidth,height=\textwidth]{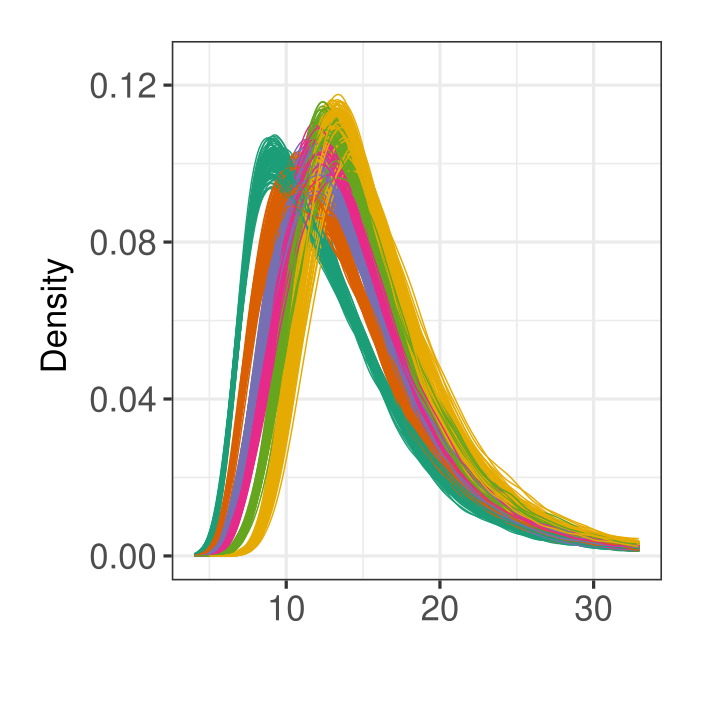}\\
		\end{minipage}
		\begin{minipage}[l]{0.1\textwidth}
			\hspace{0.3cm}
			\includegraphics[height=0.14\textheight]{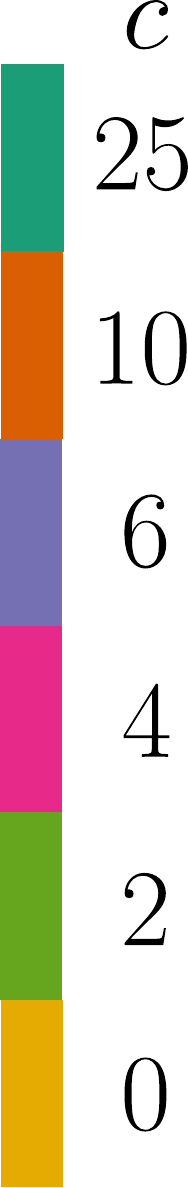}
			\vspace{0.2cm}
		\end{minipage}
		
		\caption{\textbf{Chromatin Loop Analysis I:} Illustration of the DTM-density estimators $\kdexvar{i,c}{m}$, $1\leq i\leq 102$, for $c=25$ (blue-green), $c=10$ (orange), $c=6$ (blue), $c=4$ (pink), $c=2$ (green) and $c=0$ (yellow) and $m\in\{{1}/{9960},{1}/{4980}, {1}/{1245}\}$ (from left to right).
		}\label{fig:chromatinanalysisI}
	\end{figure}
	In order to make a more qualitative comparison between the estimators $\kdexvar{i,c}{m}$, we use the strategy developed in Section \ref{subsec:disc props} and perform an average linkage clustering based on the $L_1$-distance between the estimated densities. For clarity, we restrict ourselves to the comparison of the loop density $c=25$ against $c=10$ as well as $c=2$ against $c=0$ and point out that the comparison between the setting $c=2$ against $c=0$ is very difficult as the loop frequencies are very low. The dendrograms in the upper row of Figure \ref{fig:chromatinloopdendro} illustrate the comparison of $c=25$ and $c=10$. It is remarkable that for each $m$ the correct clusters are obtained. The lower row of Figure \ref{fig:chromatinloopdendro} showcases the dendrograms for the comparison of the estimators $\kdexvar{i,2}{m}$ and $\kdexvar{i,0}{m}$, $1\leq i\leq 102$ and $m\in\{{1}/{9960},{1}/{4980}, {1}/{1245}\}$. For $m\in\{{1}/{9960},{1}/{4980}\}$, we obtain (up to one exception) the correct clusters, although they are much closer (w.r.t. the $L^1$-distance) than the clusters for the previous comparisons. However, for $m={1}/{1245}$, it is no longer possible to reliably identify two clusters that correspond to $c=2$ and $c=0$. It seems that in this case $m$ is too large to yield a perfect discrimination.\\
	\begin{figure}
		\centering
		\begin{minipage}[c]{0.24\textwidth}
			\centering
			\small{\quad$m={1}/{9960}$}
			\includegraphics[width=\textwidth]{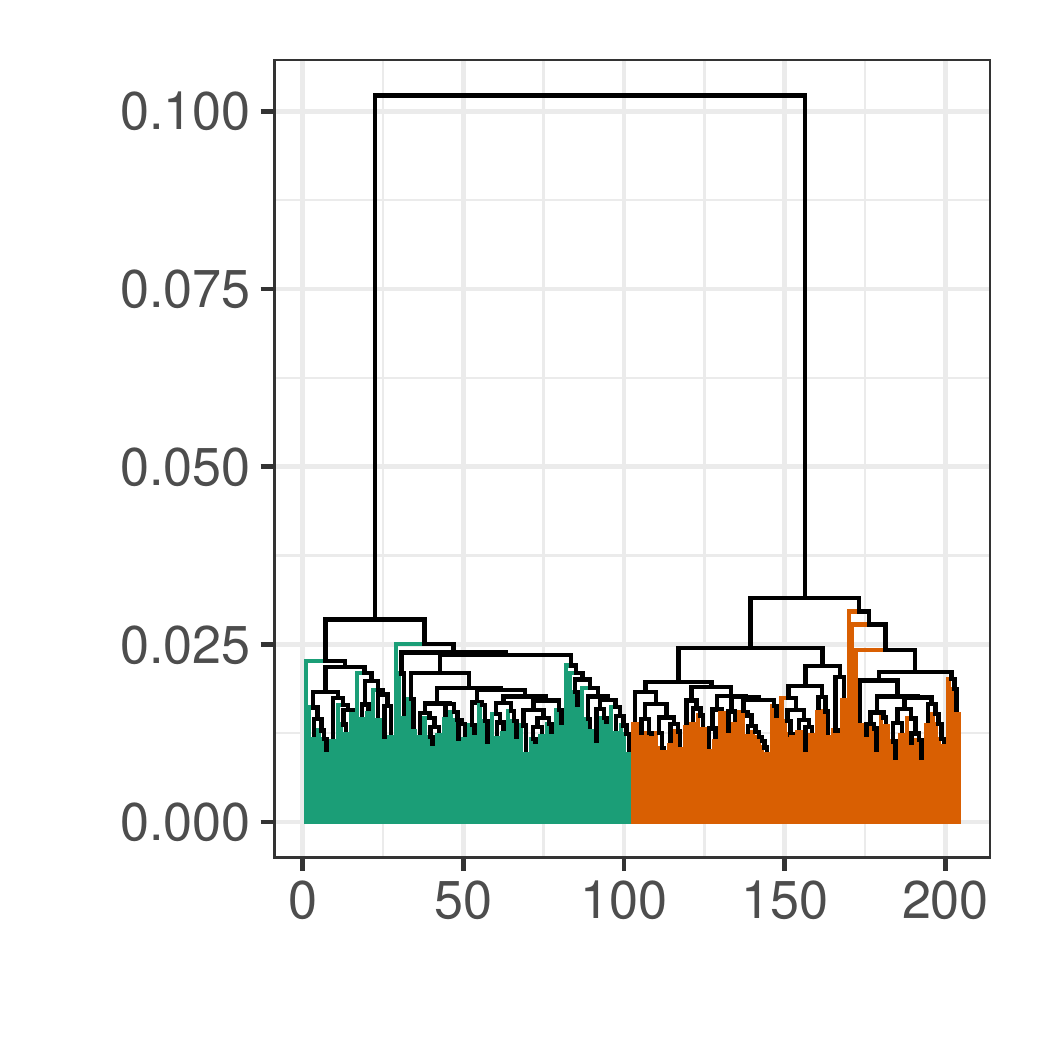}
		\end{minipage}
		\begin{minipage}[c]{0.24\textwidth}
			\centering
			\small{\quad$m={1}/{4980}$}
			\includegraphics[width=\textwidth]{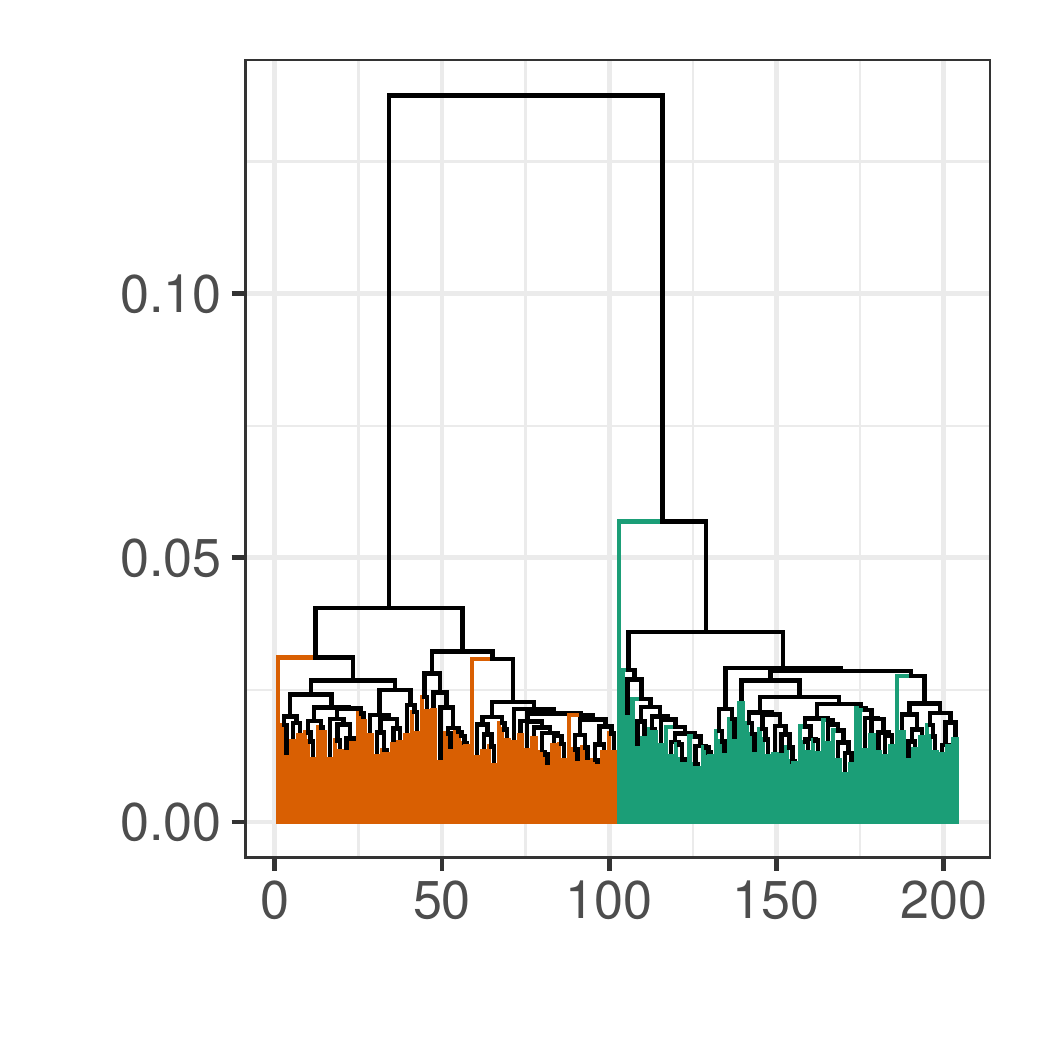}
		\end{minipage}
		\begin{minipage}[c]{0.24\textwidth}
			\centering
			\small{\quad$m={1}/{1245}$}
			\includegraphics[width=\textwidth]{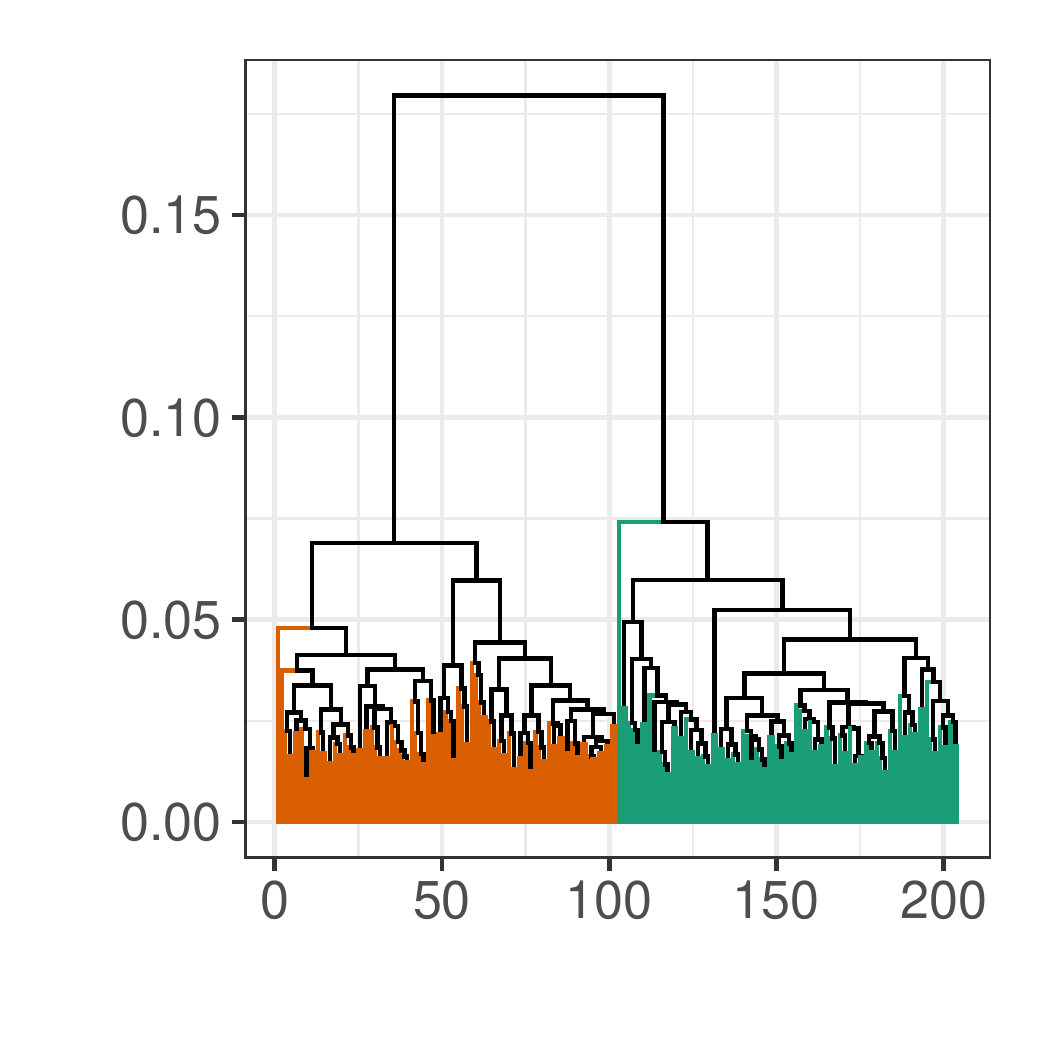}
		\end{minipage}
		\begin{minipage}[l]{0.1\textwidth}
			
			\includegraphics[height=0.12\textheight]{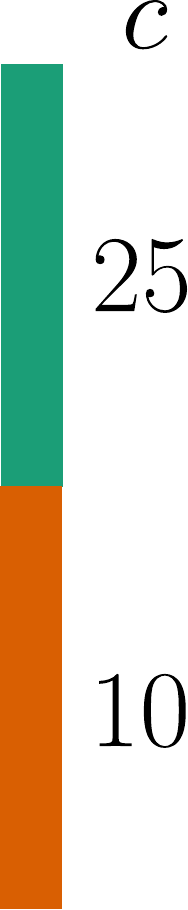}
			\vspace{0.4cm}
		\end{minipage}\\
		\begin{minipage}[c]{0.24\textwidth}
			\centering
			\includegraphics[width=\textwidth]{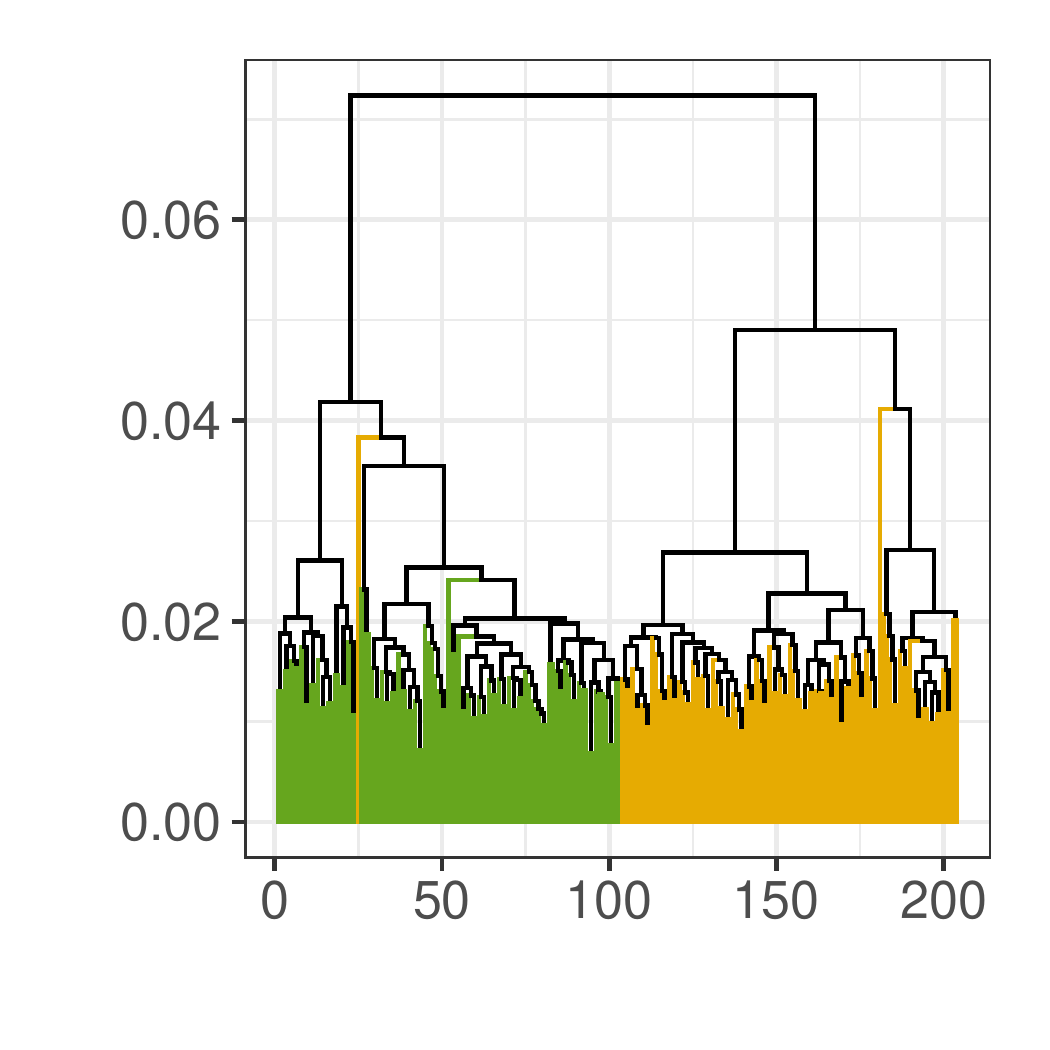}
		\end{minipage}	\begin{minipage}[c]{0.24\textwidth}
			\centering
			\includegraphics[width=\textwidth]{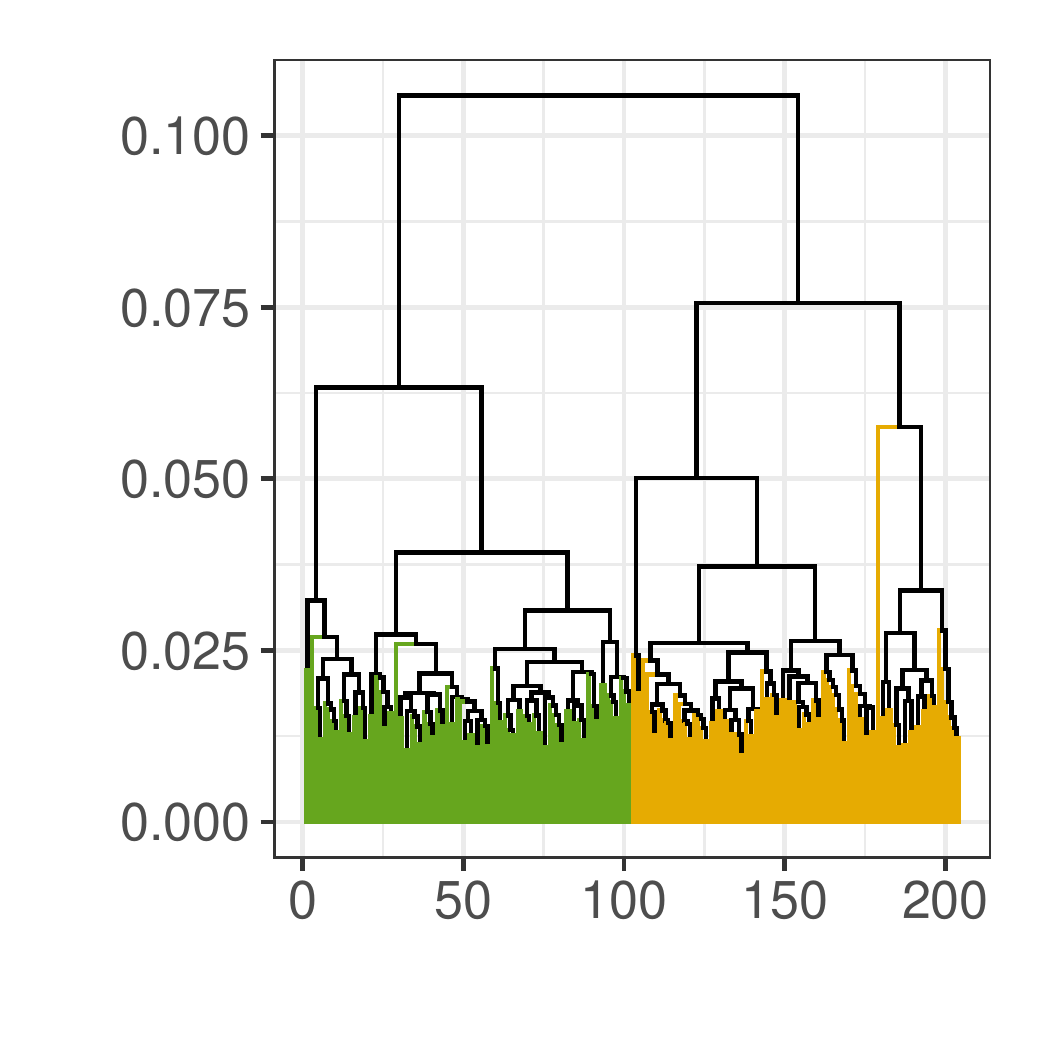}
		\end{minipage}	\begin{minipage}[c]{0.24\textwidth}
			\centering
			\includegraphics[width=\textwidth]{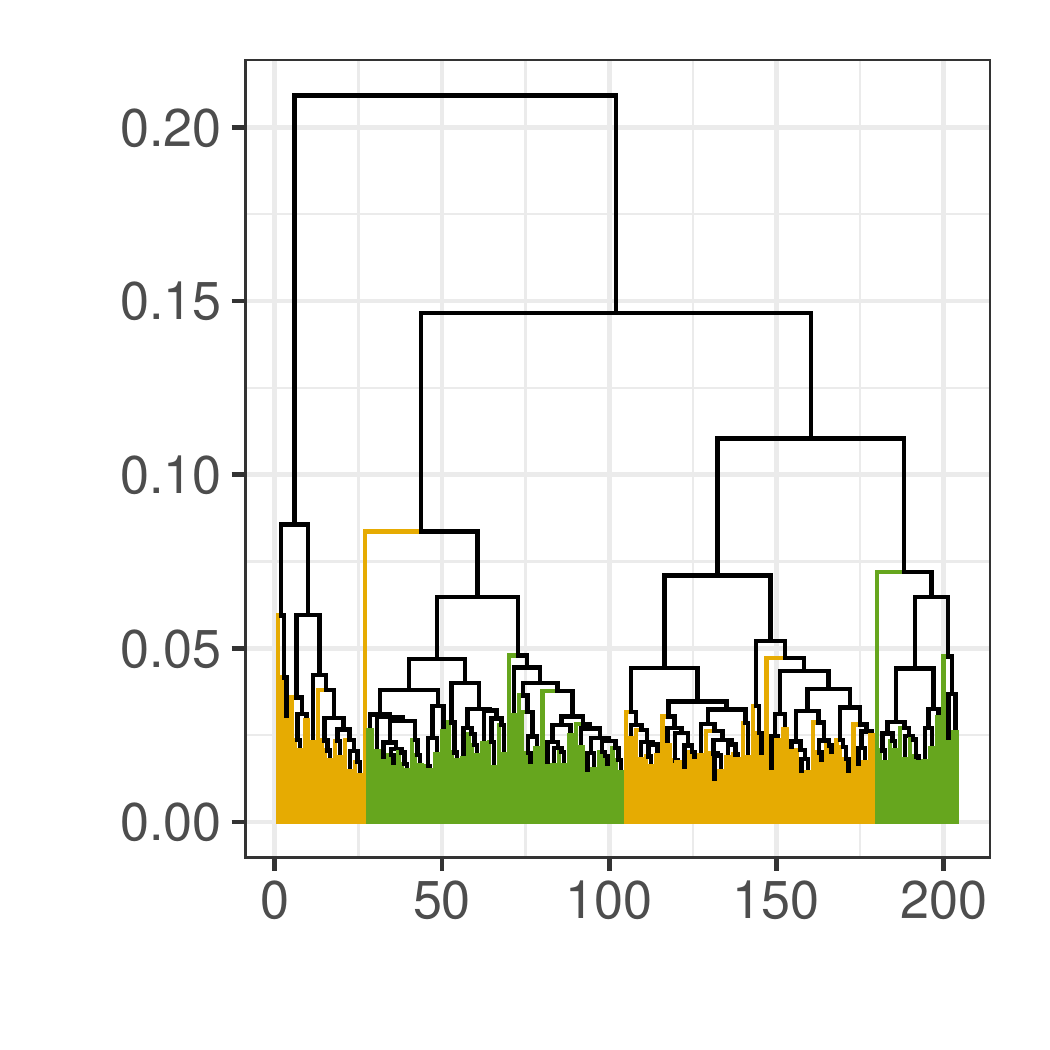}
		\end{minipage}
		\begin{minipage}[l]{0.1\textwidth}
			
			\includegraphics[height=0.12\textheight]{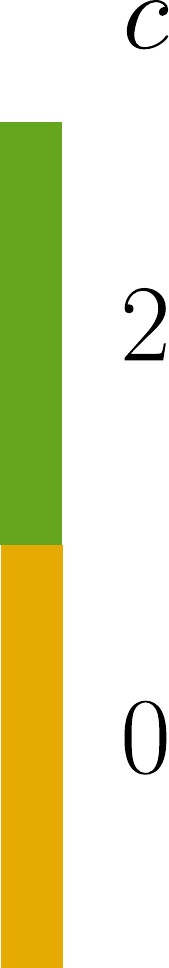}
			\vspace{0.9cm}
		\end{minipage}
		\caption{\textbf{Chromatin Analysis II:} Upper row: The results of an average linkage clustering of the kernel density estimators $\kdexvar{i,25}{m}$ (blue-green) and $\kdexvar{i,10}{m}$ (orange), $1\leq i\leq 102$, for $m\in\{{1}/{9960},{1}/{4980}, {1}/{1245}\}$ (from left to right) based on the $L^1$-distance. Lower row: The results of an average linkage clustering of the kernel density estimators $\kdexvar{i,2}{m}$ (green) and $\kdexvar{i,0}{m}$ (yellow), $1\leq i\leq 102$, for $m\in\{{1}/{9960},{1}/{4980}, {1}/{1245}\}$ (from left to right) based on the $L^1$-distance.
		}\label{fig:chromatinloopdendro}
	\end{figure}
	
	Furthermore, we investigate whether classification based on the DTM-density estimates $\kdexvar{i,c}{m}$ is possible. Here, we restrict ourselves once again to the comparison of $c=25$ with $c=10$ as well as of $c=2$ with $c=0$. For each comparison, we randomly select $5\%/10\%$ (rounded up) of the density estimates for each considered loop density and classify the remaining ones according to the  majority of the labels of their $k=1,3,5$ nearest neighbors in the randomly selected sample. We repeat this procedure for both comparisons 10,000 times and report the relative number of misclassifications in Table \ref{tab:chromatinloopKNN}. The upper row of said table highlights that in the comparison of $c=25$ and $c=10$ the DTM-density estimates are always classified correctly. Things change in the comparison of $c=2$ with $c=0$. While for all $m$ at least $90\%$ of the classifications are correct, there is a noticeable difference between the individual values of $m$. We observe that $m={1}/{4980}$ yields by far the best performance in this setting. It is clear that the loop distributions of the respective chromatin fibers for these two values of $c$ are extremely similar (the chromatin admits few to no loops). Hence, choosing $m$ too large incorporates too much global (non-loop) structure and makes it difficult to discriminate between these two loop densities. On the other hand, choosing $m$ too small seems to incorporate too little structure.\\
	
	To conclude this section, we find that it is possible for a suitable choice of $m$ to clearly distinguish between the different loop densities based on the DTM-density estimators $\kdexvar{i,c}{m}$. We have illustrated that these estimators yield a good summary of the data and can be used to approach the (already quite difficult) problem of chromatin loop analysis for noisy synthetic data. 
	
	\begin{table}[]
		\centering
		\begin{minipage}[c]{0.31\textwidth}
			\begin{center}
				\footnotesize{\begin{tabular}{ |c|c|c|c| } 
						\hline
						& $k=1$ & $k=3$ & $k=5$ \\ \hline
						5\%&0.000 & 0.000 & 0.000 \\ 
						10\% &0.000 & 0.000 & 0.000 \\ 
						\hline
					\end{tabular}
				}
			\end{center}
		\end{minipage}
		\begin{minipage}[c]{0.31\textwidth}
			\begin{center}
				\footnotesize{\begin{tabular}{ |c|c|c|c| } 
						\hline
						& $k=1$ & $k=3$ & $k=5 $\\ \hline
						5\%&0.000 & 0.000 & 0.000 \\ 
						10\% &0.000 & 0.000 & 0.000 \\ 
						\hline
					\end{tabular}
				}
			\end{center}
		\end{minipage}
		\begin{minipage}[c]{0.31\textwidth}
			\begin{center}
				\footnotesize{\begin{tabular}{ |c|c|c|c| } 
						\hline
						& $k=1$ & $k=3$ & $k=5 $\\ \hline
						5\%&0.000 & 0.000 & 0.000 \\ 
						10\% &0.000 & 0.000 & 0.000 \\ 
						\hline
					\end{tabular}
				}
			\end{center}
		\end{minipage}
		\begin{minipage}[c]{0.31\textwidth}
			\begin{center}
				\footnotesize{\begin{tabular}{ |c|c|c|c| } 
						\hline
						& $k=1$ & $k=3$ & $k=5$ \\ \hline
						5\%&0.029 & 0.049& 0.069 \\ 
						10\% &0.020 &0.033& 0.043 \\ 
						\hline
					\end{tabular}
				}
			\end{center}
		\end{minipage}	\begin{minipage}[c]{0.31\textwidth}
			\begin{center}
				\footnotesize{\begin{tabular}{ |c|c|c|c| } 
						\hline
						& $k=1$ & $k=3$ & $k=5 $\\ \hline
						5\%&0.012 &0.019& 0.039 \\ 
						10\% &0.001 &0.007 &0.012 \\ 
						\hline
					\end{tabular}
				}
			\end{center}
		\end{minipage}
		\begin{minipage}[c]{0.31\textwidth}
			\begin{center}
				\footnotesize{
					\begin{tabular}{ |c|c|c|c| } 
						\hline
						& $k=1$ & $k=3$ & $k=5$ \\ \hline
						5\%&0.029& 0.064& 0.100 \\ 
						10\% &0.008& 0.021& 0.025 \\ 
						\hline
					\end{tabular}
				}
			\end{center}
		\end{minipage}
		
		\caption{\textbf{Chromatin Analysis III:} Upper row: The relative number of missclassifications of a $k$-nearest neighbor classification (w.r.t. the $L^1$-distance) based on the kernel density estimators $\kdexvar{i,25}{m}$ and $\kdexvar{i,10}{m}$, $1\leq i\leq 102$, for $m\in\{{1}/{9960},{1}/{4980}, {1}/{1245}\}$ (from left to right). Lower row: The relative number of missclassifications of a $k$-nearest neighbor classification (w.r.t. the $L^1$-distance) based on the kernel density estimators $\kdexvar{i,2}{m}$ and $\kdexvar{i,0}{m}$, $1\leq i\leq 102$, for $m\in\{{1}/{9960},{1}/{4980}, {1}/{1245}\}$ (from left to right).
		}\label{tab:chromatinloopKNN}
	\end{table}

	\subsection{Sensitivity Regarding the Choice of \textit{m}}
	The parameter $m$ determining the percentage of nearest neighbors, which are included in the analysis, is set by the experimenter and can therefore be seen as tuning parameter. Naturally, the question arises in which sense and to which extend the proposed methodology is robust regarding the choice of $m$. In this paper, we have considered an example in which spaces differ regarding their large scale characteristics (see Section \ref{subsec:disc props}, spaces $\mathcal{Y}_1\ldots,\mathcal{Y}_7$ in Figure \ref{fig:geodesic}) and the example of chromatin loops in Section \ref{sec:CLA}, where structures differ locally and globally, but only local features are of interest. General information of this sort is typically known to the experimenter prior to data analysis.\\ We have seen in Section \ref{subsec:disc props} that large scale analyses require large values of $m$. There, it is shown that $m=1$ provides very good classification results already for small sample sizes (Figure \ref{fig:dtm dsicrimination}). However, as sample sizes increase, values of $m$ as small as $m=0.2$ also permit perfect classification between the spaces (Figure \ref{fig:influence of m}). \\
	Considering small scale characteristics on the other hand, the parameter $m$ should be chosen small, as indicated by the results presented in Section \ref{sec:CLA}. In the example considered there, robustness with respect to the choice of $m$ is also given, albeit in a smaller range (see Figures \ref{fig:chromatinanalysisI} and \ref{fig:chromatinloopdendro}). The global behaviour of distances ($m=1$) is not discriminative in this case at all as can be seen in Figure \ref{fig:loopsmlarge} (right hand side), where the two extreme conditions $c=0$ and $c=25$ cannot be separated with $m=1$. Yet, the two plots on the left of Figure \ref{fig:loopsmlarge} for $m=0.5$ show that the conditions $c=0$ and $c=25$ can still be decently separated, which is quite impressive given the difficulty of the problem. 
	
	\begin{figure}
		\centering
		\begin{minipage}[c]{0.215\textwidth}
			\centering
			\small\qquad$m=0.5$
			\includegraphics[width=\textwidth]{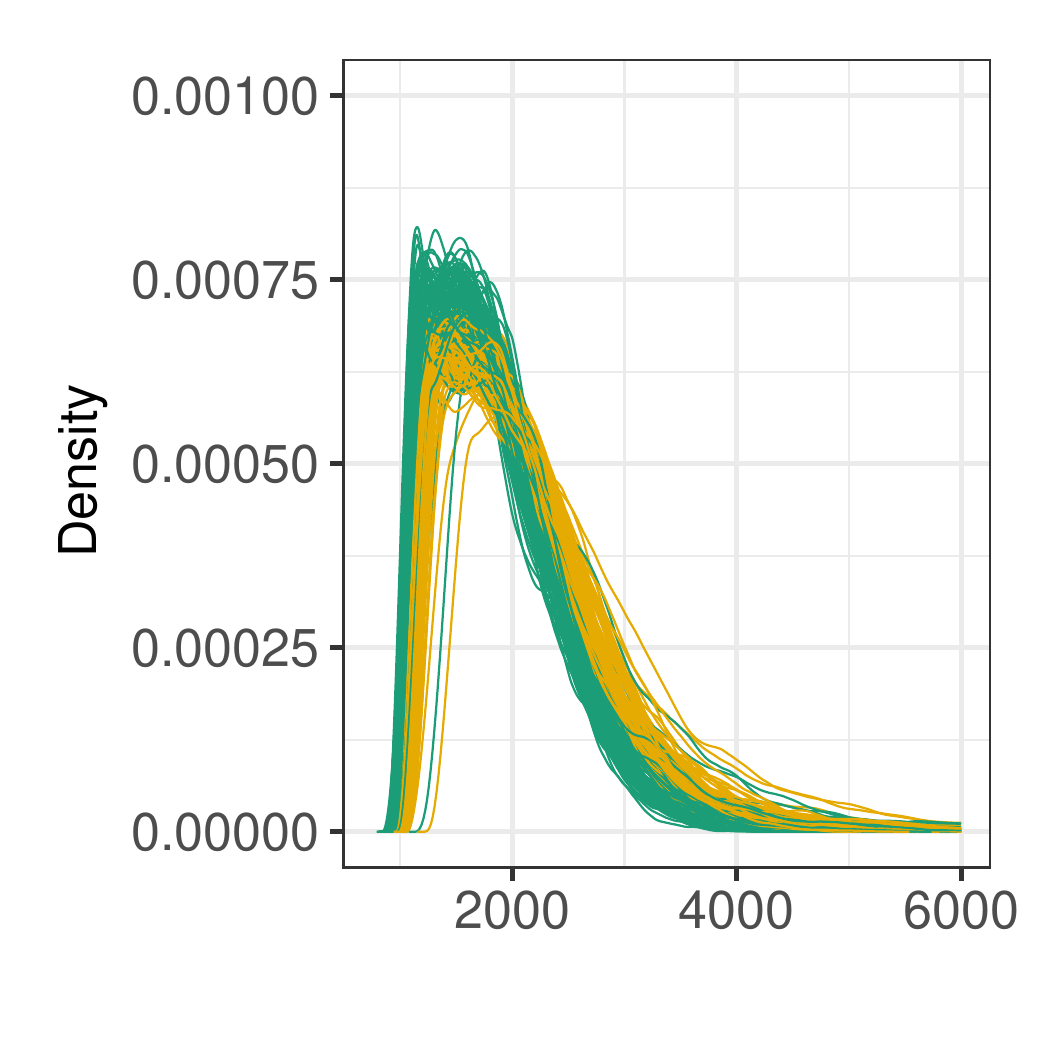}
		\end{minipage}
		\begin{minipage}[c]{0.215\textwidth}
			\centering
			\small\qquad$m=0.5$
			\includegraphics[width=\textwidth]{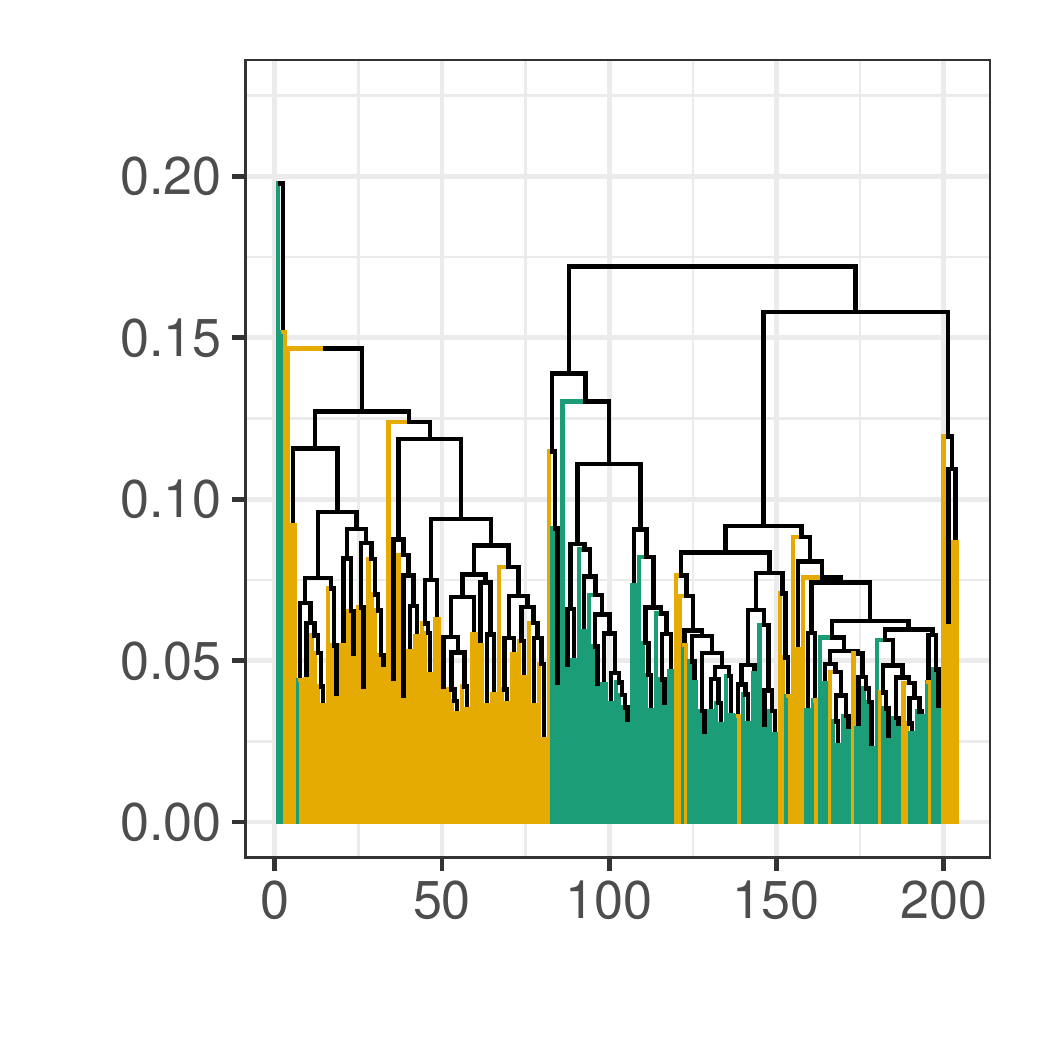}
		\end{minipage}
		\begin{minipage}[c]{0.215\textwidth}
			\centering
			\small\qquad$m=1$
			\includegraphics[width=\textwidth]{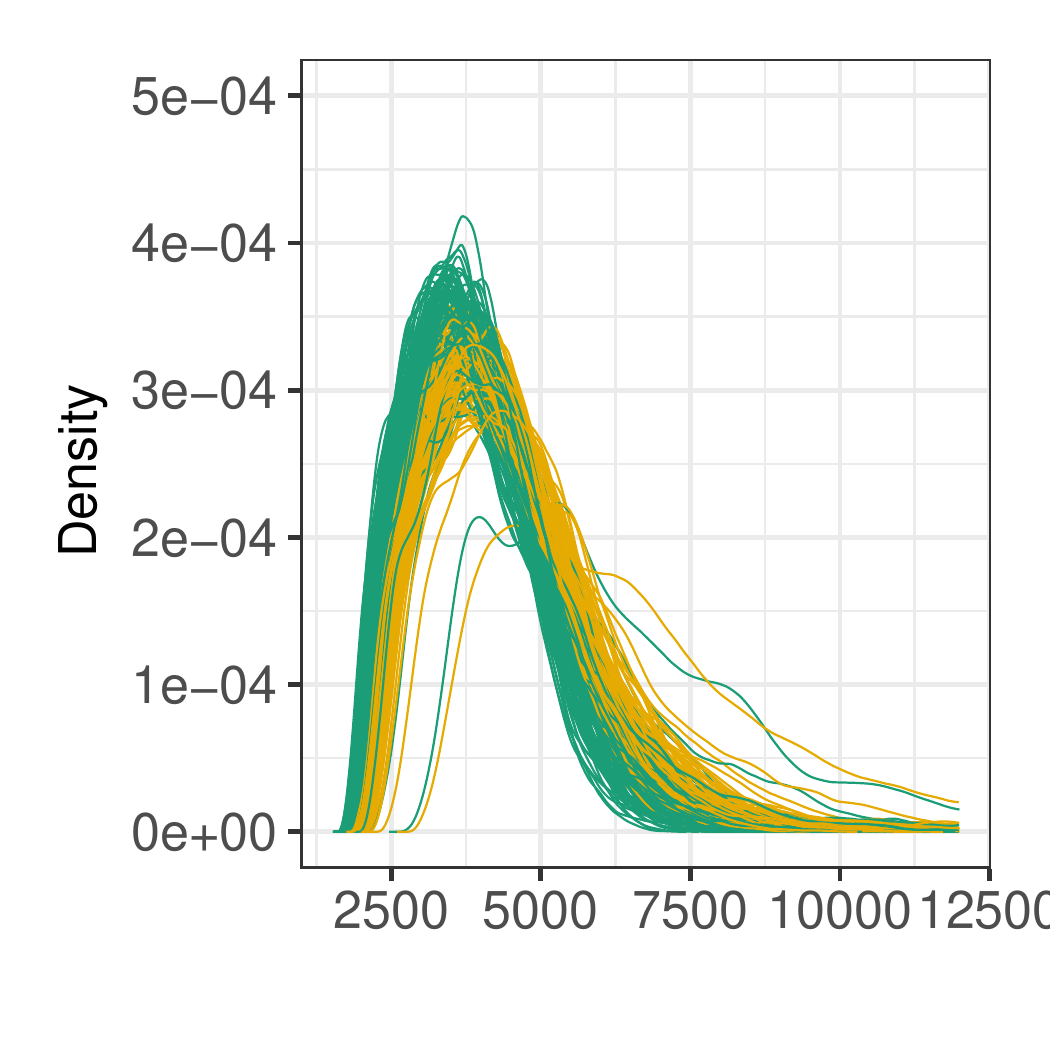}
		\end{minipage}
		\begin{minipage}[c]{0.215\textwidth}
			\centering
			\small\qquad$m=1$
			\includegraphics[width=\textwidth]{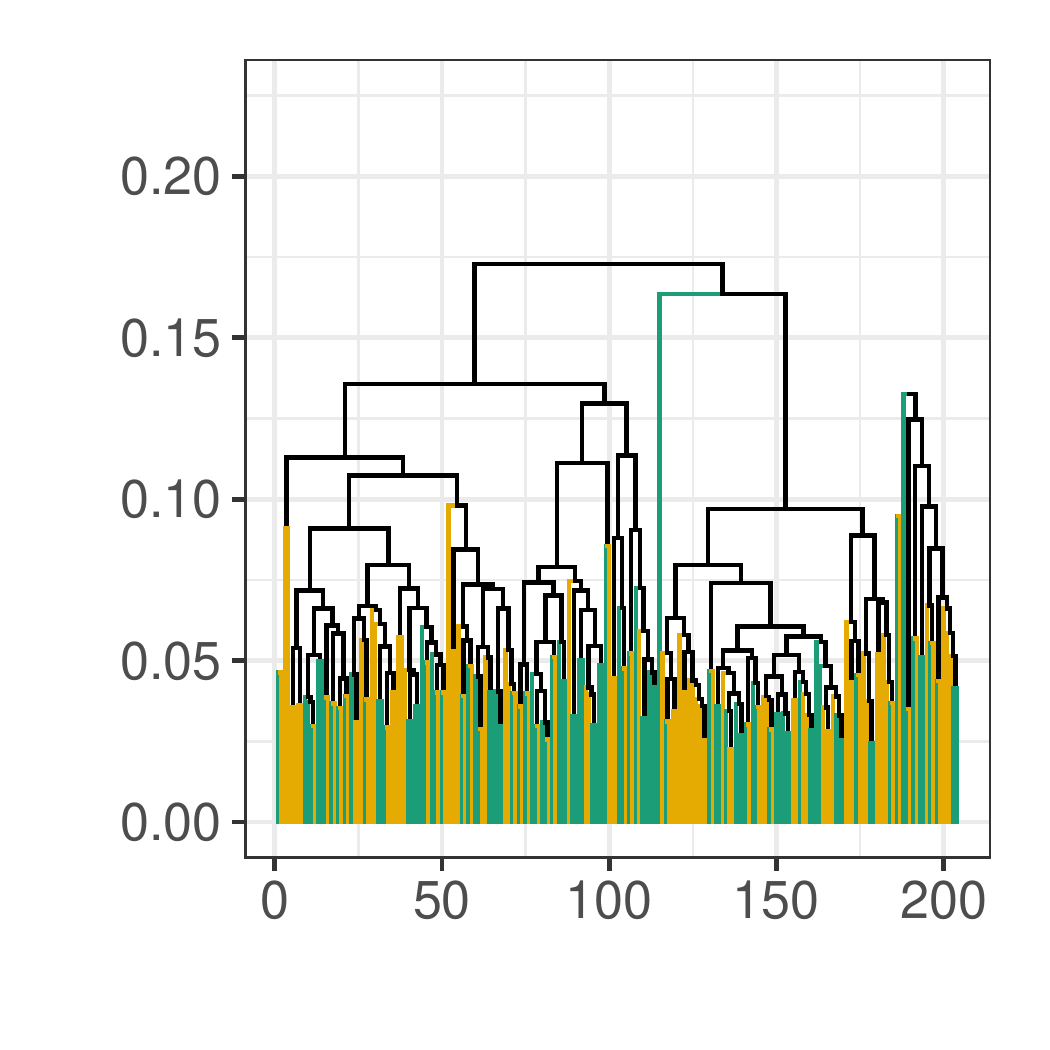}
		\end{minipage}
		\begin{minipage}[l]{0.1\textwidth}
			\hspace{0.2cm}
			\includegraphics[height=0.14\textheight]{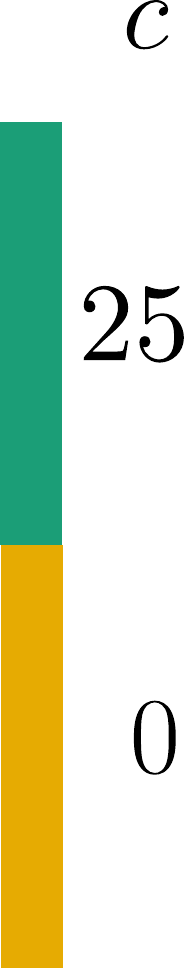}
			\vspace{0.5cm}
		\end{minipage}
		
		\caption{\textbf{Chromatin Analysis IV:}  The results of an average linkage clustering of the kernel density estimators $\kdexvar{i,0}{m}$ (yellow) and $\kdexvar{i,25}{m}$ (blue green), $1\leq i\leq 102$, for $m\in\{0.5,1\}$ (from left to right).
		}\label{fig:loopsmlarge}
	\end{figure}
	Overall, we establish that the methodology provided in this paper does depend on the choice of $m$, but not critically so, as stability over suitable parameter ranges seems to be well maintained in general.

	
	\acks{We thank the Action Editor and two anonymous referees for their careful reading of our manuscript and their many insightful and constructive comments and suggestions.
		C. Weitkamp gratefully acknowledges support by the DFG Research Training Group 2088.}
	
	
	\appendix
	\numberwithin{theorem}{section}
	\section{Proof of Lemma \ref{lemma:dtm m=1}}\label{sec:proof of lemma:dtm m=1}
	In this section, we state the full proof of Lemma \ref{lemma:dtm m=1}.\\
	
	\begin{proof}[Lemma \ref{lemma:dtm m=1}]
		Let $X=(X_1,\dots,X_d)\sim\muX$ and $x=(x_1,\dots,x_d)\in\Rd$.
		\begin{enumerate}
			\item  We observe that 
			\begin{align*}
				\dtmone(x)=\E{||X-x||^2}=&\sum_{i=1}^d \left(\E{X_i^2}-2x_i\E{X_i}+x_i^2\right)\\
				=&\sum_{i=1}^d \left(\left(x_i-\E{X_i}\right)^2+\E{X_i^2}-(\E{X_i})^2\right).
			\end{align*}
			Setting $c_i= \E{X_i}$ and $\zeta =\sum_{i=1}^d\left(\E{X_i^2}-(\E{X_i})^2\right)$ yields the claim.
			
			\item  This follows directly from the fist statement.
			
			\item  The fist statement implies that $\nabla\dtmone(x)=2(x-\E{X}).$ Clearly, this is zero if and only if $x=\E{X}$. 
			
			\item By the second and third statement $\dtmone$ is three times continuously differentiable and $\nabla\dtmone>0$ on ${\dtmone}^{-1}([y-2h_0,y+2h_0])$. In consequence, there exists an open set  $U\supset{\dtmone}^{-1}([y-h_0,y+h_0])$ such that the function
			\[\varphi:U\subset\R^d\to\R^d; ~x \mapsto \frac{\nabla\dtmone(x)}{||\nabla\dtmone(x)||^2}\]
			is $C^{2}(U, \R^d)$. By Theorem 2 in Chapter 15 of \citet{hirsch1974differential} there is a unique flow $\Phi^*:[-h_0,h_0]\times W\to \R^d$ with 
			\begin{equation}\label{eq:ODE dtmone}\begin{cases}
					\frac{\partial}{\partial v}\Phi^*(v,x)=\frac{\nabla\dtmone(\Phi^*(v,x))}{||\nabla\dtmone(\Phi^*(v,x))||^2}\\
					\Phi^*(0,x)=x,
			\end{cases}\end{equation}
			where $W\subset \R^d$ is an open set that contains ${\dtmone}^{-1}([y-h_0,y+h_0])$. Differentiating the function $v\mapsto \dtmone(\Phi^*(v,x))$ immediately shows that $\dtmone\left(\Phi^*(v,x)\right)=\dtmone(x)+v$. This implies that $\Phi^*(v,{\dtmone}^{-1}(\{y\}))={\dtmone}^{-1}(\{y+v\})$. In consequence, it only remains to prove that $\Phi$ defined in the statement is a solution of the ordinary differential equation \eqref{eq:ODE dtmone}, which is straightforward. 
			
		\end{enumerate}
	\end{proof}
	\section{Additional Details on Example \ref{ex:set level set ex}-\ref{ex:2d nonunif}}\label{sec:detail for examples}
	In this section, we will provide additional details on the examples considered in Section \ref{sec:examples}. 
	\subsection{Example \ref{ex:first ex}}\label{subsec:details ex 1}
	Let $\X=[0,1]$ and let $\muX$ denote the uniform distribution on $\X$. 
	First, we derive $\dtmone$ and $\fdtmone$. To this end, we observe that for $x\in\X$ and $X\sim\muX$
	\[\dtmone(x)=\int_0^1\!F_x^{-1}(t)\,dt = \E{(X-x)^2}=\frac{1}{3}-x+x^2.\]
	Based on this, we derive that
	\[\fdtmone(t)=\begin{cases}
		\frac{2\sqrt{3}}{\sqrt{12t-1}}  & \frac{1}{12}< t\leq \frac{1}{3} ,\\
		0 & \, \text{else.} \end{cases}\]
	
	Next, we come to general $m$. In order to calculate $\dtmvar{m}$ (and potentially $\fdtmvar{m}$), it is necessary to derive the family $\left(F_x^{-1}\right)_{x\in\X}$ explicitly. This is straightforward in this setting. 
	We find that
	\[\dtmvar{m}(x)=\frac{1}{m}\int_0^{m}\!F_x^{-1}(u)\,du=\begin{cases}
		x^2 -mx+ \frac{m^2}{3}&\text{for }\, 0\leq x<\frac{m}{2},\\
		\frac{m^2}{12}  &\text{for }\, \frac{m}{2}\leq x \leq 1-\frac{m}{2},\\
		(1-x)^2-m(1-x)+\frac{m^2}{3}   &\text{for }\, 1-\frac{m}{2}\leq x \leq 1.\\
	\end{cases}\]
	Since $\dtmvar{m}$ is constant for $x\in[m/2,1-m/2]$ it is immediately clear that the corresponding distribution function $\Fdtmvar{0.1}$ is not continuous. Indeed, we find that
	\[\Fdtmvar{m}(y)=\Prob\left(\dtmvar{m}(X)\leq y\right)=\begin{cases}
		0 &\text{for }\, y<\frac{m^2}{12},\\
		1-m+2\sqrt{y-\frac{m^2}{12}}  & \text{for }\,\frac{m^2}{12}\leq y \leq \frac{m^2}{3},\\
		1   & \text{for }\,y > \frac{m^2}{3}.\\
	\end{cases}\]
	\subsection{Example \ref{ex:example 2}}\label{subsec:details ex 2}
	Let $\X=[0,1]$ and let $\muX$ denote the probability distribution on $[0,1]$ with density $f(x)=2x$. Let $m=0.1$. As previously, we have to explicitly calculate the family $\left(F_x^{-1}\right)_{x\in\X}$. This is again fairly simple.
	We find that
	\[\dtmvar{0.1}(x)=\frac{1}{0.1}\int_0^{0.1}\!F_x^{-1}(u)\,du=\begin{cases}
		x^2-\frac{2}{3}\sqrt{\frac{2}{5}}x+\frac{1}{20} & 0\leq x\leq \frac{\sqrt{0.1}}{2},\\
		\frac{1}{4800x^2}  & \frac{\sqrt{0.1}}{2}\leq x \leq \frac{1}{2}+\frac{3}{2\sqrt{10}},\\
		x^2+\left(18 \sqrt{\frac{2}{5}}-\frac{40}{3}\right)x+\frac{19}{20}  & \frac{1}{2}+\frac{3}{2\sqrt{10}}< x \leq 1.\\
	\end{cases}\]
	It is obvious that in this case $\dtmvar{0.1}$ is almost nowhere constant. Furthermore, 
	this allows us to derive that
	\[\fdtmvar{0.1}(y) =\begin{cases}
		\frac{80 \sqrt{5}-108 \sqrt{2}}{\sqrt{13661 - 4320 \sqrt{10} + 180 y}}& -\frac{13661}{180}+24\sqrt{10}<y\leq \frac{19-6\sqrt{10}}{120} ,\vspace{2mm}\\
		1 + \frac{40 \sqrt{5}-54 \sqrt{2}}{\sqrt{
				13661 - 4320 \sqrt{10} + 180 y}}+\frac{1}{4800y^2} & \frac{19-6\sqrt{10}}{120}<y\leq \frac{-683}{60}+18\sqrt{\frac{2}{5}} ,\vspace{2mm}\\
		\frac{1}{4800y^2} & \, \frac{-683}{60}+18\sqrt{\frac{2}{5}}<y\leq \frac{1}{120}\vspace{2mm}\\
		-1 + \frac{2}{\sqrt{-\frac{1}{2} + 90 y}} & \, \frac{1}{120}<y\leq \frac{1}{20}\vspace{2mm}\\
		0& \text{else}. \end{cases}\]
	
	\subsection{Example \ref{ex:constant part}}\label{subsec:details ex 3}
	Let $\X=[0,1]^2$ and let $\muX$ stand for the uniform distribution on $\X$. Choose $m = 1$ and let $\X\sim\muX$. Then, it is possible to derive that for $x=(x_1,x_2)\in\X$
	\[\dtmone(x)=\E{||X-x||^2} =x_1^2 + x_2^2-x_1-x_2+\frac{2}{3}.\]
	Based on this, we derive that
	\[\fdtmone(y) =\begin{cases}
		\pi  & \frac{1}{6}\leq y\leq \frac{5}{12} ,\\
		2\arccot{2\sqrt{y-\frac{5}{12} }} - 2\arctan\left(\sqrt{4y-\frac{5}{3} }\right) & \frac{5}{12}<y\leq \frac{2}{3} ,\\
		0 & \, \text{else.} \end{cases}\]
	
	\subsection{Example \ref{ex:2d nonunif}}\label{subsec:details ex 4}
	Let $\X$ denote a disk in $\R^2$ centered at $(0,0)$ with radius 1 and let $\muX$ denote probability measure with density $f$ defined in \eqref{eq:2d nonunif density}.
	In this case, it is for $m=1$ straightforward to derive that for $x=(x_1,x_2)\in\X$
	\begin{equation}\label{eq:ex4dtm}
		\dtmone(x)= x_1^2 + x_2^2+\frac{1}{3}.
	\end{equation}
	The corresponding DTM-density is given as
	\begin{equation}\label{eq:ex4den}
		\fdtmone(y) =\begin{cases}
			-2y+\frac{8}{3} & \frac{1}{3}<y\leq \frac{4}{3} ,\\
			0 & \, \text{else.} \end{cases}
	\end{equation}
	\section{Proof of Theorem \ref{thm:kde limit}}\label{sec:proof of thm:kde limit}
	In this section, we give the full proof of Theorem \ref{thm:kde limit}. The proof is composed of four steps, each of which formulated as an independent lemma (see Section \ref{sec:technical lemmas}). 
	\begin{itemize}
		\item[]\underline{\textbf{Step 1:}} Replacement of $\empdtm(X_i)$ by $ \dtm(X_i)$ (Lemma \ref{lemma:stat rewrite}).\\
		We provide a decomposition of $\kde$ in a sum of two leading terms in which $\empdtm(X_i)$ is replaced by $\dtm(X_i)$ in the argument of the kernel $K$ and we show that the remainder terms are negligible.
		\item[]\underline{\textbf{Step 2:}} Introducing $U$-statistics (Lemma \ref{lemma:ustat reform}).\\
		It is shown that the leading terms obtained in Step 1 can be written as a (sum of two) $U$-statistic(s) asymptotically.
		\item[]\underline{\textbf{Step 3:}} Hoeffding decomposition (Lemma \ref{lemma:Hoeffding decomposition}).\\
		Applying a Hoeffding decomposition allows to derive a representation of the (sum of two) $U$-statistic(s) of step 2 as a sum of a deterministic term (expectation), a stochastic leading term consisting of a sum of independent random variables and a remainder term.
		\item[]\underline{\textbf{Step 4:}} CLT for the leading term of Step 3 (Lemma \ref{lemma:limit of Un}). \\
		Since the leading term of Step 3 is a sum of centered independent random variables, we can apply a standard CLT to show its asymptotic normality.
	\end{itemize}
	
	\subsection{Auxiliary lemmas representing Step 1 - Step 4}\label{sec:technical lemmas} Before we come to the proof of Theorem \ref{thm:kde limit}, we will establish several auxiliary results.
	In order to highlight the overall proof strategy, the corresponding proofs are deferred to Section \ref{subsec:proof of aux}. We begin this section, by addressing Step 1.
	\begin{lemma}[Step 1]\label{lemma:stat rewrite} Assume that Setting \ref{setting} holds and let Condition \ref{condition} be met. Then, it follows that
		\begin{align*}
			\kde (y) =\frac{1}{n}\sum_{i=1}^{n}\left[\frac{1}{h}K\left(\frac{\dtm(X_i)-y}{h}\right)+\frac{1}{h^2}K'\left(\frac{\dtm(X_i)-y}{h}\right)A_n(X_i)\right]\\ +\mathcal{O}_P\left(\frac{1}{nh^3}\right)+ o_P\left(\frac{\log(n)^{1/(2b)}}{n^{1/2+1/(2b)}h}\right),
		\end{align*}
		where
		\[A_n(x)\coloneqq\frac{1}{m}\int_0^{F_x^{-1}(m)}\!F_x(t)-\widehat{F}_{x,n}(t)\,dt.\]
	\end{lemma}

	As a direct consequence of Lemma \ref{lemma:stat rewrite}, we find that for $h\in o\left(1/n^{1/5}\right)$ the statistic
	\begin{align}
		V_n(y)= \frac{1}{n}\sum_{i=1}^n\frac{1}{h} K\left(\frac{\dtm(X_i)-y}{h}\right) +\frac{1}{n}\sum_{i=1}^n\frac{1}{h^2}K'\left(\frac{\dtm(X_i)-y}{h}\right)A_n(X_i)\nonumber\\
		\eqqcolon V_n^{(1)}(y)+V_n^{(2)}(y)\label{eq:def of Vn}
	\end{align}
	drives the limit behavior of $\sqrt{nh}\left(\kde(y)-\fdtm(y)\right)$. Next, we will establish that the statistic $V_n(y)$ can, up to asymptotically negligible terms, be written as a $U$-statistic (see e.g.\ \citet{van2000asymptotic} for more information on $U$-statistics). 
	\begin{lemma}[Introduction of $U$-statistics, Step 2]\label{lemma:ustat reform}
		Assume Setting \ref{setting} and let $V_n^{(1)}$ and $V_n^{(2)}$ be as defined in \eqref{eq:def of Vn}. Then, we have
		\[V_n^{(1)}(y)=\frac{2}{n(n-1)}\sum_{1\leq i<j\leq n}g^{(1)}_{y,h}(X_i,X_j),\]
		where
		\[g^{(1)}_{y,h}(x_1,x_2)=\frac{1}{2h}\left(K\left(\frac{\dtm(x_1)-y}{h}\right)+K\left(\frac{\dtm(x_2)-y}{h}\right)\right).\]
		Furthermore,
		\[V_n^{(2)}(y)=\frac{2}{n(n-1)}\sum_{1\leq i<j\leq n}g^{(2)}_{y,h}(X_i,X_j)+\mathcal{O}_P\left(\frac{1}{nh^2}\right),\]
		where 
		\begin{align*}
			g^{(2)}_{y,h}(x_1,x_2)=&\frac{1}{2mh^2}\Bigg[K'\left(\frac{\dtm(x_1)-y}{h}\right)\int_0^{F_{x_1}^{-1}(m)}\!F_{x_1}(t)-\mathds{1}_{\{||x_1-x_2||^2\leq t\}}\,dt\\&+K'\left(\frac{\dtm(x_2)-y}{h}\right)\int_0^{F_{x_2}^{-1}(m)}\!F_{x_2}(t)-\mathds{1}_{\{||x_1-x_2||^2\leq t\}}\,dt\Bigg].\end{align*}
	\end{lemma}
	\begin{remark}
		It is important to note that $g^{(1)}_{y,h}$ and $g^{(2)}_{y,h}$ are symmetric by definition, i.e., $V_n^{(1)}(x)$ is a $U$-statistic and $V_n^{(2)}(x)$ can be decomposed into a $U$-statistic and an asymptotically negligible remainder term.
	\end{remark}
	
	Combining Lemma \ref{lemma:stat rewrite} and Lemma \ref{lemma:ustat reform}, we see that
	\begin{equation}\label{eq:kde to Un}
		\kde(y)=U_n+\mathcal{O}_P\left(\frac{1}{nh^3}\right)+ o_P\left(\frac{\log(n)^{1/(2b)}}{n^{1/2+1/(2b)}h}\right),
	\end{equation}
	where $U_n=U_n(y)$ denotes the $U$-statistic with kernel function $g_{y,h}(x_1,x_2)\coloneqq g^{(1)}_{y,h}(x_1,x_2)+g^{(2)}_{y,h}(x_1,x_2)$. Before we use \eqref{eq:kde to Un} to finalize the proof of Theorem \ref{thm:kde limit}, we establish two further auxiliary results. Next, we rewrite $U_n$ using the Hoeffding decomposition (see \citet[Sec. 11.4]{van2000asymptotic}), which is the key ingredient to handling the stochastic dependencies introduced by the terms $A_n(X_i)$.
	\begin{lemma}[Hoeffding decomposition, Step 3]\label{lemma:Hoeffding decomposition}
		Assume Setting \ref{setting}. Let $U_n$ be the $U$-statistic with kernel function $g_{y,h}(x_1,x_2)= g^{(1)}_{y,h}(x_1,x_2)+g^{(2)}_{y,h}(x_1,x_2)$. Then, it follows that
		\[U_n=\Theta_{y,h}+\frac{2}{n}\sum_{i=1}^ng_{y,h,1}(X_i)+\frac{2}{n(n-1)}\sum_{1\leq i<j\leq n}g_{y,h,2}(X_i,X_j).\]
		Here, we have that
		\[\Theta_{y,h}=\int\!K\left(v\right)\fdtm(x+vh)\,dv.\]
		Furthermore, let $Z_1,Z_2\iid\muX$. Then, it holds that
		\begin{align*}
			g_{y,h,1}(x_1)=& \frac{1}{2h}K\left(\frac{\dtm(x_1)-y}{h}\right)-\frac{1}{2}\Theta_{y,h} +\Esplit{Z_1}{\frac{1}{2mh^2}K'\left(\frac{\dtm(Z_1)-y}{h}\right)\Psi(x_1,Z_1)}, 
		\end{align*}
		where
		\begin{equation}\label{eq:def of psi}
			\Psi(x_1,x_2)\coloneqq||x_1-x_2||^2\wedge{F_{x_2}^{-1}(m)}-\Esplit{Z_2}{||Z_2-x_2||^2\wedge{F_{x_2}^{-1}(m)}},\end{equation}
		and 
		\begin{equation}
			g_{y,h,2}(x_1,x_2)=g_{y,h}(x_1,x_2)-g_{y,h,1}(x_1)-g_{y,h,1}(x_2)-\Theta_{y,h}. \label{eq:Ustat uncorrelated} \end{equation}
	\end{lemma}
	\begin{remark}\label{rem:correlation}
		It is well known that the mean zero random variables $\left(g_{y,h,2}(X_i,X_j)\right)_{1\leq i<j\leq n}$ are uncorrelated (see \citet[Sec. 11.4]{van2000asymptotic}).
	\end{remark}
	
	For our later considerations it is important to derive a certain regularity for the function $\Psi$ defined in \eqref{eq:def of psi}.
	\begin{lemma}\label{lemma:lippsy} 
		Assume Setting \ref{setting} and let $\Psi$ be the function defined in \eqref{eq:def of psi}. Then, the function $z\mapsto\Psi(x_1,z)$ is Lipschitz continuous for all $x_1\in\X$ and the corresponding Lipschitz constant does not depend on the choice of $x_1$, i.e., it holds
		\[\left|\Psi(x_1,z_1)-\Psi(x_1,z_2)\right|\leq C||z_1-z_2||,\]
		for all $z_1,z_2\in \X$, where the constant $0<C<\infty$ does not depend on $x_1$.
	\end{lemma}

	The next step in the proof of Theorem \ref{thm:kde limit} is to derive for $n\to\infty$ and $h\to 0$ the limit distribution of $\sqrt{nh}\big(\frac{2}{{n}}\sum_{i=1}^ng_{y,h,1}(X_i)\big)$.
	\begin{lemma}\label{lemma:limit of Un} Assume Setting \ref{setting} and let Condition \ref{condition} be met. Let $n\to\infty,h\to0$ such that $nh\to\infty$ and recall that
		\begin{align*}
			g_{y,h,1}(x_1)=& \frac{1}{2h}K\left(\frac{\dtm(x_1)-y}{h}\right)-\frac{1}{2}\Theta_{y,h}+\Esplit{Z_1}{\frac{1}{2mh^2}K'\left(\frac{\dtm(Z_1)-y}{h}\right)\Psi(x_1,Z_1)},
		\end{align*}
		It holds  
		\begin{equation}\label{eq:limit of Un}
			\frac{2\sqrt{h}}{\sqrt{n}}\sum_{i=1}^ng_{y,h,1}(X_i)\Rightarrow N\left(0,\fdtm(y)\int\!K^2(u)\,du\right).
		\end{equation}
	\end{lemma}
	
	With all auxiliary results required established, we can finally come to the proof of Theorem \ref{thm:kde limit}. The proof strategy is to demonstrate that the limit of $\sqrt{nh}\left(\kde(y)-\fdtm(y)\right)$ coincides with the limit of $ \frac{2\sqrt{h}}{\sqrt{n}}\sum_{i=1}^ng_{y,h,1}(X_i)$.\\
	\subsection{Proof of Theorem \ref{thm:kde limit}}
	The proof of Theorem \ref{thm:kde limit} is now a consequence of the lemmas provided in the previous subsection. 
	\begin{proof}[ Theorem \ref{thm:kde limit}]
		We find that
		\begin{align}
			&\left|\sqrt{nh}\left(\kde(y)-\fdtm(y)\right)-\sqrt{nh}\left(\frac{2}{n}\sum_{i=1}^ng_{y,h,1}(X_i)\right)\right|\nonumber\\
			\leq \sqrt{nh}&\left|\kde(y)-\left(\frac{2}{n}\sum_{i=1}^ng_{y,h,1}(X_i)+\Theta_{y,h}\right)\right|+\sqrt{nh}\left|\fdtm(y)-\Theta_{y,h}\right|\label{eq:final proof two summands}.
		\end{align}
		In the following, we consider both summands separately.\\
		
		\textit{First summand:} For the first summand, we obtain that
		\begin{align*}
			\mathcal{S}_n(y)=\sqrt{nh}&\left|\kde(y)-\left(\frac{2}{n}\sum_{i=1}^ng_{y,h,1}(X_i)+\Theta_{y,h}\right)\right|
			\\\leq\sqrt{nh}\Bigg(\!\!&\left|\kde(y)-V_n(y)\right|\!+\!\left|V_n(y)-U_n(y)\right|+\left|U_n(y)-\left(\frac{2}{n}\sum_{i=1}^ng_{y,h,1}(X_i)+\Theta_{y,h}\!\right)\right|\!\Bigg),
		\end{align*}
		where $V_n(y)$ and $U_n(y)$ are defined in\eqref{eq:def of Vn} and \eqref{eq:kde to Un} respectively. By Lemma \ref{lemma:stat rewrite} and $h=o\left(n^{-1/5}\right)$ we obtain that 
		\[\sqrt{nh}\left|\kde(y)-V_n(y)\right|=\mathcal{O}_P\left(\frac{\sqrt{nh}}{nh^3}\right)+ o_P\left(\frac{\sqrt{nh}\log(n)^{1/(2b)}}{n^{1/2+1/(2b)}h}\right)=o_P(1).\]
		Similarly, we get by Lemma \ref{lemma:ustat reform} and $h=o\left(n^{-1/5}\right)$ that 
		\[\sqrt{nh}\left|V_n(y)-U_n(y)\right|=\mathcal{O}_P\left(\frac{\sqrt{nh}}{nh^2}\right)=o_P(1).\]
		Hence, it remains to consider
		\[\sqrt{nh}\left|U_n(y)-\left(\frac{2}{n}\sum_{i=1}^ng_{y,h,1}(X_i)+\Theta_{y,h}\right)\right|=\sqrt{nh}\left|\frac{2}{n(n-1)}\sum_{1\leq i<j\leq n}g_{y,h,2}(X_i,X_j)\right|,\]
		where the last equality follows by Lemma \ref{lemma:Hoeffding decomposition}. Considering the definition of $g_{y,h,2}(x_1,x_2)$ in \eqref{eq:Ustat uncorrelated}, we recognize that $g_{y,h,2}(x_1,x_2)\in\mathcal{O}\left(\frac{1}{h^2}\right)$, as $h\to 0$. Let now $g^*_{y,2}(x_1,x_2)=h^2g_{y,h,2}(x_1,x_2)$. Then, $g^*_{y,h,2}(x_1,x_2)=\mathcal{O}(1)$, as $h\to 0$. Furthermore, we have by Remark \ref{rem:correlation} that the random variables $\{g_{y,h,2}(X_i,X_j)\}_{1\leq i<j\leq n}$ are uncorrelated, whence the same holds for the random variables $\{g^*_{y,h,2}(X_i,X_j)\}_{1\leq i<j\leq n}$. In consequence, we obtain that
		\[\Var{\frac{2}{n(n-1)}\sum_{1\leq i<j\leq n} g^*_{y,2}(X_i,X_j)}=\mathcal{O}\left(n^{-2}\right).\]
		This in turn implies by Chebyshev's inequality that
		\[\frac{2}{n(n-1)}\sum_{1\leq i<j\leq n} g^*_{y,2}(X_i,X_j)=\mathcal{O}_P(n^{-1}).\]
		Therefore, we obtain with $h=o\left(n^{-1/5}\right)$ that
		\begin{align*}\frac{2\sqrt{h}}{\sqrt{n}(n-1)}\sum_{1\leq i<j\leq n} g_{y,h,2}(X_i,X_j)=&\frac{\sqrt{nh}}{h^2}\left(\frac{2}{n(n-1)}\sum_{1\leq i<j\leq n}{g^*_{y,2}}(X_i,X_j)\right)\\
			=&\mathcal{O}_P\left(\frac{1}{n^{1/2}h^{1.5}}\right)=o_P(1).\end{align*}
		Thus, we have shown that $\mathcal{S}(y)=o_P(1)$.\\
		\textit{Second summand:} Finally, we come to the second summand in \eqref{eq:final proof two summands}. First of all, we observe that
		\[\Theta_{y,h}=\int\!K\left(v\right)\fdtm(x+vh)\,dv=\E{\frac{1}{nh}\sum_{i=1}^{n}K\left(\frac{\dtm(X_i)-y}{h}\right)},\]
		where $\{\dtm(X_i)\}_{i=1}^n$ is a collection of i.i.d. random variables with density $\fdtm$. Since $\fdtm$ is assumed to be twice differentiable on $(y-\epsilon,y+\epsilon)$ and $K$ is symmetric, i.e., $\int\!u K(u)\,du=0$, it follows by a straightforward adaptation of Proposition 1.2 of \citet{tsybakov2008introduction} that
		\[\left|\Theta_{y,h}-\fdtm(y)\right|\leq Ch^2.\]
		Here, $C$ denotes a constant independent of $n$ and $h$. We get that
		\[\sqrt{nh}\left|\Theta_{y,h}-\fdtm(y)\right|=\mathcal{O}_P(\sqrt{nh^5})=o_P(1),\]
		as $h=o\left(n^{-1/5}\right)$.\\
		
		In conclusion, we have shown that
		\[\left|\sqrt{nh}\left(\kde(y)-\fdtm(y)\right)-\sqrt{nh}\left(\frac{2}{n}\sum_{i=1}^ng_{y,h,1}(X_i)\right)\right|=o_P(1),\]
		which yields that \[\sqrt{nh}\left(\kde(y)-\fdtm(y)\right)\Rightarrow N\left(0,\fdtm(y)\int_\R\!K^2(u)\,du\right)\]
		as claimed.
	\end{proof}
	\subsection{Proofs of the Auxiliary Lemmas from Section \ref{sec:technical lemmas}}\label{subsec:proof of aux}
	In this section, we gather the full proofs of Lemma \ref{lemma:stat rewrite}-\ref{lemma:limit of Un}.
	\subsubsection{Proof of Lemma \ref{lemma:stat rewrite}}
	In the course of this proof we have to differentiate between the cases $0<m<1$ and $m=1$. \\
	
	\textit{The case $0<m<1$:} By assumption the kernel $K$ is twice continuously differentiable. Using a Taylor series approximation, we find that
	\begin{align*}
		K\!\left(\frac{\empdtm(X_i)-y}{h}\right)= K\!\left(\frac{\dtm(X_i)-y}{h}\right)+\frac{1}{h}K'\!\left(\frac{\dtm(X_i)-y}{h}\right)\!\!\left(\empdtm(X_i)-\dtm(X_i)\right)\\+\frac{1}{2h^2}K''\!\left(\frac{\zeta_i-y}{h}\right)\!\!\left(\empdtm(X_i)-\dtm(X_i)\right)^2,
	\end{align*}
	for some $\zeta_i$ between $\dtm(X_i)$ and $\empdtm(X_i)$. By Theorem 9 in \citet{chazal2017robust} (whose conditions are met by assumption) it holds
	\begin{equation}\label{eq:dtm limit}
		\sup_{x\in \X}\left|\empdtm(x)-\dtm(x)\right|=\mathcal{O}_P(1/\sqrt{n}).\end{equation}
	In consequence, we obtain that
	\begin{align*}
		\kde(y)
		=&\frac{1}{nh}\sum_{i=1}^n\left[K\left(\frac{\dtm(X_i)-y}{h}\right)+\frac{1}{h}K'\left(\frac{\dtm(X_i)-y}{h}\right)\left(\empdtm(X_i)-\dtm(X_i)\right)\right]\\+&\mathcal{O}_P(1/(nh^3)).
	\end{align*}
	Furthermore, it has been shown (see the proof of Theorem 5 in \citet{chazal2017robust}) that for any $x\in\X$
	\begin{align}\empdtm(x)-\dtm(x)=\frac{1}{m}\int_0^{F_x^{-1}(m)}\!F_x(t)-\widehat{F}_{x,n}(t)\,dt + \frac{1}{m}\int_{F_x^{-1}(m)}^{\widehat{F}_{x,n}^{-1}(m)}\!m-\widehat{F}_{x,n}(t)\,dt\nonumber\\\eqqcolon A_n(x)+R_n(x)\label{eq:def of Rn and An}.
	\end{align}
	In consequence, it remains to estimate $\sup_{x\in\X}|R_n(x)|$. Clearly, we have that
	\begin{equation}\label{eq:R_n leq SnTn}
		\left|R_n(x)\right|\leq \frac{1}{m}\left|S_n(x)\right|\left|T_n(x)\right|,\end{equation}
	where
	\begin{equation*}
		S_n(x)=\left|F_x^{-1}(m)-\widehat{F}_{x,n}^{-1}(m)\right|
		\text{ and }T_n(x)=\sup_t\left|F_x(t)-\widehat{F}_{x,n}(t)\right|. \end{equation*}
	\begin{claim}\label{claim:Sn and Tn}
		It holds that 
		\[\sup_{x\in\X}|S_n(x)|=o_P\left(\left(\frac{\log(n)}{n}\right)^{1/(2b)}\right) \text{ as well as } \sup_{x\in\X}|T_n(x)|=\mathcal{O}_P\left(\sqrt{\frac{d}{n}}\right), \]
		where $1\leq b<5$.\vspace{3mm}\\
	\end{claim}
	\noindent Combining Claim \ref{claim:Sn and Tn} with \eqref{eq:R_n leq SnTn} yields $\sup_{x\in\X}|R_n(x)|={o}_P\left(\frac{\log(n)^{1/(2b)}}{n^{1/2+1/(2b)}}\right)$, which gives the statement for $0<m<1$.\vspace{3mm}\\
	\noindent\textit{Proof of Claim \ref{claim:Sn and Tn}:
	} It has already been established in the proof of Theorem 9 in \citet{chazal2017robust} that under the assumptions made \[\sup_{x\in\X}|T_n(x)|=\mathcal{O}_P\left(\sqrt{\frac{d}{n}}\right).\]
	Hence, it only remains to demonstrate the first equality. To this end, let $\xi_i\iid \text{Uniform(0,1)}$, $1\leq i\leq n$, and denote by $H_n$ their empirical distribution function. Define $k=mn$. Then, it holds that $\widehat{F}_{x,n}^{-1}(m)\overset{\mathcal{D}}{=}{F}_{x}^{-1}(\xi_{(k)})={F}_{x}^{-1}\left(H_n^{-1}(m)\right)$. Here, $\xi_{(k)}$ is the $k$-th order statistic and $\overset{\mathcal{D}}{=}$ denotes equality in distribution. Hence, we have for any $m>0$ and $x\in\X$ that
	\begin{align*}
		\Prob\left(|S_n(x)|>\epsilon\right)=&\Prob\left(\left|F_x^{-1}\left(H_n^{-1}(m)\right)-F_x^{-1}(m)\right|>\epsilon\right)\leq\Prob\left(\omega_x\left(\left|m-H_n^{-1}(m)\right|\right)>\epsilon\right),
	\end{align*}
	where $\omega_x$ denote the modulus of continuity for $F^{-1}_x$. This means that for $u\in(0,1)$
	\[\omega_x(u)\coloneqq\sup_{t,t'\in(0,1)^2,|t-t'|<u}\left|F_x^{-1}(t)-F_x^{-1}(t')\right|.\]
	By assumption, there exists a constant $\kappa\in\RR$ such that $\omega_\X(u)=\sup_{x\in\X}{\omega}_x(u)\leq \kappa u^{1/b}$ for all $u\in(0,1)$. Hence, we find that
	\begin{align}
		\Prob\left(\omega_x\left(\left|m-H_n^{-1}(m)\right|\right)>\epsilon\right)\leq& \Prob\left(\left|m-H_n^{-1}(m)\right|>\left(\frac{\epsilon}{\kappa}\right)^b\right)\nonumber\\\leq& 2\exp\left(-\frac{n\left(\frac{\epsilon}{\kappa}\right)^{2b}}{m}\frac{1}{1+\frac{2\left(\frac{\epsilon}{\kappa}\right)^{b}}{3m}}\right),\label{eq:Sn(x)}
	\end{align}
	where the last line follows from \citet{shorack2009empirical} (Inequality 1 on Page 453 and Proposition 1, page 455).
	Next, we observe that
	\begin{align*}
		\Prob\left(\sup_{x\in\X}|S_n(x)|>\epsilon\right)\leq& \Prob\left(\sup_{x\in \X}\omega_x\left(|m-H_n^{-1}(m)|\right)>\epsilon\right)
		\leq \Prob\left(\kappa|m-H_n^{-1}(m)|^{1/b}>\epsilon\right).
	\end{align*}
	Using \eqref{eq:Sn(x)}, we find that
	\begin{align*}
		\Prob\left(\sup_{x\in\X}|S_n(x)|>\epsilon\right)\leq& 2 \exp\left(-\frac{n\left(\frac{\epsilon}{\kappa}\right)^{2b}}{m}\frac{1}{1+\frac{2\left(\frac{\epsilon}{\kappa}\right)^{b}}{3m}}\right)\leq 2\exp\left(-\frac{3n}{5}\left(\frac{\epsilon} {\kappa}\right)^{2b}\right),
	\end{align*} 
	where we used that  $m\leq 1$ and $\epsilon/\kappa<1$ for epsilon small enough. 
	Let $\epsilon =\tau\left(\frac{\log(n)}{n}\right)^{1/(2b)}\!.$ It follows that
	\begin{align*}
		\Prob\left(\sup_{x\in\X}|S_n(x)|>\epsilon\right)\leq2\exp\left(-\frac{3}{5}\left(\frac{\tau} {\kappa}\right)^{2b}\log(n)\right),
	\end{align*}
	and thus
	\[\sup_{x\in\X}|S_n(x)|=o_P\left(\frac{\log(n)}{n}\right)^{1/(2b)},\]
	which yields Claim \ref{claim:Sn and Tn}.\vspace{6mm}\\
	\textit{The case $m=1$:}
	Similar as for $0<m<1$, we find that
	\begin{align*}
		K\left(\frac{\empdtmone(X_i)-y}{h}\right)= K\left(\frac{\dtmone(X_i)-y}{h}\right)+\frac{1}{h}K'\left(\frac{\dtmone(X_i)-y}{h}\right)\left(\empdtmone(X_i)-\dtmone(X_i)\right)\\+\frac{1}{2h^2}K''\left(\frac{\zeta_i-y}{h}\right)\left(\empdtmone(X_i)-\dtmone(X_i)\right)^2,
	\end{align*}
	for some $\zeta_i$ between $\dtmone(X_i)$ and $\empdtmone(X_i)$. While Theorem 9 in \cite{chazal2017robust} is only stated for $m<1$, it is straightforward to adapt its proof to the case $m =1$. Hence, we obtain
	\begin{equation}\label{eq:dtm limit2}
		\sup_{x\in \X}\left|\empdtmone(x)-\dtmone(x)\right|=\mathcal{O}_P(1/\sqrt{n}).\end{equation}
	Furthermore, we note that for $x\in\X$
	\begin{equation}\label{eq:dtm1}
		\empdtmone(x)-\dtmone(x)=\int_0^{D_x}\!F_x(t)-\widehat{F}_{x,n}(t)\,dt,\end{equation}
	where $[0,D_x]$ denotes the support of $F_x$. In combination with our previous considerations, we find that\begin{align*}
		\kdeone(y)=\frac{1}{nh}\sum_{i=1}^n\left[K\left(\frac{\dtmone(X_i)-y}{h}\right)+\frac{1}{h}K'\left(\frac{\dtmone(X_i)-y}{h}\right)A_n(x)\right]+&\mathcal{O}_P\left(\frac{1}{nh^3}\right),
	\end{align*}
	which yields the claim.\\ \QEDA
	\subsubsection{Proof of Lemma \ref{lemma:ustat reform}}
	First, we consider $V_n^{(1)}(x)$. Clearly, we have 
	\begin{align*}
		V_n^{(1)}(y)=&
		\frac{2}{n(n-1)}\sum_{1\leq i<j\leq n}\frac{1}{2h}\left(K\left(\frac{\dtm(X_i)-y}{h}\right)+K\left(\frac{\dtm(X_j)-y}{h}\right)\right)\\
		=& \frac{2}{n(n-1)}\sum_{1\leq i<j\leq n} g^{(1)}_{y,h}(X_i,X_j).
	\end{align*}
	Next, we come to $V_n^{(2)}(y)$. We have that
	\begin{align*}
		V_n^{(2)}(y)=&
		\frac{1}{n}\sum_{i=1}^n\frac{1}{h^2}K'\left(\frac{\dtm(X_i)-y}{h}\right)\frac{1}{m}\int_0^{F_{X_i}^{-1}(m)}\!F_{X_i}(t)-\widehat{F}_{X_i,n}(t)\,dt \\
		=& \frac{1}{n^2}\sum_{i=1}^n\sum_{j=1}^n\frac{1}{mh^2}K'\left(\frac{\dtm(X_i)-y}{h}\right)\int_0^{F_{X_i}^{-1}(m)}\!F_{X_i}(t)-\mathds{1}_{\{||X_i-X_j||^2\leq t\}}\,dt.
	\end{align*}
	Further, we obtain that
	\begin{align*}
		V_n^{(2)}(y)
		=&\frac{1}{n^2}\!\!\sum_{1\leq i<j\leq n} \!g^{(2)}_{y,h}(X_i,X_j)
		+\frac{1}{n^2}\!\sum_{i=1}^n\frac{1}{mh^2}\!\!\left[K'\!\left(\frac{\dtm(X_i)-y}{h}\right)\!\!\int_0^{F_{X_i}^{-1}(m)}\!F_{X_i}(t)-1\,dt\right]\!.
	\end{align*}
	We note that $K$ is twice differentiable and $\X$ is compact, i.e.,
	\[\left|\int_0^{F_{x_1}^{-1}(m)}\!F_{x_2}(t)-\mathds{1}_{\{||x_1-x_2||^2\leq t\}}\,dt\right|\leq \diam{\X}<\infty\]\\
	for all $x_1,x_2\in \X$. This yields that
	\begin{align*}
		V_n^{(2)}(y){=}&\frac{2}{n^2}\sum_{1\leq i<j\leq n}g^{(2)}_{y,h}(X_i,X_j)+\mathcal{O}_P\left(\frac{1}{nh^2}\right)\\
		=&\frac{n-1}{n}\left(\frac{2}{n(n-1)}\sum_{1\leq i<j\leq n}g^{(2)}_{y,h}(X_i,X_j)\right)+\mathcal{O}_P\left(\frac{1}{nh^2}\right)\\
		{=}&\frac{2}{n(n-1)}\sum_{1\leq i<j\leq n}g^{(2)}_{y,h}(X_i,X_j)+\mathcal{O}_P\left(\frac{1}{nh^2}\right).
	\end{align*}
	\QEDA
	\subsubsection{Proof of Lemma \ref{lemma:Hoeffding decomposition}}
	Since $X_1,\dots X_n\iid\muX$, the claim follows by the Hoeffding decomposition \citep[Lemma 11.11]{van2000asymptotic} once we have shown that
	\begin{enumerate}
		\item $\Theta_{y,h}=\E{U_n}$
		\item $g_{y,h,1}(x_1) =\E{ g_{y,h}(x_1,Z_1)}-\Theta_{y,h},$ where $Z_1\sim\muX$.
	\end{enumerate}
	\textit{First equality:} We start by verifying the first equality. Clearly,
	\[\E{U_n}=\E{g^{(1)}_{y,h}(X_1,X_2)}+\E{g^{(2)}_{y,h}(X_1,X_2)}.\]
	Since $X_1,X_2\iid \muX$, we obtain that
	\begin{align*}
		\E{g^{(1)}_{y,h}(X_1,X_2)} =& \E{\frac{1}{2h}\left(K\left(\frac{\dtm(X_1)-y}{h}\right)+K\left(\frac{\dtm(X_2)-y}{h}\right)\right)}\\ 
		=& \int\!\frac{1}{h}K\left(\frac{u-y}{h}\right)\,d(\dtm\#\muX)(u).
	\end{align*}
	Here, the last equality follows by the change-of-variables formula ($\dtm\#\muX$ denotes the pushforward measure of $\muX$ with respect to $\dtm$). By assumption the measure $\dtm\#\muX$ possesses a density $\fdtm$ with respect to the Lebesgue measure. Hence,
	\begin{align*}
		\E{g^{(1)}_{y,h}(X_1,X_2)}=& \int\!\frac{1}{h}K\left(\frac{u-y}{h}\right)\fdtm(u)\,du=\int\!K\left(v\right)\fdtm(y+vh)\,dv.
	\end{align*}
	
	As $X_1,X_2\iid\muX$, we obtain for the second summand that 
	\begin{align*}
		\E{g^{(2)}_{y,h}(X_1,X_2)}
		\!\!=& \E{\frac{1}{mh^2}K'\!\left(\frac{\dtm(X_1)-y}{h}\right)\!\!\int_0^{F_{X_1}^{-1}(m)}\!\!F_{X_1}(t)-\mathds{1}_{\{||X_1-X_2||^2\leq t\}}\,dt}
	\end{align*}
	Since $X_1$ and $X_2$ are independent, the Theorem of Tonelli/Fubini \citep[Thm. 18]{BillingsleyConvergenceProbabilityMeasures2013} yields that
	\begin{align*}
		\E{g^{(2)}_{y,h}(X_1,X_2)}
		\!\!=&\Esplit{X_1}{\frac{1}{mh^2}K'\!\left(\frac{\dtm(X_1)-y}{h}\right)\!\!\int_0^{F_{X_1}^{-1}(m)}\!\!F_{X_1}(t)-\Esplit{X_2}{\mathds{1}_{\{||X_1-X_2||^2\leq t\}}}dt}\\
		=&0.
	\end{align*}
	Combining our results, we find that $\E{U_n}=\Theta_{y,h}$.\vspace{3mm}\\
	\textit{Second equality:} Recall that $Z_1\sim\muX$. We demonstrate that 
	\[g_{y,h,1}(x_1)=\E{g_{y,h}(x_1,Z_1)}-\Theta_{y,h}=\E{g^{(1)}_{y,h}(x_1,Z_1)}+\E{g^{(2)}_{y,h}(x_1,Z_1)}-\Theta_{y,h}.\]
	Once again, we consider the two summands separately. We observe that
	\begin{align*}\E{g^{(1)}_{y,h}(x_1,Z_1)}
		=&\frac{1}{2h}K\left(\frac{\dtm(x_1)-y}{h}\right) +\frac{1}{2}\Theta_{y,h}.\end{align*}
	Here, the last equality follows by our previous considerations for $\E{U_n}$. For the second summand, it follows that
	\begin{align*}
		\E{g^{(2)}_{y,h}(x_1,Z_1)}\!\!=&\E{ \frac{1}{2mh^2}K'\!\left(\frac{\dtm(x_1)-y}{h}\right)\!\!\int_0^{F_{x_1}^{-1}(m)}\!F_{x_1}(t)\!-\!\mathds{1}_{\{||x_1-Z_1||^2\leq t\}}\,dt}\\+&\E{\frac{1}{2mh^2}K'\!\left(\frac{\dtm(Z_1)-y}{h}\right)\!\!\int_0^{F_{Z_1}^{-1}(m)}\!F_{Z_1}(t)\!-\!\mathds{1}_{\{||Z_1-x_1||^2\leq t\}}\,dt}\!\eqqcolon \!T_1\! +\!T_2.
	\end{align*}
	The Theorem of Tonelli/Fubini \citep[Thm. 18]{BillingsleyConvergenceProbabilityMeasures2013} shows that
	\begin{align*}
		T_1=\frac{1}{2mh^2}K'\left(\frac{\dtm(x_1)-y}{h}\right) \int_0^{F_{x_1}^{-1}(m)}\!F_{x_1}(t)-\E{\mathds{1}_{\{||x_1-Z_1||^2\leq t\}}}\,dt
		=0.
	\end{align*}
	
	Furthermore, we obtain for $Z_2\sim\muX$ independent of $Z_1$ that
	\begin{align*}
		T_2=&\Esplit{Z_1}{\frac{1}{2mh^2}K'\left(\frac{\dtm(Z_1)-y}{h}\right)\int_0^{F_{Z_1}^{-1}(m)}\!\Esplit{Z_2}{\mathds{1}_{\{||Z_1-Z_2||^2\leq t\}}}-\mathds{1}_{\{||Z_1-x_1||^2\leq t\}}\,dt}\\=&\Esplit{Z_1}{\frac{1}{2mh^2}K'\left(\frac{\dtm(Z_1)-y}{h}\right)\Esplit{Z_2}{\int_0^{F_{Z_1}^{-1}(m)}\!\mathds{1}_{\{||Z_1-Z_2||^2\leq t\}}-\mathds{1}_{\{||Z_1-x_1||^2\leq t\}\,dt}}},
	\end{align*}
	where the last step follows by the theorem of Tonelli/Fubini \citep[Thm. 18]{BillingsleyConvergenceProbabilityMeasures2013}. Moreover, the determination of the integral in the above expression yields that
	\begin{align*}
		T_2\!
		=&\Esplit{Z_1}{\frac{1}{2mh^2}K'\!\left(\frac{\dtm(Z_1)-y}{h}\right)\!\!\left(||x_1-Z_1||^2\!\wedge\!{F_{Z_1}^{-1}(m)}-\Esplit{Z_2}{||Z_1-Z_2||^2\!\wedge\!{F_{Z_1}^{-1}(m)}}\right)\!}\!\!.
	\end{align*}
	Combining all of our results, we finally get that
	\begin{align*}
		g_{y,h,1}(x_1)=& \frac{1}{2h}K\!\left(\frac{\dtm(x_1)-y}{h}\right)-\frac{1}{2}\Theta_{y,h}+\Esplit{Z_1}{\frac{1}{2mh^2}K'\left(\frac{\dtm(Z_1)-y}{h}\right)\Psi(x_1,Z_1)},
	\end{align*}
	as claimed.\\\QEDA
	
	\subsubsection{Proof of Lemma \ref{lemma:lippsy}}
	Let $x_1\in\X$ be arbitrary. We observe that for any $z_1,z_2\in\X$
	\begin{align*}
		|\Psi(x_1,z_1)-\Psi(x_1,z_2)|
		\leq &\big|||x_1-z_1||^2\wedge{F_{z_1}^{-1}(m)}-||x_1-z_2||^2\wedge{F_{z_2}^{-1}(m)}\big|\\+&\big|\Esplit{Z_2}{||Z_2-z_1||^2\wedge{F_{z_1}^{-1}(m)}}-\Esplit{Z_2}{||Z_2-z_2||^2\wedge{F_{z_2}^{-1}(m)}}\big|\\
		\eqqcolon&\Psi_1(z_1,z_2)-\Psi_2(z_1,z_2).
	\end{align*}
	In the following, we consider $\Psi_1$ and $\Psi_2$ separately. We have that for $z_1,z_2\in\X$
	\begin{align*}
		\Psi_1(z_1,z_2)
		\leq& \big| ||x_1-z_1||^2-||x_1-z_2||^2\big|+\big|F_{z_1}^{-1}(m)-F_{z_2}^{-1}(m) \big|\\
		\leq& D \big|||x_1-z_1||-||x_1-z_2||\big|+2\sqrt{D}||z_1-z_2||,
	\end{align*}
	where the last inequality follows with $D=\diam{\X}<\infty$ and Lemma 8 in \citet{chazal2017robust}. In particular, note that in the current setting we have that
	\[\sup_{t\in(0,1)}\sup_{x\in\X}F_x^{-1}(t)\leq D<\infty.\]
	In consequence, we obtain that for $z_1,z_2\in\X$
	\begin{align*}
		\Psi_1(z_1,z_2)\leq D\big|||z_1-z_2||\big|+2\sqrt{D}||z_1-z_2||
		\leq C||z_1-z_2||,
	\end{align*}
	where $C$ denotes a constant that only depends on $\X$.\\
	
	Next, we observe
	\begin{align*}
		\Psi_2(z_1,z_2)
		\leq& \Esplit{Z_2}{\big|||Z_2-z_1||^2\wedge{F_{z_1}^{-1}(m)}-||Z_2-z_2||^2\wedge{F_{z_2}^{-1}(m)}\big|}.
	\end{align*}
	Considering our previous calculation, we immediately obtain that
	\begin{align*}
		\Psi_2(z_1,z_2)\leq&\Esplit{Z_2}{C||z_1-z_2||}=C||z_1-z_2||,
	\end{align*}
	where $C$ denotes the same constant as previously. Combining our results, we find that
	\[|\Psi(x_1,z_1)-\Psi(x_1,z_2)|\leq C||z_1-z_2||,\]
	where the constant $C$ only depends on $\X$ and not on $x_1$. This yields the claim.\\\QEDA
	
	\subsubsection{Proof of Lemma \ref{lemma:limit of Un}}
	Next, we derive \eqref{eq:limit of Un} using Lyapunov's Central Limit Theorem for triangular arrays \citep[Sec. 27]{billingsley2008probability}. To this end, we define $z_{in}\coloneqq 2g_{y,h,1}(X_i)$, $\Bar{z}_n=\frac{1}{n}\sum_{i=1}^nz_{in}$ and $\sigma_{in}^2\coloneqq\Var{z_{in}}$. Clearly, for $n$ fixed the $z_{in}$'s are independent and identically distributed. In order to check the assumptions of Lyapunov's Central Limit Theorem, it remains to find $r>2$ such that
	\begin{equation}\label{eq:clt cond 1}
		\rho_{1n}\coloneqq \E{|z_{1n}-\E{z_{1n}}|^r}<\infty\end{equation}
	and
	\begin{equation}\label{eq:clt cond 2}
		\frac{n\rho_{1n}}{(n\sigma^2_{1n})^{r/2}}\to 0,
	\end{equation}
	as $n\to\infty$. \\
	\underline{Calculation of $\sigma_{1n}^2$:} The next step is to consider $\sigma_{1n}^2$. As $\E{z_{1n}}=0$ by construction, we find that
	\begin{align*}
		\sigma_{1n}=& \E{|z_{1n}|^2}=\E{\left|2g_{y,h,1}(X_1)\right|^2}=\E{\left|T_3+T_4\right|^2},
	\end{align*}
	where
	\begin{align}\label{eq:def of T3}
		T_3\coloneqq \frac{1}{h}K\left(\frac{\dtm(X_1)-y}{h}\right)-\Theta_{y,h}
	\end{align}
	and
	\begin{equation}\label{eq:def of T4}
		T_4\coloneqq \Esplit{Z_1}{\frac{1}{mh^2}K'\left(\frac{\dtm(Z_1)-y}{h}\right)\Psi(X_1,Z_1)}.
	\end{equation}
	Here, $\Psi(x_1,x_2)$ is the function defined in \eqref{eq:def of psi}.
	Obviously, we obtain that
	\begin{align*}
		\sigma_{1n}^2=\E{T_3^2}+\E{T_4^2}+2\E{T_3T_4}.
	\end{align*}
	In the following, we treat each of these summands separately.\\
	
	\textit{First summand:} Considering the first term, we see that
	\[\E{T_3^2}=\E{\left|\frac{1}{h}K\left(\frac{\dtm(X_1)-y}{h}\right) -\Theta_{y,h}\right|^2}.\]
	which is essentially the variance of the kernel density estimator of the real valued random variable $\dtm(X_1)$. Hence, one can show using standard arguments (see e.g.\ \citet[Sec. 3.3]{silverman2018density}) that
	\[\E{|T_3|^2}= \frac{\fdtm(y)}{h}\int\!|K(u)|^2\,du+o\left(\frac{1}{h}\right).\]
	as $h\to0$.\\
	
	\textit{Second summand:} Next, we consider $\E{|T_4|^2}$. We have that 
	\[\E{|T_4|^2}=\E{\left|\Esplit{Z_1}{\frac{1}{mh^2}K'\left(\frac{\dtm(Z_1)-y}{h}\right)\Psi(X_1,Z_1)}\right|^2}.\]
	Recall that $Z_1\sim\muX$ and that $\muX$ has, by assumption, a Lipschitz continuous Lebesgue density. Denote this density by $g_{\muX}$. Then, it follows that
	\begin{align}
		&\Esplit{Z_1}{\frac{1}{mh^2}K'\left(\frac{\dtm(Z_1)-y}{h}\right)\Psi(X_1,Z_1)}\nonumber\\=&\frac{1}{mh^2}\int_\X\! K'\left(\frac{\dtm(z)-y}{h}\right)\Psi(X_1,z)g_{\muX}(z)\,d\lambda^d(z)\label{eq:T4 estimate}
	\end{align}
	Next, we realize that 
	\[\sup_{x_1,x_2}|\Psi(x_1,x_2)|\leq D\]
	and hence there is a constant $0<C<\infty$ such that
	\begin{equation}\label{eq:mult lip bound}
		\max\left\{\sup_{x\in\X}g_{\muX}(x),\sup_{x_1,x_2}|\Psi(x_1,x_2)|\right\}<C.
	\end{equation}
	Further, it follows by Lemma \ref{lemma:lippsy} that the function $\Psi^*_{x_1}:\X\to\R$, $z\mapsto\Psi(x_1,z)$ is Lipschitz continuous for all $x_1\in\X$ with a Lipschitz constant that does not depend on $x_1$. This in combination with the Lipschitz continuity of $g_{\muX}$ and \eqref{eq:mult lip bound} implies that the function
	\[\psi_{x_1}:\X\to\R,~~ z\mapsto\Psi(x_1,z)g_{\muX}(z)\]
	is Lipschitz continuous for all $x_1\in\X$ with a Lipschitz constant that does not depend on $x_1$. We have that the function $x\mapsto\dtm(x)$ is coercive, that $\dtm$ is $C^{2,1}$ on an open neighborhood of $\Gamma_y={\dtm}^{-1}(y)$ and that $\nabla\dtm\neq0$ on $\Gamma_y$ by assumption. By Condition \ref{condition}, there exists $h_0>0$ such that for all $-h_0<v<h_0$
	\begin{equation*}
		\int_{\Gamma_y}\!\left|\indifunc{\Phi(0,x)\in\X}-\indifunc{\Phi(v,x)\in\X}\right|\,d\Haus{d-1}(x)\leq C_y |v|,
	\end{equation*}
	where $\Phi$ denotes the canonical level set flow of $\Gamma_y$ and $C_y$ denotes a finite constant that depends on $y$ and $\dtm$. Furthermore, the kernel $K$ is twice continuously differentiable and $\supp(K)=[-1,1]$. Since $K$ is also even, by assumption, it follows that $K'$ is odd, i.e.\
	\[\int_{-1}^1\! K'(z)\,dz=0.\]
	Thus, we find by Theorem \ref{thm:h2bound} that there exists some constants $c_y>0$ and $h_0>0$ (depending on $\dtm$, $y$ and $\X$) such that for any $h<h_0$ we obtain that
	\begin{equation*}
		\left|\Esplit{Z_1}{K'\left(\frac{\dtm(Z_1)-y}{h}\right)\Psi(X_1,Z_1)}\right|\leq c_yh^2.
	\end{equation*}
	In consequence, we find that for $h$ small enough
	\[\E{|T_4|^2}\leq\E{\frac{1}{m^2h^4}\left|\Esplit{Z_1}{K'\left(\frac{\dtm(Z_1)-y}{h}\right)\Psi(X_1,Z_1)}\right|^2}\leq \frac{c_y}{m^2}.\]
	This in particular shows that $\E{|T_4|^2}= \mathcal{O}(1)$ as $h\to 0$.\\
	
	\textit{Third summand:}. By Hölder's inequality, we obtain that
	\begin{align*}
		\E{T_3T_4}\leq \E{|T_3T_4|}\leq \left(\E{|T_3|^2}\right)^{1/2}\left(\E{|T_4|^2}\right)^{1/2}.
	\end{align*}
	Plugging in our previous findings, we find that
	\[ \E{T_3T_4}=\mathcal{O}\left(\frac{1}{\sqrt{h}}\right)\cdot\mathcal{O}\left(1\right)=\mathcal{O}\left(\frac{1}{\sqrt{h}}\right).\]
	as $h \to 0$. In consequence, we find that
	\[\sigma_{1n}^2=\frac{\fdtm(x)}{h}\int\!|K(u)|^2\,du+o\left(\frac{1}{h}\right).\]
	This concludes our consideration of $\sigma_{1n}^2$.\vspace{6mm}\\
	\underline{Calculation of third moments:} We choose $r=3$ and consider
	$\rho_{1n}= \E{|z_{1n}-\E{z_{1n}}|^r}$. By construction $\E{z_{1n}}=0$. Thus, we obtain
	\begin{align*}
		\rho_{1n}=& \E{|z_{1n}|^3}=\E{\left|g_{y,h,1}(X_1)\right|^3}=\E{\left|T_3+T_4\right|^3},
	\end{align*}
	where $T_3$ and $T_4$ denote the terms defined in \eqref{eq:def of T3} and \eqref{eq:def of T4}, respectively. 
	Furthermore, it follows that
	\begin{align*}
		\rho_{1n}\leq \E{\left(|T_3|+|T_4|\right)^3}\leq 8\E{|T_3|^3}+8\E{|T_4|^3}.
	\end{align*}
	Considering the first summand, this yields that
	\begin{align*}
		\E{|T_3|^3}=\E{\left|\frac{1}{h}K\left(\frac{\dtm(X_1)-y}{h}\right) -\Theta_{y,h}\right|^3},
	\end{align*}
	which is the third moment of the kernel density estimator of the real valued random variable $\dtm(X_1)$. In particular, one can show using standard arguments that
	\[\E{|T_3|^3}\leq \frac{8\fdtm(x)}{h^2}\int\!|K(u)|^3\,du+o\left(\frac{1}{h^2}\right).\]
	It remains to consider $\E{|T_4|^3}$.
	We have already shown that for $h\to0$
	\[\left|\frac{1}{2mh^2}\Esplit{Z_1}{K'\left(\frac{\dtm(Z_1)-y}{h}\right)\Psi(X_1,Z_1)}\right|=\mathcal{O}(1).\]
	Consequently, this implies that $\E{|T_4|^3}=\mathcal{O}(1).$ Hence, we obtain that
	\[\rho_{1n}\leq \frac{8\fdtm(y)}{h^2}\int_{-1}^1\!|K(u)|^3\,du+o\left(\frac{1}{h^2}\right).\]
	\underline{Applying Lyapunov's CLT:} Now that we have calculated $\rho_{1n}$ and $\sigma^2_{1n}$, we can verify the remaining assumption of Lyapunov's Central Limit Theorem for triangular array's \cite[Sec. 26]{billingsley2008probability}.
	First of all, we observe that $\rho_{1n}<\infty$, since $K$, $K'$ and $\Psi$ are continuous and compactly supported. Furthermore, we obtain
	\begin{align*}
		\frac{n\rho_{1n}}{\left(n\sigma^2_{1n}\right)^{3/2}}\leq&  \frac{\frac{8n\fdtm(x)}{h^2}\int\!|K(u)|^3\,du+o\left(\frac{n}{h^2}\right)}{\left(\frac{n\fdtm(x)}{2h}\int\!|K(u)|^2\,du+o\left(\frac{n}{h}\right)\right)^{3/2}}
		=\mathcal{O}\left((nh)^{-1/2}\right)\to 0,
	\end{align*}
	if $nh\to\infty.$
	In consequence, Lyapunov's Central Limit Theorem for triangular arrays is applicable. It yields that
	\begin{equation*}
		\frac{\Bar{z}_n-\E{\Bar{z}_n}}{\sqrt{\Var{\Bar{z}_n}}}\overset{D}{\to}N(0,1).
	\end{equation*}
	This in turn implies that
	\begin{equation*}
		\sqrt{nh}\left(\frac{2}{n}\sum_{i=1}^n g_{y,h,1}(X_i)\right)\Rightarrow N\left(0,\fdtm(y)\int_{-1}^1\!|K(u)|^2\,du\right)
	\end{equation*}
	which gives the claim.\\
	
	\QEDA
	
	\section{Some Geometric Measure Theory}
	
	In the proof of Lemma \ref{lemma:limit of Un}, we need to bound the term \eqref{eq:T4 estimate}:
	
	\begin{align*}
		\mathcal{I}(y):=\int_\X\! K'\left(\frac{\dtm(z)-y}{h}\right)\Psi(X_1,z)g_{\muX}(z)\,d\lambda^d(z),
	\end{align*}
	where $g_{\muX}$ denotes the Lebesgue density of $\muX$ (which exists by assumption) and $X_1\sim\muX$. Since the kernel $K$ and thus also its derivative $K'$ are supported on $[-1,1]$ , we obtain
	\begin{align*}
		\mathcal{I}(y)=\int_{A_{h}(y)}\! K'\left(\frac{\dtm(z)-y}{h}\right)\Psi(X_1,z)g_{\muX}(z)\,d\lambda^d(z),
	\end{align*}
	where
	\begin{align}\label{eq:Ah}
		A_h(y):=\left\{z\in\X\,|\, y-h\leq\dtm(z)\leq y+h\right\}=\left(\dtm\right)^{-1}[y-h,y+h]\cap \X.
	\end{align}
	In the following, we will show how to control such integrals over thickened level sets such as $A_h(y)$ for small $h$. More precisely, we prove the subsequent theorem that has already been applied to bound the term $\mathcal{I}(y)$ in the proof of Lemma \ref{lemma:limit of Un}.
	
	\begin{theorem}\label{thm:h2bound}
		Let $\X\subset \R^d$ be a compact set. Let $g\colon\X \to [-\alpha, \alpha]$ be $\alpha$-Lipschitz continuous and suppose that $k\colon \RR \to [-\alpha, \alpha]$ for some $\alpha > 0$. Assume that $\supp(k) = [-1, 1]$ and $\int
		k(s)\,d s = 0$. Let $\dfct\colon \RR^d \to \RR$ be a coercive function, i.e., $\lim_{||x||\to\infty}\dfct(x)=\infty$, with level sets $\Gamma_y
		= \dfct^{-1}\{y\}$ for $y\in \RR$. Call $y\in \RR$ a $C^{2,1}$-\emph{regular bounded
			value} of $d$ with respect to $\X$ if
		\begin{enumerate}[label=\textbf{C.\arabic*},ref=C.\arabic*]
			\item $\Gamma_y$ has an open neighborhood on which $\dfct$ is $C^{2,1}$, 
			\item $\nabla\dfct \neq 0$ on $\Gamma_y$.
			\item There exists $h_0^*>0$ and such that for all $-h_0^*<v<h_0^*$
			\begin{equation*}
				\int_{\Gamma_y}\!\left|\indifunc{\Phi(0,x)\in\X}-\indifunc{\Phi(v,x)\in\X}\right|\,d\Haus{d-1}(x)\leq C_y |v|,
			\end{equation*}
			where $\Phi$ denotes the canonical level set flow of $\Gamma_y$ and $C_y$ denotes a constant that only depends on the function $\dfct,$ the variable $y$ and the underlying space $\X$.
		\end{enumerate}
		If $y$ is a $C^{2,1}$-regular bounded value of $d$ with respect to $\X$, then
		\begin{equation}
			\left|\int_\X k\left(\frac{\dfct(x) - y}{h}\right) g(x)\,d\lambda^d(x)\right| \le c_y h^2
		\end{equation}
		for some $c_y > 0$ and any $0<h < h_0$, where $c_y$ and $h_0>0$ only depend on $d$,
		$y$, $\alpha$ and $\X$ (and not on $k$ and $g$ explicitly).
	\end{theorem}
	The proof of Theorem \ref{thm:h2bound} consists of two steps, each of which is formulated as an independent lemma (see Section \ref{sec:lemmasB1toB4}).
	\begin{itemize}
		\item[]\underline{\bf Step 1:} Splitting the integration (Lemma \ref{lem:coarea application}). \\We first note that integration over $A_h(y)$ can be split into integrating first over the surface $(\dtm)^{-1}(v)\cap \X$ with respect to the $(d-1)$-dimensional Hausdorff measure $\Haus{d-1}$ (see \citet{federer2014geometric, morgan2016geometric} for an introduction) and afterwards over $v\in[y-h,y+h]$.
		\item[]\underline{\bf Step 2:} Local Lipschitz continuity (Lemma \ref{lem:local Lipschitz}).\\
		We prove that the integral of a bounded, $\alpha$-Lipschitz function $g:\X\subset\R^d\to[-\alpha,\alpha]$ over the level set of a $C^{2,1}$-regular bounded value $\dfct$ with respect to $\X$, denoted as $y$, is locally Lipschitz continuous in $y$. More precisely, we prove that there exists $h_0>0$ such that for all $-h_0<v<h_0$ it holds that
		\[\left|\int_{\Gamma_{y+v}\cap\X}\!g(x)\,d\Haus{d-1}(x)-\int_{\Gamma_{y}\cap\X}\!g(x)\,d\Haus{d-1}(x)\right|\leq C_y|v|,\]
		where $C_y>0$ denotes a constant that only depends on $\dfct,y,\alpha$ and $\X$.
	\end{itemize}
	
	\subsection{Auxiliary Lemmas Representing Step 1 and Step 2}
	\label{sec:lemmasB1toB4}
	\begin{lemma}\label{lem:coarea application}
		Let $f:\R^d\to\R$ be a Lipschitz continuous function. Let $h>0$, $\X\subset\R^d$ a compact space and $g:\R^d\to \R$ such that the function
		\begin{equation}\label{eq:coarea condition}
			x\mapsto\frac{|g(x)|}{||\nabla f(x)||}\indifunc{x\in\X\,:\,|f(x)|\leq h}
		\end{equation}
		is integrable with respect to $\lambda^d$. Then, it follows that
		\[\int_{\{x\in \X\,:\, |f(x)|\leq h\}}g(x)\,d\lambda^d(x)=\int_{-h}^h\!\int_{f^{-1}(v)\cap\X}\!\frac{g(x)}{||\nabla f(x)||}\,d\Haus{d-1}(x)\,dv,\]
		where $\Haus{d-1}$ denotes the $(d-1)$-dimensional Hausdorff measure. 
	\end{lemma}
	\begin{proof}
		First of all, we observe that
		\begin{align*}
			\int_{\{x\in \X\,:\, |f(x)|\leq h\}}g(x)\,d\lambda^d(x)=\int_{\R^d}\frac{g(x)}{||\nabla f(x)||}\indifunc{x\in \X\,:\, |f(x)|\leq h}||\nabla f(x)||\,d\lambda^d(x).
		\end{align*}
		Since the function defined in \eqref{eq:coarea condition} is integrable, it follows by the \textit{Co-Area Formula} (see \citet[Thm. 3.2.12]{federer2014geometric}, where the \textit{k-dimensional Jacobian} of $f$ is $||\nabla f||$ in this setting) that 
		\begin{align*}
			&\int_{\R^d}\frac{g(x)}{||\nabla f(x)||}\indifunc{x\in \X\,:\, |f(x)|\leq h}||\nabla f(x)||\,d\lambda^d(x)\\=&\int_{-\infty}^\infty\!\int_{f^{-1}(v)}\!\frac{g(x)}{||\nabla f(x)||}\indifunc{x\in\R^d\,:\,-h\leq f(x)\leq h}\indifunc{x\in\X}\,d\Haus{d-1}(x)\,dv\\=&
			\int_{-h}^h\!\int_{f^{-1}(v)\cap\X}\!\frac{g(x)}{||\nabla f(x)||}\,d\Haus{d-1}(x)\,dv.
		\end{align*}
		This yields the claim.
	\end{proof}
	\begin{lemma}\label{lem:local Lipschitz}
		Let $\X\subset\R^d$ be a compact set. Let $g:\X\to [-\alpha,\alpha]$ be an $\alpha$-Lipschitz function for some $\alpha>0$. Let $\dfct\colon \RR^d \to \RR$ be a coercive function and $y\in \RR$ a $C^{2,1}$-\emph{regular bounded
			value} of $\dfct$ with respect to $\X$. Let $\Gamma_y
		= \dfct^{-1}\{y\}$. Then, it holds that 
		\begin{equation}\label{eq:local Lipschitz goal}
			\left|\int_{\Gamma_{y+v}\cap\X}\!g(x)\,d\Haus{d-1}(x)-\int_{\Gamma_{y}\cap\X}\!g(x)\,d\Haus{d-1}(x)\right|\leq C_y |v|,
		\end{equation}
		for some $C_y > 0$ and any $-h_0<v < h_0$, where $C_y$ and $h_0>0$ only depend on $\dfct$,
		$y$, $\alpha$ and $\X$ (and not on $g$ explicitly).
	\end{lemma}
	\begin{proof}
		Before we prove \eqref{eq:local Lipschitz goal}, we ensure that the statement is well defined and prove that under the assumptions made
		\[\int_{\Gamma_{y}\cap\X}\!|g(x)|\,d\Haus{d-1}(x)\leq\alpha\int_{\Gamma_{y}\cap\X}\!\,d\Haus{d-1}(x)<\infty.\]
		To this end, we observe that $\Haus{d-1}(\dfct^{-1}(\{y\}\cap \X)\leq\Haus{d-1}(\dfct^{-1}(\{y\}))$. As $\dfct$ is coercive it follows that the set $\dfct^{-1}([0,y])$ is bounded.  
		Hence, the same holds true for $\dfct^{-1}(\{y\})$. Furthermore, as $\dfct$ is $C^{2,1}$ in an open neighborhood of the level set $\Gamma_y$ and $\nabla \dfct\neq 0$ on $\Gamma_y$, it follows that $\Gamma_y$ is a compact  $C^1$-manifold of dimension $d-1$ \cite[Thm. 9]{villanacci2002differential}, which obviously has finite volume (and hence finite $(d-1)$-dimensional Hausdorff measure \citep{federer2014geometric,morgan2016geometric}).\\ 
		
		Now, we focus on proving the statement \eqref{eq:local Lipschitz goal}. By assumption, $\dfct$ is $C^{2,1}$ on an open neighborhood $U$ of $\Gamma_y$ with $||\nabla\dfct||>0$ on $\Gamma_y$. In consequence, there exists $h_0'>0$ such that $\dfct^{-1}([y-h_0',y+h_0'])\subset U$ and $||\nabla\dfct||>0$ on $\dfct^{-1}([y-h_0',y+h_0'])$. This means that the function
		\[\varphi(u):\R^d\to\R^d, ~u \mapsto \frac{\nabla\dfct(u)}{||\nabla\dfct(u)||^2}\]
		is $C^{1,1}(\dfct^{-1}((y-h_0',y+h_0')), \R^d)$. By Theorem A.6 in \citet{eldering2013normally}  (or more generally by Cauchy-Lipschitz's theory \citep{hirsch1974differential, amann2011ordinary}) there exists $0<h_0\leq h_0'$ such that one can construct a flow $\Phi:[-h_0,h_0]\times W\to \R^d$ with 
		\[\begin{cases}
			\frac{\partial}{\partial t}\Phi(t,x)=\frac{\nabla\dfct(\Phi(t,x))}{||\nabla\dfct(\Phi(t,x))||^2}\\
			\Phi(0,x)=x,
		\end{cases}\]
		where $W\subset \R^d$ is an open set that contains $\dfct^{-1}([y-h_0,y+h_0])$. Differentiating the function $t\mapsto \dfct(\Phi(t,x))$ immediately shows that $\dfct\left(\Phi(t,x)\right)=\dfct(x)+t$. This implies that $\Phi(t,\dfct^{-1}(\{y\}))=\dfct^{-1}(\{y+t\})$. In particular, \citet[Thm. A.6]{eldering2013normally} yields that $\Phi$ is in $C^{1,1}$. Consequently, we find that
		\begin{align*}
			&\left|\int_{\Gamma_{y+v}\cap\X}\!g(x)\,d\Haus{d-1}(x)-\int_{\Gamma_{y}\cap\X}\!g(x)\,d\Haus{d-1}(x)\right|\\=&\left|\int_{\dfct^{-1}(\{y\})}g(x)\indifunc{x\in\X}\!\,d\Haus{d-1}(x)-\int_{\Phi(v,\dfct^{-1}(\{y\}))}\!g(x)\indifunc{x\in\X}\,d\Haus{d-1}(x)\right|\\
			\leq&\int_{\dfct^{-1}(\{y\})}\! \left|g(\Phi(0,x))J_{\Phi(0,\cdot)}(x)\indifunc{\Phi(0,x)\in\X}-g(\Phi(v,x))J_{\Phi(v,\cdot)}(x)\indifunc{\Phi(v,x)\in\X}\right|\,d\Haus{d-1}(x),
		\end{align*}
		where $J_{\Phi(v,\cdot)}$ denotes the \textit{Jacobian determinant} of $\Phi(v,\cdot)$. The last line follows by a change of variables (see e.g.\ \citet[Thm. 56]{merigot2021optimal}) and the fact that $\Phi(0,\cdot)$ is the identity. By Kirszbraun's Theorem \citep[Thm. 2.10.43]{federer2014geometric} we can extend $g:\X\to[-\alpha,\alpha]$ to a Lipschitz continuous function $\tilde{g}:\R^d\to\R$, that has the same Lipschitz constant $\alpha$. Obviously, it holds that
		\begin{align*}
			&\int_{\dfct^{-1}(\{y\})}\! \left|g(\Phi(0,x))J_{\Phi(0,\cdot)}(x)\indifunc{\Phi(0,x)\in\X}-g(\Phi(v,x))J_{\Phi(v,\cdot)}(x)\indifunc{\Phi(v,x)\in\X}\right|\,d\Haus{d-1}(x)\\=&\int_{\dfct^{-1}(\{y\})}\! \left|\tilde{g}(\Phi(0,x))J_{\Phi(0,\cdot)}(x)\indifunc{\Phi(0,x)\in\X}-\tilde{g}(\Phi(v,x))J_{\Phi(v,\cdot)}(x)\indifunc{\Phi(v,x)\in\X}\right|\,d\Haus{d-1}(x).
		\end{align*}
		Therefore, we find that 
		\begin{align*}
			&\left|\int_{\Gamma_{y+v}\cap\X}\!g(x)\,d\Haus{d-1}(x)-\int_{\Gamma_{y}\cap\X}\!g(x)\,d\Haus{d-1}(x)\right|\\\leq &\int_{\dfct^{-1}(\{y\})}\! \left|\tilde{g}(\Phi(0,x))J_{\Phi(0,\cdot)}(x)-\tilde{g}(\Phi(v,x))J_{\Phi(v,\cdot)}(x)\right|\indifunc{\Phi(0,x)\in\X}\,d\Haus{d-1}(x)\\+&\int_{\dfct^{-1}(\{y\})}\!\left|\indifunc{\Phi(0,x)\in\X}-\indifunc{\Phi(v,x)\in\X}\right|\left|\tilde{g}(\Phi(v,x))J_{\Phi(v,\cdot)}(x)\right|\,d\Haus{d-1}(x).
		\end{align*}
		Since $\Phi$ is in $C^{1,1}([-h_0,h_0]\times W)$, it follows that $(v,x)\mapsto \tilde{g}(\Phi(v,x))$ and $(v,x)\mapsto J_{\Phi(v,\cdot)}(x)$ are Lipschitz continuous functions. We observe that for $(v,x)\in[-h_0,h_0]\times \X$ \begin{align*}
			\left|\tilde{g}(\Phi(v,x))\right|\leq \left|\tilde{g}(\Phi(0,x))\right|+\left|\tilde{g}(\Phi(v,x))-\tilde{g}(\Phi(0,x))\right|\leq \alpha+\alpha||\Phi(v,x)-\Phi(0,x) ||\\\leq \alpha+\alpha L_\Phi h_0,
		\end{align*}
		where $L_\Phi$ denotes the Lipschitz constant of $\Phi$. This implies immediately that the function $(v,x)\mapsto \tilde{g}(\Phi(v,x))J_{\Phi(v,\cdot)}(x)$ is Lipschitz continuous on $[-h_0,h_0]\times \X$ with a Lipschitz constant that only depends on $\dfct,y,\alpha$ and $\X$. Further, we realize that
		\begin{equation}\label{eq: diff of indis pos}
			\left|\indifunc{\Phi(0,x)\in\X}-\indifunc{\Phi(v,x)\in\X}\right|>0
		\end{equation}
		implies that either $x\in\X$ or $\Phi(v,x)\in\X$ (but not both). Given \eqref{eq: diff of indis pos}, our previous calculations show that
		\[\left|\tilde{g}(\Phi(v,x))J_{\Phi(v,\cdot)}(x)\right|\leq C_y,\]
		where $C_y$ denotes a finite constant that depends only on $\dfct,y,\alpha$ as well as $\X$. For the remainder of this proof, this constant may vary from line to line. We obtain that
		\begin{align*}
			&\left|\int_{\Gamma_{y+v}\cap\X}\!g(x)\,d\Haus{d-1}(x)-\int_{\Gamma_{y}\cap\X}\!g(x)\,d\Haus{d-1}(x)\right|\\
			\leq &\int_{\dfct^{-1}(\{y\})}\! C_y\left|v\right|\indifunc{x\in\X}\,d\Haus{d-1}(x)+C_y\int_{\dfct^{-1}(\{y\})}\!\left|\indifunc{\Phi(0,x)\in\X}-\indifunc{\Phi(v,x)\in\X}\right|\,d\Haus{d-1}(x).
		\end{align*}
		Since $y$ is a $C^{2,1}$-regular value of $\dfct$ with respect to $\X$, we find that (by potentially adjusting $h_0$) there exists $h_0>0$ such that for all $-h_0<v<h_0$ 
		\begin{align*}
			\left|\int_{\Gamma_{y+v}\cap\X}\!g(x)\,d\Haus{d-1}(x)-\int_{\Gamma_{y}\cap\X}\!g(x)\,d\Haus{d-1}(x)\right|\leq C_y|v|\Haus{d-1}\left(\Gamma_y\right)+C_y|v|\leq C_y|v|.
		\end{align*}
		This gives the claim.
	\end{proof}
	
	\subsection{Proof of Theorem \ref{thm:h2bound}}
	By assumption, we have that $\nabla\dfct \neq 0$ on the level set $\Gamma_y$. Furthermore, we have assumed that the function $\dfct$ is $C^{2,1}$ on an open neighborhood of $\Gamma_y$. Thus, there exists $h_0>0$ such that $||\nabla\dfct|| > 0 $ on \begin{equation}
		\dfct^{-1}[y-h_0,y+h_0]=\{x\in\R^d\,:\,y-h_0\leq \dfct(x)\leq y+h_0 \}.
	\end{equation} Throughout the following let $0<h< h_0$. As $\supp(k)=[-1,1]$, we get that
	\begin{align*}
		\left|\int_\X k\left(\frac{\dfct(x) - y}{h}\right) g(x)\,d\lambda^d(x)\right|
		=&\left|\int_{\{x\in\X\,:\,| \dfct(x)-y|\leq h\}} k\left(\frac{\dfct(x) - y}{h}\right) g(x)\,d\lambda^d(x)\right|.
	\end{align*}
	Since $||\nabla\dfct(x)|| > 0 $ for $x\in \dfct^{-1}[y-h_0,y+h_0]$ and $|g(x)|\leq \alpha$ for all $x$, we obtain that 
	\begin{equation}\label{eq:coarea assumption verification}
		\sup_{x\in\R^d}\left|\frac{g(x)}{||\nabla\dfct(x)||}\indifunc{x\in\X\,:\,|\dfct(x)-y|\leq h}\right|<C_y,
	\end{equation}
	where $C_y$ denotes a constant that only depends on $\dfct,y,\alpha$ and $\X$ (in particular it can be chosen independently from $h$). In the following, $C_y$ may vary from line to line. Clearly, \eqref{eq:coarea assumption verification} implies that the function 
	\[x\mapsto\frac{|g(x)|}{||\nabla\dfct(x)||}\indifunc{x\in\X\,:\,|\dfct(x)-y|\leq h}\]
	is $\lambda^d$-integrable for any $0\leq h\leq h_0$. Therefore, it follows by Lemma \ref{lem:coarea application} in combination with \eqref{eq:coarea assumption verification} that
	\begin{align*}
		\left|\int_\X k\left(\frac{\dfct(x) - y}{h}\right) g(x)d\lambda^d(x)\right|\!=&\left| \int_{-h}^h\!\int_{\{x\in\X\,:\,\dfct(x)-y=v\}}\!\!k\left(\frac{\dfct(x) - y}{h}\right)\!\frac{g(x)}{||\nabla\dfct(x)||}d\Haus{d-1}(x)dv\right|\\
		=&\left| \int_{-h}^h k\left(\frac{v}{h}\right)\!\int_{\{x\in\X\,:\,\dfct(x)-y=v\}}\!\frac{g(x)}{||\nabla\dfct(x)||}d\Haus{d-1}(x)\,dv\right|.
	\end{align*}
	We note that 
	\[\{x\in\X\,:\,\dfct(x)-y=v\}=\dfct^{-1}(y+v)\cap\X =\Gamma_{y+v}\cap\X.\]
	This yields that
	\begin{align*}
		&\left|\int_\X k\left(\frac{\dfct(x) - y}{h}\right) g(x)\,d\lambda^d(x)\right|\nonumber\\ \leq& 
		\left| \int_{-h}^h k\left(\frac{v}{h}\right)\int_{\Gamma_{y}\cap\X}\!\frac{g(x)}{||\nabla\dfct(x)||}\,d\Haus{d-1}(x)\,dv\right|\nonumber\\+&\left|\int_{-h}^h k\left(\frac{v}{h}\right)\left(\int_{\Gamma_{y+v}\cap\X}\!\frac{g(x)}{||\nabla\dfct(x)||}\,d\Haus{d-1}(x)-\int_{\Gamma_{y}\cap\X}\!\frac{g(x)}{||\nabla\dfct(x)||}\,d\Haus{d-1}(x)\!\!\right)dv\right|\eqqcolon T_5+T_6.
	\end{align*}
	Next, we consider both summands separately. First of all, we observe that the integral \[\int_{\Gamma_{y}\cap\X}\!\frac{g(x)}{||\nabla\dfct(x)||}\,d\Haus{d-1}(x)\] does not depend on $v$. Consequently, we obtain that
	\[T_5 \overset{(i)}{\leq} C_y\left| \int_{-h}^h k\left(\frac{v}{h}\right)\,dv\right|\left|\int_{\Gamma_{y}\cap\X}\!\,d\Haus{d-1}(x)\right|\overset{(ii)}{\leq} C_y\left| \int_{-h}^h k\left(\frac{v}{h}\right)\,dv \right|. \]
	Here, $(i)$ follows by \eqref{eq:coarea assumption verification} and $(ii)$ follows since $\Haus{d-1}(\Gamma_{y}\cap\X)\leq C<\infty$ for some constant $C$, as already argued in the proof of of Lemma \ref{lem:local Lipschitz}. Setting $u=v/h$ and using $\int\!k(u)\,du=0$ gives that
	\[T_5\leq C_y h \left|\int_{-1}^1\! k\left(u\right)\,du\right|=0. \]
	Hence, it only remains to consider the second summand $T_6$. Let $\X^*=\X\cap\dfct^{-1}([y-h_0,y+h_0])$. Since $h\leq h_0$, we obtain that
	\[T_6\leq \int_{-h}^h\left| k\left(\frac{v}{h}\right)\right|\left|\int_{\Gamma_{y+v}\cap\X^*}\!\frac{g(x)}{||\nabla\dfct(x)||}\,d\Haus{d-1}(x)-\int_{\Gamma_{y}\cap\X^*}\!\frac{g(x)}{||\nabla\dfct(x)||}\,d\Haus{d-1}(x)\right|dv\]
	We realize that the function 
	\[g^*:\X^*\to\R,~~x\mapsto\frac{g(x)}{||\nabla\dfct(x)||}\]
	is Lipschitz continuous, as $||\nabla\dfct(x)||>0$ for $x\in \dfct^{-1}([y-h_0,y+h_0])$, the function $||\nabla\dfct(x)||$ is in $C^{1,1}(\dfct(^{-1}(y-h_0,y+h_0))$ and $g$ is Lipschitz continuous and bounded by assumption. As $y$ is a $C^{2,1}$-regular bounded value of $\dfct$ with respect to $\X$, it is straightforward to verify that it is also one with respect to $\X^*$.  Thus, the requirements of Lemma \ref{lem:local Lipschitz} are met. By potentially decreasing $h_0$, we find for all $h$ small enough that 
	\begin{align*}
		T_6
		\leq C_y\int_{-h}^h\! \left|k\left(\frac{v}{h}\right)\right|v\,dv.
	\end{align*}
	Setting $u=v/h$ gives that
	\begin{align*}
		T_6
		\leq C_yh^2\int_{-1}^1\! \left|k\left(u\right)\right|u\,du\leq c_yh^2,
	\end{align*}
	where $c_y>0$ depends only on $\dfct,y,\alpha$ and $\X$.\\
	
	All in all, this gives
	\[\left|\int_\X k\left(\frac{\dfct(x) - y}{h}\right) g(z)\,d\lambda^d(z)\right|\leq T_5+T_6\leq c_yh^2 ,\]
	which yields the claim.
	\QEDA

	\vskip 0.2in	

\end{document}